\documentclass[12pt]{article}
\usepackage{setspace}
\usepackage{amsmath,amssymb,amsthm,amsfonts}
\usepackage[pdftex,xdvi]{graphicx}
\usepackage{caption}
\usepackage{epstopdf}
\usepackage{rotating}
\usepackage[colorlinks, citecolor=blue, urlcolor=blue, pdfborder={0 0 1}, citebordercolor={1 1 1}, pdfencoding=auto, psdextra]{hyperref}

\usepackage{enumitem}
\usepackage{geometry}
\usepackage{booktabs,longtable,multirow,colortbl,threeparttable,threeparttablex,makecell,array}
\usepackage{wrapfig}
\usepackage{float,morefloats}
\usepackage{pdflscape}
\usepackage{tabu}
\usepackage[normalem]{ulem}
\usepackage{xcolor}
\usepackage{subcaption}
\usepackage[round]{natbib}

\usepackage{comment}

\usepackage{fancyhdr}
\usepackage{xargs}
\usepackage[disable]{todonotes}
\newcommandx{\andy}[2][1=]{\todo[inline,linecolor=red,backgroundcolor=red!25,bordercolor=red,#1]{#2}}
\newcommandx{\yuan}[2][1=]{\todo[inline,linecolor=blue,backgroundcolor=blue!25,bordercolor=blue,#1]{#2}}

\definecolor{fashionfuchsia}{rgb}{0.96, 0.0, 0.63}

\numberwithin{equation}{section}
\theoremstyle{definition}

\theoremstyle{plain}

\newtheorem{theorem}{Theorem}
\newtheorem{proposition}[theorem]{Proposition}
\newtheorem{lemma}[theorem]{Lemma}
\newtheorem{corollary}[theorem]{Corollary}
\newtheorem{remark}{Remark}

\newtheorem{assumption}{Assumption}

\newcommand{\bomega}{\mbox{\boldmath $\omega$}}

\DeclareMathOperator{\Cov}{Cov}

\DeclareMathOperator{\Var}{Var}

\DeclareMathOperator{\DNN}{DNN}

\DeclareMathOperator{\SE}{SE}
\DeclareMathOperator{\sgn}{sgn}
\DeclareMathOperator{\MV}{MV}
\DeclareMathOperator{\RS}{RS}
\DeclareMathOperator{\UA}{UA}

\DeclareMathOperator{\FSE}{FSE}

\allowdisplaybreaks

\begin{document}



\title{\large\bf{The Uncertainty of Machine Learning Predictions in Asset Pricing}\footnote{We thank Damir Filipovic, Semyon Malamud and seminar audiences at \'{E}cole Polytechnique F\'{e}d\'{e}ral de Lausanne, Washington University in St.~Louis, Wolfe Research.}}

\author{
\and Yuan Liao\\[-3mm]  \footnotesize Department of Economics \\[-3mm]  \footnotesize{Rutgers University}
\and Xinjie Ma\\[-3mm]  \footnotesize NUS Business School\\[-3mm] \footnotesize National University of Singapore
\and  Andreas Neuhierl\\[-3mm]  \footnotesize  Olin School of Business\\[-3mm] \footnotesize Washington University in St. Louis
\and  Linda Schilling \\[-3mm]  \footnotesize  Olin School of Business\\[-3mm] \footnotesize Washington University in St. Louis
}

\newcommand{\blind}{0}

\ifnum\blind=1
    \author{} 
\else

\author{
\and Yuan Liao\\[-3mm]  \footnotesize Department of Economics \\[-3mm]  \footnotesize{Rutgers University}
\and Xinjie Ma\\[-3mm]  \footnotesize NUS Business School\\[-3mm] \footnotesize National University of Singapore
\and  Andreas Neuhierl\\[-3mm]  \footnotesize  Olin School of Business\\[-3mm] \footnotesize Washington University in St. Louis
\and  Linda Schilling \\[-3mm]  \footnotesize  Olin School of Business\\[-3mm] \footnotesize Washington University in St. Louis
}

\fi

\maketitle
\begin{abstract}

Machine learning in asset pricing typically predicts expected returns as point estimates, ignoring uncertainty.  We develop new methods to construct forecast confidence intervals for expected returns obtained from neural networks. We show that neural network  forecasts of expected returns share the same asymptotic distribution as classic nonparametric methods,  enabling a closed-form expression for their standard errors.  We also propose a computationally feasible bootstrap to obtain the asymptotic distribution. We incorporate these forecast confidence intervals into an uncertainty-averse investment framework. This provides an economic rationale for shrinkage implementations of portfolio selection.  Empirically, our methods improve out-of-sample performance.
\end{abstract}


\newpage
\onehalfspacing
\clearpage
\setcounter{page}{1}

\listoftodos 

\newpage
\onehalfspacing
\clearpage
\setcounter{page}{1}

\section{Introduction} \label{sec:intro}
Pension funds, wealth advisors, and hedge funds worldwide employ predictive methods to forecast asset returns and construct optimal portfolios that aim to maximize client returns. Recently, machine learning (ML) models have gained prominence in predicting asset returns, selecting portfolios, and estimating stochastic discount factors, with significant success in these areas. ML techniques, by capturing complex and nonlinear relationships in financial data, are particularly well-suited for enhancing portfolio management decisions. For example, within the mean-variance portfolio framework, ML methods are increasingly used to estimate expected returns and (co)variances, often leading to more effective portfolio allocations. The literature consistently demonstrates the effectiveness of machine learning in these and other applications (e.g., \citet*{gu2020empirical,bianchi2021bond,cong2021alphaportfolio, kelly2021virtue,patton2022risk,didisheim2023complexity,filipovic2024joint}).

Despite the success of machine learning in asset pricing, existing literature typically treats ML predictions as point estimates and conducts asset pricing analyses as if they were true values, overlooking the associated uncertainty. This is surprising, given that uncertainty about input parameters is widely acknowledged as critical in portfolio selection (e.g., \citet*{demiguel2009optimal}), and \citet*{garlappi2007portfolio} show that incorporating forecast uncertainty in mean-variance portfolio allocation leads to distinct economic insights. However, quantifying prediction uncertainty in ML forecasts, particularly with neural networks, remains a complex challenge, limiting their broader application in asset pricing. This paper addresses this gap by rigorously quantifying uncertainty in ML predictions and incorporating it into portfolio selection, specifically by constructing forecast confidence intervals (FCIs) for return predictions from neural networks.

We provide two methods for constructing forecast confidence intervals: one based on closed-form approximations and the other on the bootstrap. A key theoretical contribution of this paper is our proof that ML-based forecast methods exhibit an asymptotic distribution independent of the specific ML  model  used to generate the forecast. For example, the asymptotic distribution of neural network forecasts is not specific to neural networks.  This implies that, under appropriate technical conditions, simpler machine learning methods—such as Fourier series, particularly those with closed-form estimators—can be utilized to construct confidence intervals. These intervals can be applied to algorithmically more complex machine learning estimators, whose confidence intervals would otherwise be intractable.  Building on this insight, we derive an analytic formula for the  ML forecast standard error,  which is straightforward to calculate and paves the way for constructing the FCI.    

For the second method, we propose a novel $k$-step bootstrap approach to simulating the asymptotic distribution of the ML forecast. This method overcomes the substantial computational burden associated with conventional bootstrap procedures, which require repeatedly fully training multiple neural networks from scratch. In contrast, the $k$-step bootstrap, originally developed by \cite{davidson1999bootstrap} and \cite{andrews2002higher}, starts with the previously trained neural network and retrains it for additional $k$ steps on a bootstrap resample.  Through the $k$-step method, a pre-trained neural network can achieve a high level of training accuracy as quantified by the validity of the FCIs' coverages, even if $k$ is relatively small. 
\andy{revisit the sentence above}
 
We verify that both methods give correct asymptotic coverage probabilities for ML forecasts of expected returns, while effectively addressing cross-sectional dependence. Empirically, the two methods produce qualitatively similar results. We recommend that researchers compare the two as robustness check. In some applications, it may also be prudent to follow a conservative approach and use the larger of the two standard errors. Furthermore, through extensive simulations, we demonstrate that alternative bootstrap approaches, such as bootstrapping across assets or jointly across assets and time series, fail to produce a valid forecast confidence interval for machine learning forecasts.

In the second part of the paper, we apply our methods for quantifying estimation uncertainty to two standard problems in portfolio selection. We adopt the view of an uncertainty-averse (UA) investor.  In the first application, we build on the framework introduced by \cite{garlappi2007portfolio}: we treat the expected return in the classic mean-variance optimization problem as an unknown parameter varying within a confidence interval. Then, we find the optimal portfolio weights for the worst-case scenario of the mean-variance utility, varying the mean within the given interval. Extending the characterization by \cite{garlappi2007portfolio} to the machine learning context, we apply the proposed neural network FCI to forecast the expected return of the UA-portfolio. We  show that the solution can be formulated as an $\ell_1$-penalized regression problem.  Our results establish that portfolio weights exhibit a ``non-participation" region for risky assets, meaning an uncertainty-averse investor may choose not to invest in a risky asset if uncertainty about the asset's expected value exceeds a certain level.   In addition, this region expands as the investor's uncertainty about the expected value increases.  In contrast, the non-participation region is not present in the standard mean-variance approach which treats the expected returns as known. This insight offers a clear rationale for employing shrinkage approaches in portfolio selection, as developed by \citet*{ao2019approaching} and \citet*{kozak2020shrinking}. Empirically, the behavior of UA-portfolio of individual stocks aligns closely with our theoretical characterization, and  generates higher Sharpe ratios than benchmarks that disregard forecast uncertainty. 

In our second application, we focus on selecting assets with statistically significantly positive expected returns and constructing long-only portfolios. This exercise is particularly relevant for mutual funds, which do not employ short positions and thus face exactly this challenge of building long-only portfolios.  The forecast confidence interval facilitates the selection of individual securities with significantly positive expected returns while controlling the false discovery rate, thereby mitigating the multiple testing problem as documented by  \citet*{barras2010false} and \cite{harvey2020false}. Empirically, our long-only investment strategy, which accounts for forecast uncertainty,  yields higher out-of-sample average returns and similar standard deviations compared to benchmark methods that select assets without considering forecast uncertainty. In addition, it yields economically significant alphas when regressing the returns on various risk factors. 

Our result holds in a  setting where the machine learning model is not over-parameterized. Recently, \cite{kelly2021virtue} and \cite{didisheim2023complexity} derive results about the out-of-sample properties of portfolios in the context of overparameterized models where the number of parameters far exceeds the sample size. They show the interesting virtue of complexity in asset pricing models. Developing forecast confidence intervals in the virtue of complexity regime is an important question which we leave for future research.

\subsection*{Related Literature} \label{sec:litreview}
The literature on machine learning in asset pricing has grown rapidly in recent years, such as \citet*{kelly2019characteristics, kozak2020shrinking,freyberger2020, chen2020deep, baba2022factor} and \citet*{li2022real}. Recently, several papers in asset pricing have made progress in the theoretical analysis of machine learning predictions. \citet*{fan2022structural} and \citet*{jaganannthan2023robust} propose the so-called “period-by-period” ML and developed the FCI around the forecast indices, which relies on estimating latent risk factors. Period-by-period learning requires training the machine learning model each period. In this paper, we instead focus exclusively on the popular “pooled machine learning” approach, which is the most common in forecasting applications and does not involve estimating factors.\footnote{Pooled machine learning pools the data over the cross-sectional and time series dimension in estimation.}  

\cite{allena2021confident} is one of the first papers in asset pricing that formally develops confidence intervals for risk premia predicted using machine learning. He proposes a Bayesian ML approach that flexibly draws from posterior distributions of the predicted risk premia and demonstrates that it successfully yields a confidence-based strategy with insightful economic interpretations. The Bayesian approach leverages the fact that posterior distributions, deduced from properly specified priors, can provide valid confidence intervals. In this paper, we take a different approach. In particular, we do not pursue Bayesian estimation. One of the two methods we develop for the forecast confidence interval is based on a derived analytic formula for the ML standard error. Meanwhile, we move to expand the insights obtained by \cite{allena2021confident} and \cite{garlappi2007portfolio}. In the context of portfolio selection under uncertainty aversion, we show that machine learning confidence intervals lead to qualitatively different investment behavior relative to the standard mean-variance model.

Another strand of literature in asset pricing studies model uncertainty (e.g., \cite{avramov2002stock,anderson2016robust,bianchi2024takes}). In this line of research, no stance is taken on the ``correct model,” as the goal is often to achieve robust portfolio allocations and predictions via Bayesian model averaging. Researchers therefore usually specify a probability distribution over the different models but do not derive the forecast uncertainty within a given model.  

In theoretical econometrics, research on machine learning uncertainty is still at a very early stage. There are, however, few papers that rigorously develop the forecast standard error in the i.i.d. setting.  In a seminal paper, \cite{chen1999improved} set the stage for a rigorous analysis of the distributional properties of neural network predictions.  In the field of econometric program evaluation, there is a popular method known as ``doubly robust ML inference", which aims to develop asymptotic confidence intervals for some structural parameters in econometric models (e.g., \citet*{chernozhukov2018double}). These methods develop sophisticated procedures that require so-called orthogonal moment conditions and cross-fitting, which are quite different from the typical implementation of ML models in asset pricing. In addition, both approaches rely on the assumption of i.i.d. (or weakly dependent) data, which is invalid in asset pricing due to the strong cross-sectional dependence driven by common risk factors. This dependence is at the heart of the theoretical challenge in asset pricing. This paper proposes new approaches that explicitly account for this dependence and develops forecast confidence intervals with frequentist guarantees.

\section{Machine Learning Forecast Confidence Intervals} \label{sec:background}
\subsection{The Model} \label{sec:model}

Consider the excess return of a portfolio: 
$$
z_{t+1}= \sum_{i=1}^N w_i y_{i,t+1},
$$
where $y_{i,t+1}, i=1,...,N$ denote the excess return for base asset $i$ from $t$ to $t+1$, relative to the risk-free rate. The portfolio weights $w_i$ are assumed to be known at time $t$, but they may vary over time. This contains the special case of an individual asset, i.e., $z_{t+1}=y_{1, t+1}$, or a broad market index. At period $T$, the objective is to forecast the expected excess return of the portfolio:
$$
z_{T+1|T} =\mathbb E(z_{T+1}|\mathcal F_{T})
$$
where $\mathcal F_T$ denotes the information set up to time $T$. To do so,  researchers observe a matrix of asset-specific characteristics (features), $x_{t-1}=(x_{1,t-1},...,x_{N,t-1}),$ where $x_{i,t-1}$ is an $d$-dimensional vector of characteristics for asset $i$ at time $t-1$, such as momentum, volatility, financial liabilities.  In this setting, machine learning regression, e.g. \cite{gu2020empirical,bianchi2021bond} have become a very successful and popular methodology to obtain point predictions. ML regressions build on the nonparametric model,
$$
y_{i,t} = g(x_{i,t-1}) + e_{i,t}
$$
with an unknown function $g(\cdot): \mathbb R^d\to\mathbb R$, where $e_{i,t}$ is the error term. The unknown function is learned by {pooling} all observed data (cross-sectionally and over time) and solving a least squares problem:

\begin{equation}\label{eq3.1}
\widehat g(\cdot) = \arg\min_{g\in\mathcal G } \sum_{i=1}^N\sum_{t=1}^T(y_{i,t}-g(x_{i,t-1})  )^2,
\end{equation}
where the optimal solution is searched for in a function space $\mathcal G$ that typically corresponds to a specific machine learning method.  Once $\widehat g(\cdot)$ has been computed, the expected excess return is predicted by plugging in the most recent characteristic $x_{i,T}$ and constructing a portfolio: 
 \begin{equation}\label{eq3.2}
 \widehat z_{T+1|T} := \sum_{i=1}^N w_i  \widehat g(x_{i,T}).
 \end{equation}
This approach is widely used in both academic research and industry applications, making it the primary forecasting method analyzed in this paper.  

The choice of the function space $\mathcal G$ corresponds to the specific ML method to forecast excess returns. For example, neural networks, random forests, boosted regression trees, or random feature regressions give rise to different $\mathcal G$'s.  For the purpose of deriving confidence intervals, we focus on two types of machine learning methods,
$$
\mathcal G =  \text{either } \mathcal G_{\DNN} \text{ or } \mathcal G_B. 
$$
Here $\mathcal G_{\DNN} $ corresponds to the use of deep neural network (DNN)  functions, which is a collection of all possible neural network functions with a predetermined architecture — specifying the width and depth of the layers and the activation functions for each neuron. Then by minimizing (\ref{eq3.1}), we find the optimal neuron biases and weights so that the neural network function $\widehat g(\cdot)$  optimally fits the in-sample data.  

Alternatively, the more classic nonparametric regression gives rise to an alternative specification $\mathcal G=\mathcal G_B$, which uses a set of basis functions: $\Phi(x)= (\phi_1(x),...,\phi_J(x))$ such as the Fourier basis.
Then, (\ref{eq3.1}) searches for the optimal function within a simpler space: 
$$
\mathcal G_B= \{\Phi(x)'\theta: \theta\in\mathbb R^J\},
$$
and the function $\widehat g(\cdot)$ is estimated by finding the best combination of the basis functions. 

The crucial distinction between $\mathcal G_{\DNN}$ and $\mathcal G_{B}$ is that the estimators in $\mathcal G_B$ admit a closed-form representation, whereas $\mathcal G_{\DNN}$ does not. However, $\mathcal G_{\DNN}$ is more appealing in predicting expected returns in asset pricing leveraging the advantages of sophisticated machine learning methods, as documented in \cite{gu2020empirical} and \cite{bianchi2021bond}.  The objective of this paper is to construct forecast confidence intervals for predictions made with neural networks $\mathcal G_{\DNN}$, but it will soon become clear that $\mathcal G_{B}$ also plays an important role in our construction. Throughout, we refer to $\mathcal G_{B}$ as the ``closed-form ML", such as Fourier series regressions, B-Splines -- which are more classic nonparametric regression models. 

Despite the great popularity and empirical success of machine learning predictions, little work has been devoted to understanding the structure and sources of predictability. However, understanding the structure of the prediction, $\widehat z_{T+1|T} $, is crucially important to quantifying the prediction uncertainty. Recently, \cite{fan2022structural} provide a thorough analysis by bridging between the machine learning model and factor models. They suppose that excess returns  follow a conditional factor model (also see \citet*{gagliardini2016time, zaffaroni2019factor}):
$$
y_{i,t} = \alpha_{i,t-1} + \beta_{i,t-1}'f_t+ u_{i,t}
$$
where $\alpha_{i,t-1}$ and $\beta_{i,t-1}$ are respectively the ``alpha" and ``beta" of the asset, $f_t$ is the set of (possibly latent) risk factors, and $u_{i,t}$ is the idiosyncratic return. In addition, suppose characteristics are informative about factor loadings (betas) and mispricing (alpha), i.e., there are functions $g_{\alpha}$ and $g_{\beta}$ so that we can rewrite alpha and beta as:\footnote{This formulation is formalizing the notion that asset characteristics are informative about risk exposures which is also documented in \cite{rosenberg1973prediction,jagannathan1996conditional,connor2012efficient,gagliardini2016time,kelly2019characteristics}.}
$$
\alpha_{i,t-1} = g_{\alpha}(x_{i,t-1}) ,\quad \beta_{i,t-1} = g_{\beta}(x_{i,t-1}).
$$
 
We can thus rewrite the asset pricing model as

\begin{equation}\label{eq3.3}
y_{i,t} = g_{\alpha}(x_{i,t-1})+ g_{\beta}(x_{i,t-1})'\mathbb Ef_t +  g_{\beta}(x_{i,t-1})'v_t+ u_{i,t},
\end{equation}
where $v_t:= f_t- \mathbb E f_t$. Both $v_t$ and $u_{i,t}$ are mean-zero processes contributing to the error term:

\begin{equation}\label{eq3.4}
e_{i,t}:= g_{\beta}(x_{i,t-1})'v_t+ u_{i,t}, 
\end{equation}
then the first term is the exposure to factor shocks,  while the second term is the idiosyncratic shock. Under the assumption of constant prices of risk, i.e. $\lambda_f:=\mathbb Ef_t$ does not change over time, we can define 

\begin{equation}\label{eq3.5}
g(x):= g_{\alpha}(x)+ g_{\beta}(x)' \lambda_f.
\end{equation}

Then indeed, (\ref{eq3.3}) can be formulated as the ML model $y_{i,t} =g(x_{i,t-1}) + e_{i,t} $, with $g(\cdot)$ and $e_{i,t}$ defined in (\ref{eq3.4}) and (\ref{eq3.5}) respectively. Therefore, by applying  ML regressions within the context of this model, \cite{fan2022structural} show that $\widehat g(x)$  is estimating $g_{\alpha}(x)+ g_{\beta}(x)'\lambda_f$. More formally, they show that the ML function has the following (probability) limit:
$$
\widehat z_{T+1|T}\to^P z_{T+1|T}=\sum_{i=1}^Nw_i[g_{\alpha}(x_{i,T})+ g_{\beta}(x_{i,T})'\lambda_f].
$$
This shows that the predictive ability of machine learning regressions arises from capturing mispricing and risk premia, i.e.~ both $\sum_{i=1}^Nw_i\alpha_{i,T}$ and $\sum_{i=1}^N w_i \beta_{i,T}'\lambda_f$ --   are key components of expected returns.

Understanding the properties of the predicted expected portfolio return, $\widehat z_{T+1|T}$, is an essential first step towards understanding machine learning predictability.  Studying the uncertainty of ML forecasts is yet a challenging problem. Most statistical results are developed for i.i.d. or weakly dependent errors, however in asset pricing this is a tenuous assumption as explained by \cite{allena2021confident}: we should expect to see strong cross-sectional dependence.  The source of the cross-sectional dependence can be illustrated through the lens of the (characteristic-based) factor model. Recall the standard ML model:

\begin{equation}\label{eq3.6}
y_{i,t} = g(x_{i,t-1}) + e_{i,t}.
\end{equation}
The errors are strongly cross-sectionally dependent.  Take two assets, $i$ and $j$, then we can characterize the covariance of their errors by using equation (\ref{eq3.4}): 
\begin{equation} \label{eq:cov_errors}
\Cov(e_{i,t},e_{j,t}) = \mathbb E\left[g_{\beta}(x_{i,t-1})'\Cov(f_t) g_{\beta}(x_{j,t-1})\right]\neq 0.    
\end{equation}

We expect to see a strong correlation because the assets are exposed to the same sources of systematic risk. Hence, the regression model (\ref{eq3.6}) is not the usual ML model with i.i.d. errors. This is why new methods are needed to explicitly consider the strong cross-sectional dependence structure.\footnote{\cite{allena2021confident} specified a novel prior to account for the strong dependences in the Bayesian framework. A possible alternative approach is to treat factors as ``interactive fixed effects" as in \cite{bai09} and explicitly estimate them. However, the method in \cite{bai09} or \cite{freyberger2018non} does not cover the case of sophisticated machine learning methods, nor is it the  ``standard" implementation in the applied forecasting literature. We will therefore not pursue this approach.} It is essential to stress that accounting for the dependence in the errors is not purely an econometric challenge. From \eqref{eq:cov_errors}, it is clear that we vastly underestimate the risk associated with a prediction if we incorrectly assume i.i.d. errors.  The Monte Carlo evidence in Section \ref{sec:simulation}  will illustrate that ignoring the cross-sectional dependences will vastly understate the prediction uncertainty.

\subsection{The Intuition of the ML Forecast Error} \label{sec:FCI}
We now discuss the intuition of our approach to computing the ML standard error. Recall that the expected excess return, $\widehat z_{T+1|T}$, is predicted using either neural networks (corresponding to $\mathcal G_{\DNN}$) or a closed-form ML method (corresponding to $\mathcal G_{B}$). The intuition of our constructed FCI is based on the following theorem, which is our first main result, showing that $\widehat z_{T+1|T}$ has the same asymptotic distribution regardless of which specific ML method is used.

\begin{theorem}\label{th1} Suppose $\sum_i|w_i|<\infty$ and Assumption \ref{ass2}, Assumption \ref{ass1} and Assumption \ref{ass1a} in the appendix hold.  
There exists a function $\zeta^*(\cdot)$ so that 
\begin{eqnarray}\label{eq3.7}
\widehat z_{T+1|T}  - z_{T+1|T} &=& \frac{1}{T}\sum_{t=1}^T\mathcal B_{t-1}(f_t- \mathbb E f_t)+ o_P(T^{-1/2}),\quad  \text{where} \cr 
\mathcal B_{t-1}&:=&\frac{1}{N}\sum_{i=1}^N  \zeta^*(x_{i,t-1}) \beta_{i,t-1}' .
\end{eqnarray}
The function $\zeta^*$ only depends on the joint distribution of $(x_{i,t-1})$, but  does not depend on whether $\mathcal G_{\DNN}$ or $\mathcal G_{B}$ are used for constructing $\widehat z_{T+1|T}$. In other words, the same asymptotic expansion holds for $\widehat z_{T+1|T}$   if the closed-form ML space $\mathcal G_{B}$ is used in place of $\mathcal G_{\DNN}$.

In addition, 
\begin{equation}\label{eqdistr}
    \sqrt{T} V^{-1/2} (\widehat z_{T+1|T}  - z_{T+1|T} )\to^d \mathcal N(0,1)
\end{equation}
where $V= \frac{1}{T}\sum_{t=1}^T\Var(\mathcal B_{t-1}f_t)$ does not depend on whether $\mathcal G_{\DNN}$ or $\mathcal G_{B}$ are used to predict. Hence $\widehat z_{T+1|T} $ has the same asymptotic distribution for both neural networks and closed-form- ML. 
\end{theorem}

\begin{proof}
    The proof is given in Section \ref{sec:p1}.
\end{proof}

The theorem has two important implications. First and foremost, fundamentally, the asymptotic distribution of ML forecasts does not depend on the specific machine learning model $\mathcal{G}$ in (\ref{eq3.1}). In the asymptotic expansion, $\mathcal B_{t-1}$  depends on a function $\zeta^*(\cdot)$. Although the closed-form expression for the function $\zeta^*(\cdot)$ is very difficult to derive for neural networks, it is entirely determined by the quantities of the asset pricing model (\ref{eq3.3}), but not by the specific choice of the ML method, which can be sophisticated machine learning prediction $\mathcal G_{\DNN}$ (neural networks) or closed-form prediction $\mathcal G_B$ (e.g., Fourier series regression). The main econometric intuition is that the predictor is obtained by minimizing a \textit{regular loss function} -- it is the loss function rather than the choice of $\mathcal{G}$ that ultimately determines the asymptotic distribution.\footnote{Note that this result also allows for small shrinkage in penalized regressions. However, if models are strongly penalized, as with the Lasso, this induces a shrinkage bias, leading to different asymptotic phenomena \cite{ZhangZhang:CI}. } 

This insight forms the foundation of our approach to constructing confidence intervals for the expected return predicted by neural networks. As explained in more detail in Section \ref{metho1}, one of the two proposed methods approximates the standard error of the neural network forecast using an analytic standard error from closed-form ML.  The latter is straightforward to derive due to its closed-form nature.

The second implication of the theorem is that the prediction error $\mathcal B_{t-1}(f_t-\mathbb Ef_t)$ is only driven by the common factor shocks rather than the idiosyncratic errors. Thus, the rate of convergence is $O_P(1/\sqrt{T})$, which is much slower than $O_P(1/\sqrt{NT})$ for the usual panel data models. This aligns with the well-established asset pricing intuition that the factor risk premium can only be learned well over time, rather than cross-sectionally, due to the strong cross-sectional dependence. Therefore, the dominant source of uncertainty comes from the time series rather than the cross-sectional variation.

Theorem \ref{th1} requires three technical assumptions, which are stated in the appendix. Assumption \ref{ass2} is a standard condition about the dependence of the data and the tail behavior. Assumptions \ref{ass1} and \ref{ass1a} concern the complexity of the machine learning space. Essentially, it says that the complexity must be controlled and the models cannot be overparametrized. This is in contrast to the setting of \cite{kelly2021virtue} and \cite{didisheim2023complexity}, who derive appealing features of overparametrized models in portfolio construction. Developing distribution theory in such a setting is extremely challenging, but interesting and will be left for future work in theoretical econometrics.

\subsection{Constructing ML Forecast Confidence Interval}\label{metho1}
Throughout the rest of the paper, we denote $\widehat z_{T+1|T}$ as the expected return predicted by neural networks, i.e. $\mathcal G_{\DNN}$. These models are among the most successful in asset pricing, as reviewed before. The objective is to construct its forecast confidence interval.  Each of the two implications outlined in the previous subsection motivates a distinct method for this purpose. In the sequel, we introduce these two methods. We then later demonstrate the usefulness of these methods in asset pricing in Section \ref{sec:application}.

\subsubsection{Method I: Closed-form ML Approximation} \label{sec:analytical}
We first introduce a closed-form ML method, building on our main result in Theorem \ref{th1}.  The asymptotic distribution (\ref{eqdistr}) derived in Theorem \ref{th1} demonstrates that the asymptotic variance of $\widehat z_{T+1|T}$ is \textit{the same} regardless of whether $ \mathcal G= \mathcal G_{\DNN}$ or $\mathcal G_{B}$. This intuition allows us to approximate the forecast standard error using that of the closed-form ML method, which is much easier to derive. 

For the closed-form ML we set $\mathcal G=\mathcal G_B$  in (\ref{eq3.1}) and use Fourier series to make predictions. In this case, let $\Phi(x)$ be the vector of Fourier bases, and  the regression model is:
\begin{eqnarray*}
\widehat g_B(\cdot) &=& \Phi(\cdot)'\widehat\theta,\quad \text{ where } \widehat\theta:=\arg\min_{\theta} \sum_{i=1}^N\sum_{t=1}^T(y_{i,t}-\Phi(x_{i,t-1})'\theta  )^2. 
\end{eqnarray*}

Estimators using the Fourier series can be obtained in closed form because of their OLS-type analytic solution:

$$
\widehat g_B(\cdot )= \Phi(\cdot)' (\Psi'\Psi)^{-1}\sum_{i,t}\Phi(x_{i,t-1})y_{i,t},
$$
where $\Psi$ is the $NT\times J$ matrix stacking all $\Phi(x_{i,t-1})$. Therefore, we can easily derive its asymptotic standard error: 

\begin{equation}\label{eq3.10se}
\text{SE}(\widehat z_{T+1}):= \sqrt{\sum_{t=1}^T H'\Phi_{t-1}'\beta_{t-1}\Cov(f_t)\beta_{t-1}  '\Phi_{t-1} H},
\end{equation}
where  $W=(w_1,...,w_N)$ are the portfolio weights,  $\Phi_{t-1}=(\Phi(x_{1,t-1}),...,\Phi(x_{N,t-1}))$, and $H ' = W' \Phi_T (\Psi'\Psi)^{-1} .$ This is straightforward to estimate by:

\begin{equation}\label{eq3.11se}
\widehat{\text{SE}}(\widehat z_{T+1|T}):=\sqrt{\sum_{t=1}^T H'\Phi_{t-1}'\widehat e_{t}\widehat e_{t}'\Phi_{t-1} H},
\end{equation}
where $\widehat e_{t}$ is the vector of residuals, $\widehat e_{i,t}:=y_{i,t}-\widehat g(x_{i,t-1})$.

It is important to note that we employ closed-form ML, $\widehat g_{B}(\cdot)$, only to compute the forecast standard error. The forecast itself, $\widehat g(\cdot)$ as in (\ref{eq3.1}),  is still generated by a sophisticated neural network. This leverages our result that the two ML predictions have the same asymptotic standard error.  It is therefore tempting to ask why we should employ more sophisticated neural networks in the first place. The answer is that the benefits of using neural networks or related methods do not arise from a smaller standard error, but from fewer constraints in handling highly nonlinear functions and the capability of approximating larger classes of asset pricing functions. 

To be specific, Theorem \ref{th1} shows that for both DNN and closed-form ML predicted $\widehat z_{T+1|T}$, 
$$
\widehat z_{T+1|T}  - z_{T+1|T} =\frac{1}{T}\sum_{t=1}^T\mathcal B_{t-1}(f_t- \mathbb E f_t)+ o_P(T^{-1/2}).
$$
The remainder error term $o_P(T^{-1/2})$ captures what is known as ``approximation bias",  which arises from approximating the unknown expected return function $g(x)$ using either machine learning method. While both methods exhibit a diminishing approximation bias, the rate of decay differs. Neural networks, due to the adaptivity to the intrinsic dimension of the input features, \footnote{See \cite{schmidt2020nonparametric,kohler2021rate}  and \cite{fan2022structural} for detailed discussions on theoretical advantages of neural networks.}  can flexibly approximate a broad class of nonlinear functions with a rapidly diminishing approximation bias.
In contrast,  the closed-form ML (e.g., Fourier series) is capable of approximating a much narrower class of functions, and its approximation bias decays at a slower rate as the number of input features increases, due to the well-known curse of dimensionality.  This makes closed-form ML unsuitable for direct use in return prediction.  However, its larger bias does not affect its asymptotic variance, which still provides a good approximation to the variance of the predicted returns from the neural network.

The following theorem is our second main result, which formally justifies the validity of $\widehat\SE(\widehat z_{T+1|T})$ as the standard error when $\widehat z_{T+1|T}$ is constructed using neural networks. 

\begin{theorem}[Closed-form ML approximation]\label{th2} Suppose    Assumption \ref{ass2}, Assumption \ref{ass1} and Assumption \ref{ass1a} in the appendix hold, then 
$$
\widehat\SE(\widehat z_{T+1|T})^{-1}(\widehat z_{T+1|T} - z_{T+1|T}) \to^d\mathcal N(0,1),
$$ 
where $\widehat{z}_{T+1}$ is obtained from neural networks and $\widehat{\SE}(\widehat{z}_{T+1|T})$ is obtained via Fourier series as in (\ref{eq3.11se}). Then or a given significance level $1-\alpha\in(0,1)$,
$$
P(|\widehat z_{T+1|T} - z_{T+1|T}|<\widehat\SE(\widehat z_{T+1|T}) \times  \epsilon_{\alpha}) \to 1-\alpha,
$$
where $\epsilon_{\alpha}$  is the critical value of the standard normal distribution corresponding to size $\alpha$. 
\end{theorem}

\begin{proof}
The proof is given in Section \ref{proofth2}.
\end{proof}

\subsubsection{Method II: The  $k$-step Bootstrap} \label{sec:bootstrap}
We propose an alternative method for constructing FCIs, based on the time-series bootstrap. The bootstrap computes the critical value by repeatedly resampling from the original data set and using the quantile of the neural network predictors recomputed from the resampled data. Unlike the method of analytic standard error, the bootstrap does not involve closed-form ML methods, and therefore requires weaker conditions.

However, the performance and validity of bootstrap depend critically on how the bootstrap data are generated.  In the asset pricing context, the bootstrap data should properly capture the primary sources of uncertainty to predicted expected returns. Theorem \ref{th1} shows that the forecast uncertainty is mainly driven by time series variation. Thus, we can cluster in time by applying the wild bootstrap to mimic the sampling distribution of $\widehat{z}_{T+1|T}$.\footnote{In the case of serial correlation, the block bootstrap (\cite{kunsch1989jackknife}) or the stationary bootstrap (\cite{politis1994stationary}) can be applied directly in our setting.} Specifically, let $\eta_t^*$ denote an i.i.d. sequence of standard normal random variables and let $\widehat{g}(\cdot)$ denote the learned function using neural networks. Define the bootstrap residuals as:
\begin{equation}
 e_{i,t} ^*= (y_{i,t}- \widehat g(x_{i,t-1})) \eta_t^* . 
\end{equation}
We then apply neural networks to the resampled excess return $y_{i,t}^* = \widehat g(x_{i,t-1}) +  e_{i,t} ^*$ and $x_{i,t-1}$ to repredict the expected return, and repeat this step many times.  The forecast critical value is given by the bootstrap quantile of the repeated predictions.   

A potential limitation of the bootstrap procedure is the enormous computational burden. For example, estimating with 100 bootstrap iterations requires training 100 separate neural networks, one for each bootstrap sample. Fully training these neural networks is computationally very costly, limiting the applicability of bootstrap-based inference for larger machine learning models. To address this issue, we propose a $k$-step bootstrap method for neural network inference,  which significantly reduces computational burden. The $k$-step bootstrap was initially proposed and studied by \cite{davidson1999bootstrap} and \cite{andrews2002higher} in the context of making inference for nonlinear models. The idea is that, instead of fully training the neural network for each bootstrap sample, we only train it iteratively for $k$ epochs, with a relatively small $k$ such as 10 or 20. This approach takes advantage of the observation that the fully trained function from the original data, $\widehat{g}(\cdot)$, provides an excellent starting point for training in the bootstrap data. Thus, for each bootstrap re-sample, we initialize with $\widehat{g}_0^*(\cdot) = \widehat{g}(\cdot)$ and then train the network for $k$ epochs.

The full algorithm is given as follows:

\noindent\textbf{$k$-step Bootstrap Algorithm.} 
\begin{description}
\item[Step 1.] Generate $\eta_t^*\sim \mathcal N(0,1)$ independently;  generate 
\begin{eqnarray*}
 e_{i,t} ^*&=& (y_{i,t}- \widehat g(x_{i,t-1})) \eta_t^* \cr
 y_{i,t}^*&=& \widehat g(x_{i,t-1})+  e_{i,t} ^*.
\end{eqnarray*}

\item[Step 2.] Train the neural network on the bootstrap resampled data starting from the original $\widehat g(\cdot)$, and train for $k$ epochs. Obtain $\widehat g^*(\cdot)$.

\item[Step 3.] Repeat Steps 1-2, $B$ times to get $\widehat g^{*^1}(\cdot),...,\widehat g^{*^B}(\cdot)$.
Let $q^*_{\alpha}$  be the $1-\alpha$ quantile of 
$$
\left|\sum_i w_i\widehat g^{*,b}(x_{i,T})- \widehat z_{T+1|T}\right|,\quad  b=1,...,B.
$$
The bootstrap $1-\alpha$ level FCI for $z_{T+1|T}$ is
$$
[\widehat z_{T+1|T}- q_{\alpha}^*, \quad \widehat z_{T+1|T}+ q_{\alpha}^*  ].
$$
\end{description}

It is critical to note that we generate $\eta_t^*$ in Step 1, which only varies over time but not across individual assets. Hence, assets share the same $\eta_t^*$ at each period. This is because expression \eqref{eq3.7} clearly shows that it is the time series variation that determines the sampling distribution of the ML model. In contrast,  if the bootstrap sample were wrongly generated, such as independently from cross-sectional residuals by either generating  $\eta^*_{it}$ (resampled idenpendently across time and firms, this corresponds to the standard bootstrap implementation) or $\eta_i^*$ (resampled indepdently across firms, but fixed over time) instead of $\eta_t^*$, it will not correctly capture the strong cross-sectional dependence,  and will dramatically understate the forecast uncertainty. We illustrate this in simulation in Section \ref{sec:simulation}.

Theorem \ref{th3} formally justifies the proposed bootstrap confidence interval for predicted expected returns. 

\begin{theorem}[Bootstrap]\label{th3}
Suppose Assumption \ref{ass2} and Assumption \ref{ass1} in the appendix hold.  Then for any $\alpha\in(0,1)$,
$$
P(|\widehat z_{T+1|T} - z_{T+1|T}|<q^*_{\alpha}) \to 1-\alpha,
$$
where $q^*_{\alpha}$ is the $1-\alpha$ bootstrap sample. 
\end{theorem}

\begin{proof}
    The proof is given in Section \ref{proofth3}.
\end{proof}

\subsection{Monte Carlo Evidence} \label{sec:simulation}
We conduct Monte Carlo simulations to assess the constructed confidence intervals, using a data generating process calibrated to excess return data; the monthly returns of 3184 assets listed in NYSE and Nasdaq from January 2015 through December 2017 (calibrated period).  Simulated data are generated from a  conditional three–factor model, with $d$ characteristics as follows:  for $k=1,...,d$,
$$
x_{i,t,k}= \frac{1}{N}\text{rank}(	\bar x_{i,t,k}),\quad \bar x_{i,t,k} =0.7  \bar x_{i,t-1,k}  + 0.5\epsilon_{i,t,k},\quad \epsilon_{i,t,k}\sim \mathcal N(0, 1).
$$
The characteristics are generated via AR(1), then normalized by taking the cross-sectional ranking. Characteristics within asset $i$ have strong temporal dependence over time, but they are independent across assets. The $\beta$-functions are generated as follows:
\begin{eqnarray*}
    g_{\beta,1}(x_{i,t-1})&=&   x_{i,t-1,1}x_{i,t-1,2},\quad  g_{\beta,2}(x_{i,t-1})=\frac{1}{d}\sum_{j=1}^dx_{i,t-1,j}^2,\cr 
    g_{\beta,3}(x_{i,t-1})&=&\text{median}\{x_{i,t-1,1},...,x_{i,t-1,d}\}. 
\end{eqnarray*}

The three factors are generated from a multivariate normal distribution whose mean vector and covariance matrix are calibrated from the monthly return of Fama-French-three factors in the calibrated period. Finally, the idiosyncratic error is generated from a heteroskedastic normal distribution: $u_{i,t}\sim\mathcal  N(0,s_{i}^2\sigma^2)$, and $s_i\sim\text{Unif}[0.1,0.9]$. Here we set   $\sigma$ so that   Median$(s_i^2\sigma^2/\Var(y_{i,t}))=50\%$. Therefore, the idiosyncratic variances are determined so that the overall signal to noise ratio is fifty percent. 

Throughout we fix $N=500$ assets, $T=240$ periods and $d=80$ characteristics.  The goal is to forecast $z_{T+1}:=\frac{1}{N}\sum_iy_{i,T+1}$ using a pooled neural network and examine the forecast distribution using the proposed methods. We train three-layer feedforward neural networks with 4 neurons on each layer. The training algorithm is \texttt{Adam} with learning rate  0.01 and conducted over 500 epochs. 

\begin{figure}[h]
	\includegraphics[width=7in]{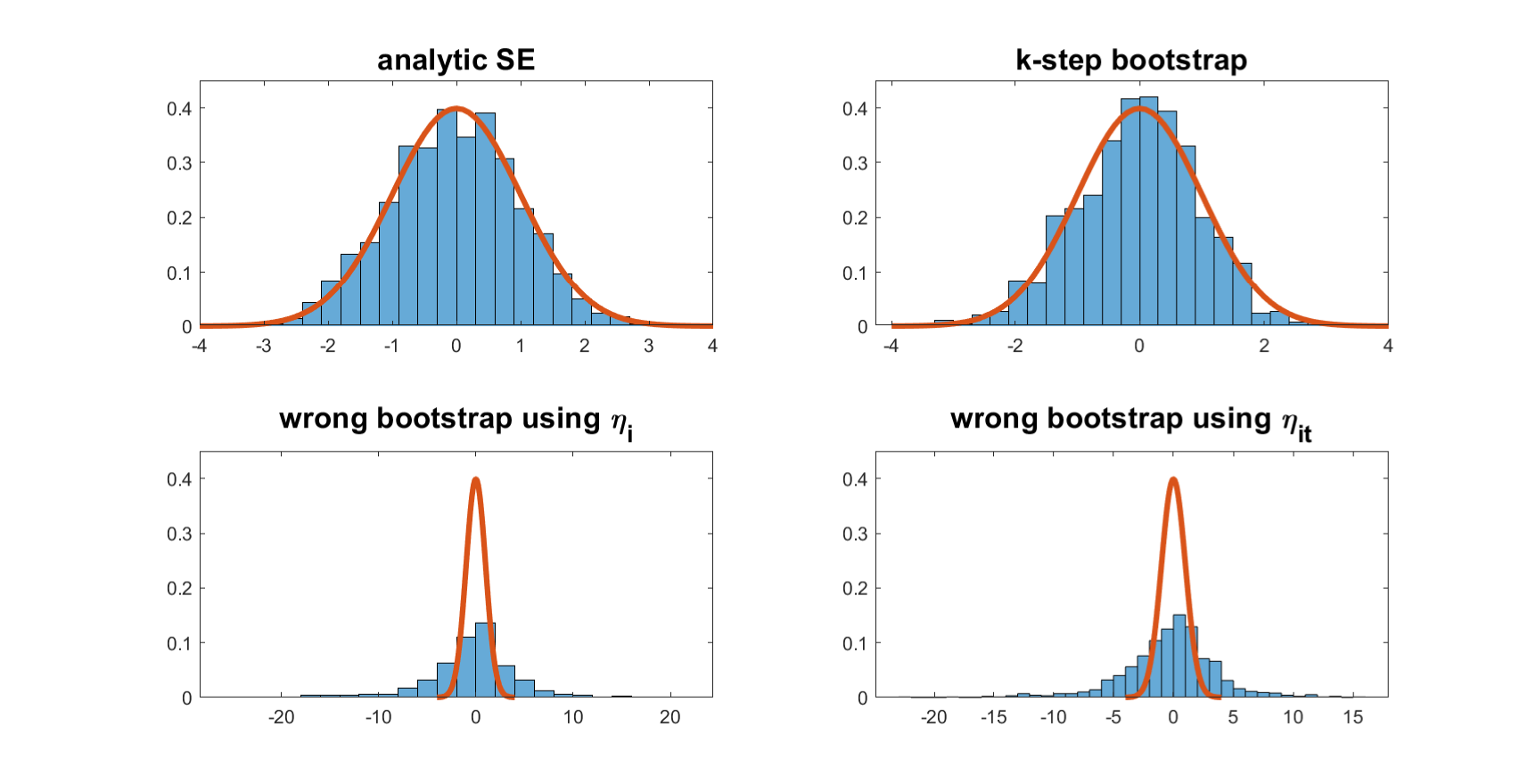}
		\caption{Histogram of t-statistics over 1000 simulation replications, and the standard normal density. The t-statistics are  standardized by either the analytical standard error $\widehat{\text{SE}}(\widehat z_{T+1})$ (top left panel) or the bootstrap interquartile range $\sigma^*$ (top right panel). The bottom panels use $\eta_i$ and $\eta_{it}$ to generate bootstrap residuals.}
	\label{fig1}
\end{figure}

To quantify the forecast uncertainty, we compute the forecast standard error of the neural network predictions, using both of the proposed methods.  For method I closed-form ML, we compute the t-statistic
$$
\frac{\widehat z_{T+1|T}- z_{T+1|T}}{\widehat{\text{SE}}(\widehat z_{T+1|T})},
$$
where the standard error $\widehat{\text{SE}}(\widehat z_{T+1|T})$ has an analytical form (\ref{eq3.11se}). Here we use five Fourier bases for $\Phi(x)$.  For method II ``$k$-step bootstrap", we generate the wild residual $\eta_t$ from the standard normal, bootstrap 100 times, and implement the $k$-step DNN bootstrap with $k=10.$ Then we compute the interquartile range of bootstrap, defined as 
$$
\sigma^* := \frac{q^*_{0.75}-q^*_{0.25}}{z_{0.75}- z_{0.25}}
$$
where $q^*_{\alpha} $ denotes the $\alpha$-quantile of bootstrap samples $ \sum_i w_i\widehat g^{*,b}(x_{i,T})- \widehat z_{T+1|T} $, and $z_{\alpha}$
denotes the $\alpha$-quantile of the standard normal distribution.  Then we  also compute the t-statistic using $\sigma^*$ in place of $\widehat{\text{SE}}(\widehat z_{T+1|T})$.  The interquantile range is a good proxy for the standard error obtained using bootstrap distribution, which is often used instead of the usual bootstrap standard error, because the former is guaranteed to be consistent but the latter is not.

The top two panels of Figure \ref{fig1} plot the histograms of the $t$ statistics over 1000 simulations and the standard normal density function. The $t$-statistics are  standardized by either $\widehat{\text{SE}}(\widehat z_{T+1|T})$ (left panel) or $\sigma^*$ (right panel). We see that although there are only 100 replications, the histograms of the t-statistics fit well to the standard normal density. Hence both proposed methods for quantifying the forecast uncertainty seem promising. 

It is critical to apply the bootstrap guided by theory. In particular, the resampling must reflect the dominant source of uncertainty, namely the factor shocks from the time series. To illustrate this, we show the distribution one would obtain if the standard bootstrap were applied, i.e. the bootstrap is applied incorrectly. The bottom two panels of Figure \ref{fig1} are the histograms of the bootstrap $t$-statistics (standardized by the bootstrap interquartile range), but the bootstrap residual is generated as 
$ e_{i,t} ^*= (y_{i,t}- \widehat g(x_{i,t-1})) \eta_i^* $ (the bottom left panel) and $ e_{i,t} ^* = (y_{i,t}- \widehat g(x_{i,t-1})) \eta_{i,t}^* $ (the bottom right panel), where $\eta_i^*, \eta_{i,t}^*\sim \mathcal N(0,1)$. These bootstraps mistreat the forecast uncertainty as  driven by the cross-sectional variation (because $\eta_i^*$ and $\eta_{i,t}^*$ both vary across $i$). Figure \ref{fig1} clearly shows that the misuse of bootstrap vastly understates the uncertainty.

\section{Applications} \label{sec:application}
In this section, we outline two applications in which the confidence intervals developed for machine learning predictions can be used. The first one leads to a sparse portfolio implementation, which features a region of ``non-participation'' (\cite{dow1992uncertainty}) in portfolio allocation. In this region, the investor does not invest in a risky asset if the associated uncertainty is too large.  In the second application (Section \ref{mult_hypothesis}), we show how the derived FCI can be used to select asset with significantly positive expected returns via multiple hypothesis testing. In addition, in the Appendix (\ref{sec:ba}), we also study a robust optimization approach building on \cite{hansen2008robustness}.  

\subsection{Portfolio Selection under Uncertainty Aversion} \label{sec:no_hold}
In the first application, we apply our developed FCI to portfolio allocation. Classic portfolio theory assumes that the investor knows the population moments determining her portfolio decisions. One of the major challenges in operationalizing this theory has been that these parameters need to be estimated and that estimation errors can often dominate portfolio decisions. This issue has long been recognized, for example, the early analysis in \cite{klein1976effect} and \cite{michaud1989markowitz}. These shortcomings have sometimes led researchers to question whether portfolio theory can be useful for applications as pointed out in \cite{demiguel2009optimal} and spurred a subsequent quest to address some of these shortcomings as in \cite{jagannathan2003risk,kan2007optimal,tu2011markowitz,yuan2023naive}.

In a seminal paper, \cite{garlappi2007portfolio} introduce a disciplined way to confront the estimation uncertainty in portfolio selection by introducing an uncertainty-averse investor. Their formulation builds on a large literature in economic theory such as \cite{ellsberg1961risk, gilboa1989maxmin, epstein1994intertemporal} that carefully distinguished the effects of risk vs. uncertainty aversion. In the context of portfolio selection, a risk-averse, expected utility investor behaves as if she knows the expected return and (co)variances. An uncertainty-averse investor, however, considers estimation uncertainty and integrates it into the portfolio selection problem. Implementing this approach requires the investor to be able to characterize the uncertainty, i.e. have a forecast confidence interval for the expected returns.

\cite{chopra2013effect} show that the effect of uncertainty in means on portfolio selection is much larger than the effect of uncertainty in variances. Our analysis therefore focuses on the effect of means. A machine learning-based estimation of covariances and their effect on portfolio selection is left for future research. Indeed, many early studies have struggled to implement mean-variance portfolio theory because means are notoriously difficult to estimate from time series. In our study, we leverage the benefits of forecasts of expected returns obtained from neural networks and also incorporate the associated estimation uncertainty. In the following, we briefly recall the main definitions of \cite{garlappi2007portfolio}.

We consider the allocation among $R$ risky assets, denoting their multivariate, conditional expected excess return as $z_{T+1|T}=(z_{1,T+1|T},...,z_{R, T+1|T})'$ and its prediction from neural networks as $\widehat z_{T+1|T}=(\widehat z_{1,T+1|T},...,\widehat z_{R, T+1|T})'$. These basis assets may be factor portfolios. In addition, we also consider the case in which the basis assets are individual asset returns. We denote the covariance matrix, $\Sigma_T$, which can be either for individual assets or factor portfolios. Throughout, we will denote the portfolio weights, which we aim to solve for as $\bomega=(\omega_1,...,\omega_R)'$. We can now describe the portfolio selection problems.\\

The \cite{Mar}, i.e.~standard mean-variance (MV) problem is given as:
\begin{eqnarray}\label{eq:mv}
\text{MV problem}&&\max_{\bomega}  \bomega'  \widehat z_{T+1|T}   -\frac{\gamma}{2}\bomega' \Sigma_T\bomega 
\end{eqnarray}
where $\gamma > 0$ is the coefficient of risk aversion.  In this problem, $\widehat{z}_{T+1|T}$ is taken as given without accounting for its associated estimation uncertainty. In contrast, the uncertainty-averse formulation explicitly considers the estimation uncertainty. It takes a ``max-min'' form. It can be interpreted as finding the best portfolio in the worst case for the expected return, i.e. 
\begin{eqnarray}\label{eq:uamv}
\text{UA-MV problem}&&\max_{\bomega} \min_{\mu\in\text{FCI}} \bomega' \mu  -\frac{\gamma}{2}\bomega' \Sigma_T\bomega,
\end{eqnarray}
where $\text{FCI}$ is the forecast confidence interval for the expected return. The forecast confidence interval takes the following form:

$$
\text{FCI}= [\widehat z_{1,T+1|T} - q_{1,\alpha},\widehat z_{1,T+1|T} + q_{1,\alpha} ]\times\cdots \times  [\widehat z_{R,T+1|T} - q_{R,\alpha},\widehat z_{R,T+1|T} + q_{R,\alpha} ],
$$
where $q_{i,\alpha}$ is the critical value for $\widehat z_{i,T+1|T}-z_{i,T+1}$ under significance level $\alpha$ obtained using either the analytic forecast standard error or bootstrap. For instance, using the analytic standard error, we can take 

$$
q_{i,\alpha}= \widehat \SE(\widehat z_{i,T+1|T}) \times  \epsilon_{\alpha},
$$ 
where $\epsilon_{\alpha}$ corresponds to the level uncertainty aversion. If the bootstrap is used to obtain a forecast confidence interval, we can use 
$$
q_{i,\alpha}= q^*_{i,\alpha},
$$
which is the $1-\alpha$ quantile of the bootstrap distribution of $ \widehat z_{i,T+1|T}$. In applications, particularly when multiple assets are considered, controlling for the Type I error rate for multiple testing is also desirable. A simple adjustment is the Bonferroni correction, i.e. setting $\alpha = 0.05/R$.

While the intuition of problem \eqref{eq:uamv} is clear, it may appear hard to solve because of the inner maximization.  \cite{garlappi2007portfolio} characterize each element of the solution as a function of other elements.  Our contribution in this framework, is in the following theorem, which shows that the UA-MV problem can be reformulated as a $\ell_1$- penalized optimization problem, known as Lasso.

\begin{theorem}\label{th5}
The multivariate UA-MV problem, \eqref{eq:uamv}, is equivalent to the following adaptive Lasso formulation:
\begin{equation}\label{eq4.5}
   \min_{\bomega=(\omega_1,...,\omega_R)}  \frac{\gamma}{2}\bomega' \Sigma_T\bomega -\sum_{i=1}^R\omega_i \widehat z_{T+1|T,i} +\sum_{i=1}^Rq_{\alpha,i}|\omega_i|.
\end{equation}
\end{theorem}

\begin{proof}
    The proof is given in Section \ref{proofth4}. 
\end{proof}

This result offers an economic justification for using penalized shrinkage portfolio selection, as developed by  \cite{ao2019approaching} and  \cite{kozak2020shrinking}. Furthermore, it justifies that the penalization parameter should be explicitly chosen as $q_{\alpha,i}$, the quantile of the ML forecast error. The following remarks provide additional guidance for implementation in different cases.

\begin{remark}
    If a risk-free asset is not available, 
    $\widehat z_{T+1|T,i}$ denotes the raw return. Then    the above problem becomes a constrained adaptive Lasso problem, i.e.
    $$
       \min_{\bomega=(\omega_1,...,\omega_R)}  \frac{\gamma}{2}\bomega' \Sigma_T\bomega -\sum_{i=1}^R\omega_i \widehat z_{T+1|T,i} +\sum_{i=1}^Rq_{\alpha,i}|\omega_i|, \;\; \text{subject to} \;\; \sum_{i=1}^{R} \omega_i=1
    $$
\end{remark}

\begin{remark}
    The problem can be solved with standard software packages. To see this,  write the Cholesky decomposition of $\Sigma_T$ as 
    $$
    \Sigma_T = LL'
    $$
    where $L$ is a  $R\times R$ lower triangular matrix and denote $Q_\alpha=\text{diag}(q_{\alpha, 1},...,q_{\alpha, R})$ be a diagonal matrix with elements $q_{\alpha, i}$. Then it is easy to derive that (\ref{eq:uamv}) is equivalent to the following adaptive Lasso problem:
    \begin{equation} \label{eq:ua_lasso_reg}
        \min_{\bomega}\frac{1}{2}\|Y^*- X^*\bomega\|_2^2 +\|Q_{\alpha}\bomega\|_1, \quad \text{with} \quad X^* =\sqrt{\gamma} L', \;\; Y^*=\frac{1}{\sqrt{\gamma}}L^{-1}\widehat z_{T+1|T}.
    \end{equation}
    If a risk-free asset is not available, the problem is again augmented with the constraint $\sum_{i=1}^{R} \omega_i=1$. It is well-known that the Lasso solution may be sparse and thus produce portfolio allocations, for which some (or many) $\omega_i$'s are zero.
\end{remark}

The $\ell_1$-penalized formulation facilitates the economic interpretation of the behavior of uncertainty-averse investors compared to that of mean-variance investors.  To gain the intuition, in below we study a simple case  of a single risky asset:
\begin{equation}\label{eq5.1add}
\max_{\omega}\min_{\mu} \omega\mu   -\frac{\gamma}{2} \omega^2\sigma^2 \quad \text{subject to} \quad |\mu- \widehat z_{T+1|T}| \leq  q_{\alpha},
\end{equation}
where $q_{\alpha}$ denotes  the $(1-\alpha)$-quantile for the predicted expected return.  Corollary \ref{th:uarep1} provides a closed-form solution and allows for a comparison with the mean-variance solution in which the effect of estimation uncertainty can be seen directly.  In the following, we denote $(x)_+ = \max\{x,0\}$ and $\sgn(x)$ as the sign of $x$. 

\begin{corollary}
\label{th:uarep1}
The uncertainty-constrained MV problem \eqref{eq5.1add} is equivalent to the following:
$$
\min_{\omega} \frac{1}{2}\left(\omega- \omega^{\MV}\right)^2 + \lambda_{\alpha} |\omega|
$$
and the optimal solution is 
$$
\omega^{\UA} = \sgn(\omega^{\MV})(|\omega^{\MV}| -\lambda_\alpha)_+
$$
where $$\omega^{\MV}=\frac{\widehat z_{T+1|T}}{\gamma\sigma^2},\quad 
\lambda_{\alpha}= \frac{q_{\alpha}}{\gamma\sigma^2}.
$$
\end{corollary}
\begin{proof}
    The proof is given in Section \ref{proofth4}. 
\end{proof}

In the above theorem, $\omega^{\MV}$ is the classic mean-variance portfolio without uncertainty constraints, i.e.~the solution to 
$$
\max_{\omega}  \omega \widehat z_{T+1|T}  -\frac{\gamma}{2} \omega^2\sigma^2.
$$
The corollary directly illustrates why the UA-MV problem yields possible non-participation. For a fixed coefficient of risk aversion ($\gamma$) and variance ($\sigma^2$), an investor will choose not to invest in the risky asset if the estimation uncertainty ($q_{\alpha}$) is too large. More importantly, this result provides a direct economic interpretation of the UA-MV problem. The solution captures investor behavior in which she tests the following null hypothesis:
$$
H_0:\quad z_{T+1|T} = 0.
$$

On the one hand, when she does not reject the null hypothesis, $|\widehat{z}_{T+1|T}  | \leq q_{\alpha}$, the investor finds the predicted excess return not significantly different from zero as the level of uncertainty is too large. In this case,  she will hold no position in the risky asset and $\omega^{\UA} = 0$. Her optimal portfolio is then solely holding the risk-free asset. At first glance, this may seem like excessively conservative behavior. However, \cite{bessembinder2018stocks} performs an extensive empirical investigation of this hypothesis and finds that many stocks do not outperform treasuries ex-post in a significant way. 

On the other hand, when the investor predicts that the expected return of the risky asset is significantly different from the risk-free rate, which happens if $|\widehat{z}_{T+1|T}  | > q_{\alpha}$, she then starts investing in the risky asset. The decision of whether to short or long the risky asset (the sign of $\omega^{\UA}$) is determined by the sign of $\widehat{z}_{T+1|T}$. However, even in this case, the investor will still invest more cautiously in the risky asset by \textit{shrinking} her investment towards the risk-free rate. For instance, suppose the investor finds that $\widehat{z}_{T+1|T} >  q_{\alpha}$, then her allocation in the risky asset is 
$$ 
\omega^{\UA} = \omega^{\MV} - \frac{q_{\alpha}}{\gamma \sigma^2} > 0.
$$
Instead of adopting the classic mean-variance portfolio, she reduces her allocation to the risky asset, and the amount of reduction, $q_{\alpha}/(\gamma \sigma^2)$, reflects her tolerance towards uncertainty, which is closely linked to risk aversion.

We illustrate the optimal solution to the UA-MV problem in Figure \ref{figNohold}. The upper panel of Figure \ref{figNohold} plots the optional min-max weight, $\omega^{\UA}$, against the classic MV portfolio. Here we fix the level of forecast standard error (FSE) and thus $q_{\alpha}$. As we can see, $\omega^{\UA}$ is zero for small magnitudes of $|\omega^{\MV}|$, i.e., the investor does not allocate towards this asset. The allocation toward the risky assets starts to increase but with constant shrinkage relative to the MVE as the latter deviates from zero. The lower panel of Figure \ref{figNohold} plots $\omega^{\UA}$ as a function of the forecast standard error for a fixed $\omega^{\MV}$. As the forecast standard error increases, her position in the risky asset decays linearly and eventually becomes zero.

In the case of two risky assets and no risk-free asset, we also provide an explicit solution and discuss its properties in \ref{sec:ua_mult}.

\begin{figure}[H]
\caption{Mean-Variance vs. Uncertainty Averse (max-min) Portfolio Allocation} \label{figNohold}
\vspace{-2mm}
\caption*{\setstretch{1.0} \scriptsize This figure shows the uncertainty averse (max-min) solution for the optimal portfolio weight in comparison with the standard mean-variance portfolio weight. In the upper panel, we hold the level of the forecast standard error (FSE) fixed and vary the mean-variance weight ($\omega^{\MV}$). In the lower panel, we fix the mean-variance weight at 0.2 and vary the forecast standard error (FSE).}
    \includegraphics[width=1.0\textwidth]{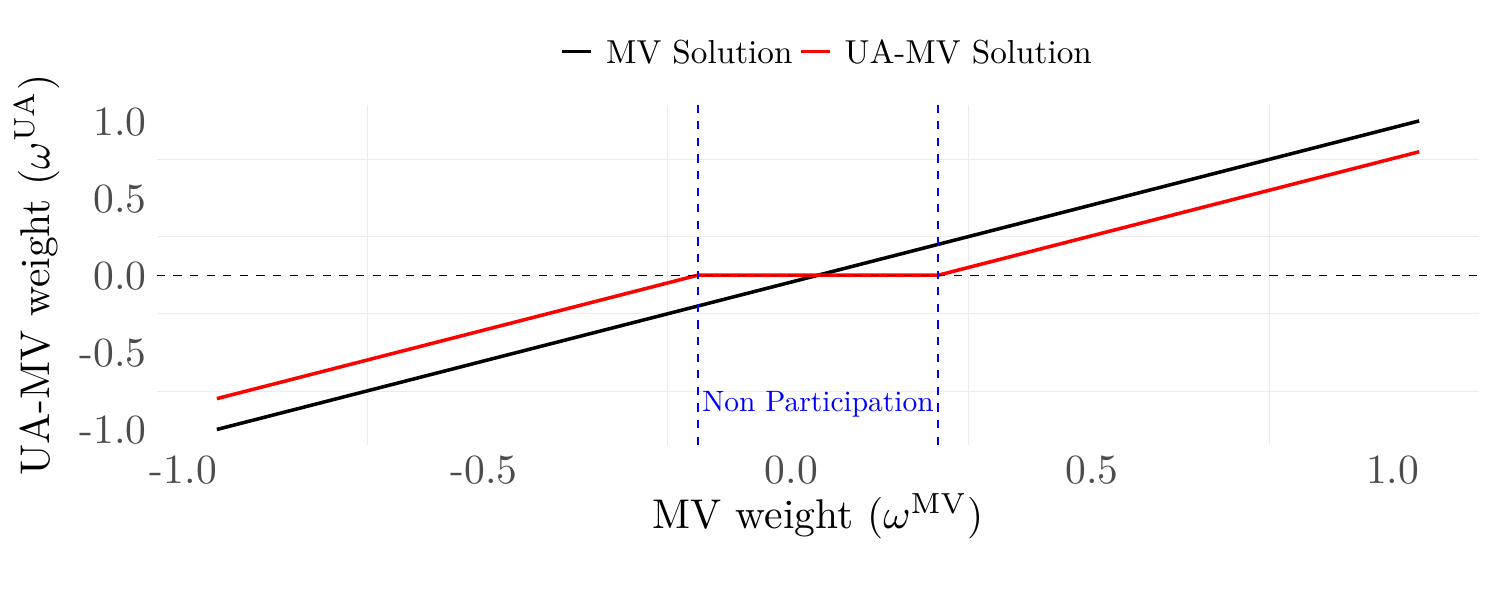}
    \includegraphics[width=1.0\textwidth]{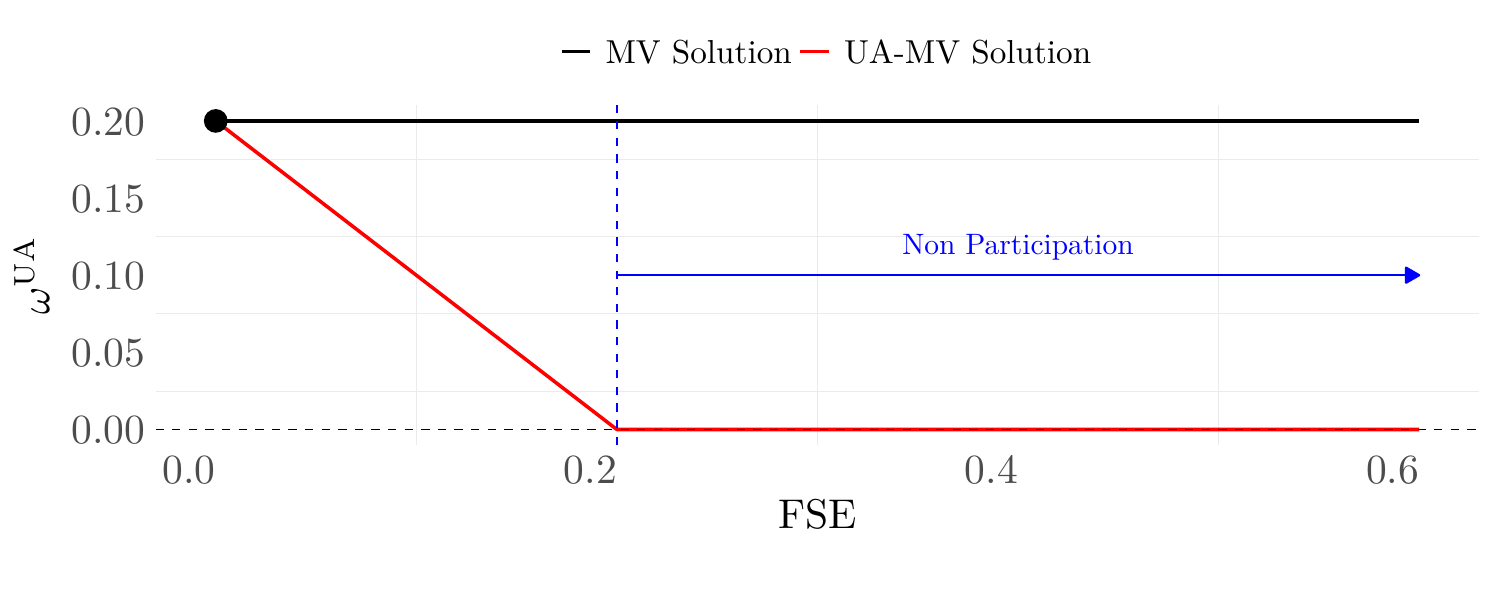}
\end{figure}

\subsection{Asset Selection}\label{mult_hypothesis}
In the second application, we apply the FCI to security selection. The literature often considers long-short portfolios sorted on characteristics or return predictions typically ignoring short-sale costs. In a recent study \cite{muravyev2022anomalies} obtain implied short-sale costs from the equity options markets and find that many anomalies are no longer profitable after accounting for short-selling fees. Moreover, several major market participants such as mutual funds are either prohibited from short-selling or do not engage in shorting for other reasons. \cite{chen2013first} document that fewer than 10\% of mutual funds engage in short-selling and in general do not have large short positions. More generally, and in particular in the light of the findings of \cite{bessembinder2018stocks}, it is natural to ask whether expected returns are indeed significantly positive. In an intriguing analysis, he documents that the US equity premium is indeed highly concentrated and can be attributed to about 5\% of stocks. While it is unlikely that our forecast confidence intervals (or any other method) can extract all ex-post winners, the use of machine learning FCI, as we describe below, highlights a good way of finding securities with positive expected returns.

In the following, we consider a long-only mutual fund manager. She needs to build a portfolio with $K$ stocks. Rather than merely allocating her funds toward the $K$ stocks with the highest predicted returns, she aims to allocate toward securities for which she has stronger conviction, i.e. those that have significantly positive predictions. So for each security under consideration, she conducts a hypothesis test 
$$
H_0^i:  z_{i,T+1|T} = 0,$$
versus the one-sided alternative 
$$
H_a^i: z_{i,T+1|T} > 0 
$$
for each $i=1,..., R,$ i.e., she aims to detect which stocks have significantly and positively predicted expected returns. Given that there may be a large number of securities under consideration, this is a \textit{multiple hypotheses} testing problem. Alternatively, one might consider testing if an expected excess return is significantly larger than the transaction costs associated with trading the security. We leave this extension for future research.\footnote{\cite{jensen2024machine} highlight the importance of transaction costs in the context of machine learning portfolios.}

We compute the t-statistics of the neural network predicted excess return for each asset: 
$$
t_i= \frac{\widehat z_{i,T+1|T}}{\widehat\SE(\widehat z_{i,T+1|T})}.
$$
Let $p_i$ denote the corresponding p-value. We then determine a cutoff value $c_0>0$ so that $H_0^i$ is rejected if $p_i<c_0$. The cutoff value should be determined to adjust for the type I error under the multiple testing setup. One of the widely used multiple testing adjustments is based on the false discovery rate, defined to be the expected value of  $\mathcal F/\mathcal R$,
where 
$$
\mathcal F=\sum_{i=1}^N1\{i\leq R: p_i<c_0,\text{ but } z_{i,T+1|T}=0 \} =\sum_{i=1}^N1\{H_0^i \text{ is falsely rejected}\}
$$
and 
$$
\mathcal R=\sum_{i=1}^N1\{i\leq R: p_i<c_0 \} =\sum_{i=1}^N1\{H_0^i \text{ is   rejected}\}
$$
We follow the  \cite{benjamini1995controlling} procedure, who first sort the p-values $p_{(1)}\leq\cdots \leq p_{(R)}$ and determine
$$
c_0= p_{(K)},\quad \text{where } K = \max\{i: p_i\leq \alpha i/R \}
$$
where $\alpha$ is the desired significance level such as 0.05.  Then the false discovery rate can be controlled to be below $\alpha$, and is robust to  cross-sectional dependence among tests (see \cite{benjamini2001control}). 

As illustrated by \cite{sullivan1999data, barras2010false} and \cite{harvey2020false}, multiple testing is a valuable tool for mitigating data-snooping bias in performance evaluation and for selecting assets that are both statistically and economically significant. In our empirical study, we apply the outlined procedure at the individual stock level, where excess returns are predicted using neural networks. Our results demonstrate that selecting assets based on ML t-statistics, when adjusted by the forecast standard errors, delivers improved performance compared to naive selections that do not account for prediction uncertainty.

\section{Empirical Analysis} \label{sec:empirics}
\subsection{Data and Implementation} \label{sec:data}
We take the dataset of \cite{jensen2022there} as our starting point. It uses stock returns, volume, and price data from the Center for Research in Security prices (CRSP) monthly stock file. Our sample starts in January 1955 and ends in December 2021. The total number of stocks is slightly over 23,000 and the average number of stocks per month is approximately 3,663. Following standard conventions in the literature, we restrict the analysis to common stocks of firms incorporated in the US trading on NYSE, Nasdaq or Amex. Balance sheet data are obtained from Compustat.  In order to avoid potential forward-looking biases, we lag all characteristics that are built on Compustat annually by at least six months and all that build on Compustat quarterly by at least four months. In order to mitigate a potential back-filling bias as noted by \cite{banz1986sample}, we discard the first 24 months for each firm. We impute missing data using the method of \cite{freyberger2024missing}. Table \ref{tab:chars_overview} in the Appendix provides an overview of the 123 characteristics we employ. We obtain the 1-month T-Bill rate from Kenneth French's data library.

Following \cite{gu2020empirical} we split the sample into 18 years of training (1955 - 1972), 10 years of validation (1973 - 1982) and use the remaining years (1983 - 2021) for out-of-sample testing. We re-train the model every 12 months and increase the training sample by 1 year, keeping the validation sample fixed at 10 years, but we roll it forward so that the most recent 12 months are included. For the neural network model, we use a three-layer feedforward neural network with 32, 16 and 8 neurons on the hidden layers.\footnote{Section \ref{sec:tuning} in the Appendix gives an overview of the tuning parameters.}

After fitting the neural network, we substitute in the firm level characteristics  $x_{i,T}$, and  predict the individual stocks for month $T+1$ as
$$
\widehat y_{i,T+1|T} = \widehat g(x_{i,T}).
$$
From the previous step, we obtain predictions for each firm $i$. We then either use the predictions for individual firms in portfolio selection. We use the expected returns and their confidence intervals obtained through the machine learning models.

To compute forecast confidence intervals, we implement both the closed-form ML approximation and the $k$-step bootstrap. For the bootstrap, we set $k=10$, i.e. we re-train the models on each bootstrap sample for 10 epochs. The closed-form ML approximation uses a Fourier basis expansion $\phi(x)= (\sin(j\pi x/4), \cos(j\pi x/4), j=1...3)$. We then obtain the standard error for individual assets or portfolios following Section \ref{sec:analytical} for use in applications.

\subsection{Uncertainty Averse Portfolios of Individual Stocks} \label{sec:ua_ind}
We implement both the mean-variance efficient portfolio (problem \eqref{eq:mv}) and the uncertainty averse portfolio (\eqref{eq:uamv}) for the 500, 750 and 1,000 largest stocks.\footnote{In the early parts of the sample, we encounter a few months in which fewer firms are available. In these cases, we use all the available stocks with a 240-month history.} Throughout, we use a 240-month estimation window to estimate the covariance matrix using the POET estimator \citep{POET}. By conditioning on large firms with a 240 months history, we deliberately create a sample of very large firms to mitigate concerns over small and illiquid stocks. In the context of cross-sectional anomalies, \cite{patton2020you} show that transaction cost play an important role and in the machine learning setting, \cite{avramov2023machine} argue that standard implementations often concentrate the predictability on small and illiquid securities. Figure \ref{fig:firm_size} plots the median effective normalized size for three scenarios, 500, 750 and 1000 largest firms with a 240-month history. The normalized size ranks the market equity of all firms each month $t$, the ranks are then divided by the number of stocks each period so that the largest firm has a normalized size of 1 and the smallest firm has a normalized size of $1/N_t$, where $N_t$ is the number of firms each period. 

\begin{figure}[H]
\caption{Firm Size Distribution our Sample} \label{fig:firm_size}
\vspace{-2mm}
\caption*{\setstretch{1.0} \scriptsize ...}
    \includegraphics[width=1.0\textwidth]{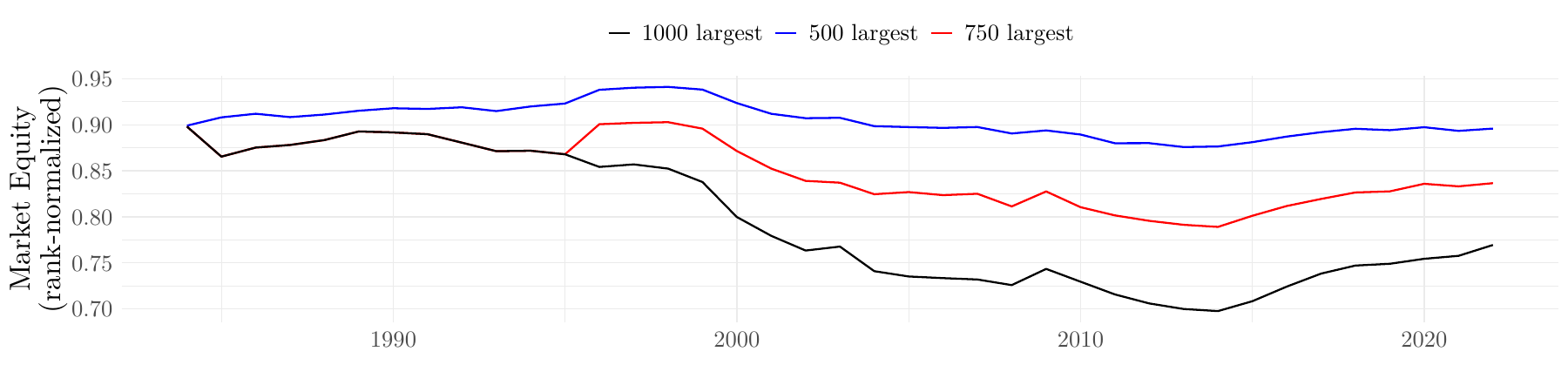}
\end{figure}

Predictions for expected returns are obtained from the neural network model. To incorporate estimation uncertainty, we can either obtain standard errors from the closed-form ML approximation, the bootstrap, or use the more conservative of the two. The latter connects well with the max-min nature of the uncertainty-averse portfolio selection approach. Table \ref{tab:ua_ind_worst} shows the results for the most conservative version of the standard error.\footnote{Table \ref{tab:ua_ind_analytical} and \ref{tab:ua_ind_bootstrap} in the appendix show  results for the analytical  and bootstrap standard error, respectively.} We compare the UA portfolio with three benchmarks: the mean-variance efficient portfolio (MVE), the equal-weighted portfolio (EW), and the global minimum-variance portfolio (GMVP). For the MVE, we use a coefficient of risk-aversion of one ($\gamma=1$). As we are using excess returns, there is no automatic constraint on the magnitude of the portfolio weights, we therefore normalize the standard deviation of the in-sample returns to 20\% annualized for the MVE and apply this constant also to the UA portfolios; we apply a different normalizing constant for the GMVP since the weights are of different magnitude.

Table \ref{tab:ua_ind_worst} shows that the MVE achieves an admirable annualized Sharpe ratio of 1.15, 1.07 and 1.16 for the three cases. This is notable since the result is achieved on a set of very large firms that are likely not subject to large liquidity frictions. It also considerably outperforms the equal-weighted portfolio and the global minimum-variance portfolio in all cases. The uncertainty-averse MV portfolio (UA 25, UA 50, and UA 75) often achieves better Sharpe ratios than the MVE. This is particularly interesting, as these portfolios are not formed to maximize the Sharpe ratio, whereas the MVE is formed to do precisely that. However, even a mean-variance investor can often achieve higher utility through the UA portfolios. The higher Sharpe ratios of the UA portfolios are achieved by reducing the standard deviation more strongly than average excess returns. Figure \ref{fig:cum_returns_ind_worst} shows the cumulative returns for the case of the 500 largest firms for the six portfolios. We can clearly see that the UA portfolios decrease less during recession periods, highlighting the mechanism of lower downside risk.

As we increase the confidence level, however, the rate of ``non-participation" increases, resulting in many weights of zero as predicted by our theory. In Figure \ref{fig:ua_weights_ind} we show the empirical analogue to Figure \ref{figNohold} for the case of the largest 500 firms. In the figure we contrast the weights of the mean-variance efficient portfolio and the UA portfolio for two dates. Clearly, we can see that the forecast uncertainty varies over time. For the example of November 1987, we set many more weights to zero in the UA portfolio, whereas we set fewer weights to zero in March 2021. In addition, we can see that the portfolio becomes more conservative as we increase the confidence level. Intuitively, the region of non-participation is proportional to the width of the confidence interval and as the confidence level increases, the region of non-participation also increases resulting in more weight of zero. Overall, the plot in Figure \ref{figNohold} demonstrates that the empirical results align closely with the theoretical predictions of UA-behavior outlined in Corollary \ref{th:uarep1}.

\begin{table}[H]
\caption{Performance Statistics for Mean-Variance and Max-Min Portfolio (conservative SE)} \label{tab:ua_ind_worst}
\begin{singlespace}
\begin{footnotesize}
This table shows the annualized mean, annualized standard deviation, annualized Sharpe ratio, Sortino ratio, the maximum drawdown, the best and worst months for the mean-variance efficient portfolio (MVE), the global minimum variance portfolio (GMVP), an equal weighted portfolio and the UA-MV portfolios for the 25\%, 50\% and 75\% confidence level. The columns Minimum Zeros, Median Zeros, Maximum Zeros show the minimum, median and maximum fraction of the portfolio weights that are set to zero in the UA-MV approach. The standard error is the maximum of the closed-form ML approximation  and the bootstrap standard error using an estimation window of 240 months. All results are out-of-sample for the period from January 1983 - December 2021.
\end{footnotesize}
\end{singlespace}
\resizebox{\ifdim\width>\linewidth\linewidth\else\width\fi}{!}{
\begin{tabular}[t]{lrrrrrrrrrr}
\toprule
\multicolumn{1}{c}{} & \multicolumn{1}{c}{Mean (\%)} & \multicolumn{1}{c}{\makecell[c]{Standard\\ Deviation (\%)}} & \multicolumn{1}{c}{\makecell[c]{Sharpe\\ Ratio}} & \multicolumn{1}{c}{\makecell[c]{Sortino\\ Ratio}} & \multicolumn{1}{c}{\makecell[c]{Maximum\\ Drawdown}} & \multicolumn{1}{c}{\makecell[c]{Best\\ Month}} & \multicolumn{1}{c}{\makecell[c]{Worst\\ Month}} & \multicolumn{1}{c}{\makecell[c]{Minimum\\ Zeros}} & \multicolumn{1}{c}{\makecell[c]{Median\\ Zeros}} & \multicolumn{1}{c}{\makecell[c]{Maximum\\ Zeros}} \\
\cmidrule(l{3pt}r{3pt}){2-2} \cmidrule(l{3pt}r{3pt}){3-3} \cmidrule(l{3pt}r{3pt}){4-4} \cmidrule(l{3pt}r{3pt}){5-5} \cmidrule(l{3pt}r{3pt}){6-6} \cmidrule(l{3pt}r{3pt}){7-7} \cmidrule(l{3pt}r{3pt}){8-8} \cmidrule(l{3pt}r{3pt}){9-9} \cmidrule(l{3pt}r{3pt}){10-10} \cmidrule(l{3pt}r{3pt}){11-11}
\addlinespace[0.3em]
\multicolumn{11}{l}{\textbf{500 largest first with 240 months history}}\\
\hspace{1em}MVE & 30.37 & 26.32 & 1.15 & 0.57 & 68.49 & 36.26 & -28.30 & 0.00 & 0.00 & 0.00\\
\hspace{1em}GMVP & 23.88 & 31.52 & 0.76 & 0.35 & 74.82 & 39.67 & -37.20 & 0.00 & 0.00 & 0.00\\
\hspace{1em}EW & 10.54 & 15.29 & 0.69 & 0.30 & 53.58 & 14.18 & -24.29 & 0.00 & 0.00 & 0.00\\
\hspace{1em}UA 25 & 18.89 & 15.94 & 1.19 & 0.60 & 49.09 & 22.79 & -18.11 & 0.25 & 0.39 & 0.57\\
\hspace{1em}UA 50 & 10.73 & 8.68 & 1.24 & 0.66 & 27.46 & 11.83 & -10.41 & 0.55 & 0.71 & 0.86\\
\hspace{1em}UA 75 & 4.73 & 4.28 & 1.11 & 0.59 & 10.97 & 6.58 & -4.12 & 0.78 & 0.90 & 0.97\\
\addlinespace[0.3em]
\multicolumn{11}{l}{\textbf{750 largest first with 240 months history}}\\
\hspace{1em}MVE & 31.52 & 29.44 & 1.07 & 0.51 & 71.14 & 37.16 & -38.59 & 0.00 & 0.00 & 0.00\\
\hspace{1em}GMVP & 23.58 & 32.72 & 0.72 & 0.32 & 78.23 & 40.08 & -50.49 & 0.00 & 0.00 & 0.00\\
\hspace{1em}EW & 10.83 & 16.01 & 0.68 & 0.29 & 53.80 & 18.44 & -25.72 & 0.00 & 0.00 & 0.00\\
\hspace{1em}UA 25 & 21.35 & 17.96 & 1.19 & 0.62 & 51.70 & 27.46 & -23.72 & 0.24 & 0.38 & 0.55\\
\hspace{1em}UA 50 & 13.17 & 11.16 & 1.18 & 0.78 & 26.81 & 28.96 & -11.68 & 0.53 & 0.70 & 0.87\\
\hspace{1em}UA 75 & 7.65 & 9.35 & 0.82 & 0.99 & 8.99 & 39.36 & -4.06 & 0.78 & 0.89 & 0.96\\
\addlinespace[0.3em]
\multicolumn{11}{l}{\textbf{1000 largest first with 240 months history}}\\
\hspace{1em}MVE & 33.89 & 29.23 & 1.16 & 0.56 & 68.42 & 37.16 & -37.57 & 0.00 & 0.00 & 0.00\\
\hspace{1em}GMVP & 24.83 & 32.33 & 0.77 & 0.34 & 84.91 & 36.92 & -45.10 & 0.00 & 0.00 & 0.00\\
\hspace{1em}EW & 11.12 & 16.32 & 0.68 & 0.29 & 55.78 & 19.83 & -25.72 & 0.00 & 0.00 & 0.00\\
\hspace{1em}UA 25 & 24.10 & 17.53 & 1.38 & 0.73 & 45.85 & 27.46 & -22.22 & 0.25 & 0.37 & 0.55\\
\hspace{1em}UA 50 & 15.59 & 10.66 & 1.46 & 1.08 & 18.72 & 28.96 & -10.19 & 0.52 & 0.69 & 0.87\\
\hspace{1em}UA 75 & 8.76 & 9.20 & 0.95 & 1.30 & 6.61 & 39.36 & -3.22 & 0.76 & 0.89 & 0.96\\
\bottomrule
\end{tabular}}
\end{table}

\begin{figure}[H]
\caption{Cumulative Returns for the MVE and UA Portfolio of the 500 Largest Firms} \label{fig:cum_returns_ind_worst}
\vspace*{-2mm}
\begin{singlespace}
\begin{footnotesize}
This figure shows the cumulative returns for the mean-variance efficient portfolio (MVE), the global minimum variance portfolio (GMVP), an equal-weighted portfolio, and the UA-MV portfolios for the 25\%, 50\%, and 75\% confidence levels. The forecast uncertainty uses the conservative approach, i.e., the maximum of the analytical and bootstrap standard error. NBER recessions are depicted in gray-shaded areas. All results are out-of-sample for the period from January 1983 - December 2021.
\end{footnotesize}
\end{singlespace}
\begin{center}
\includegraphics[width=1\linewidth]{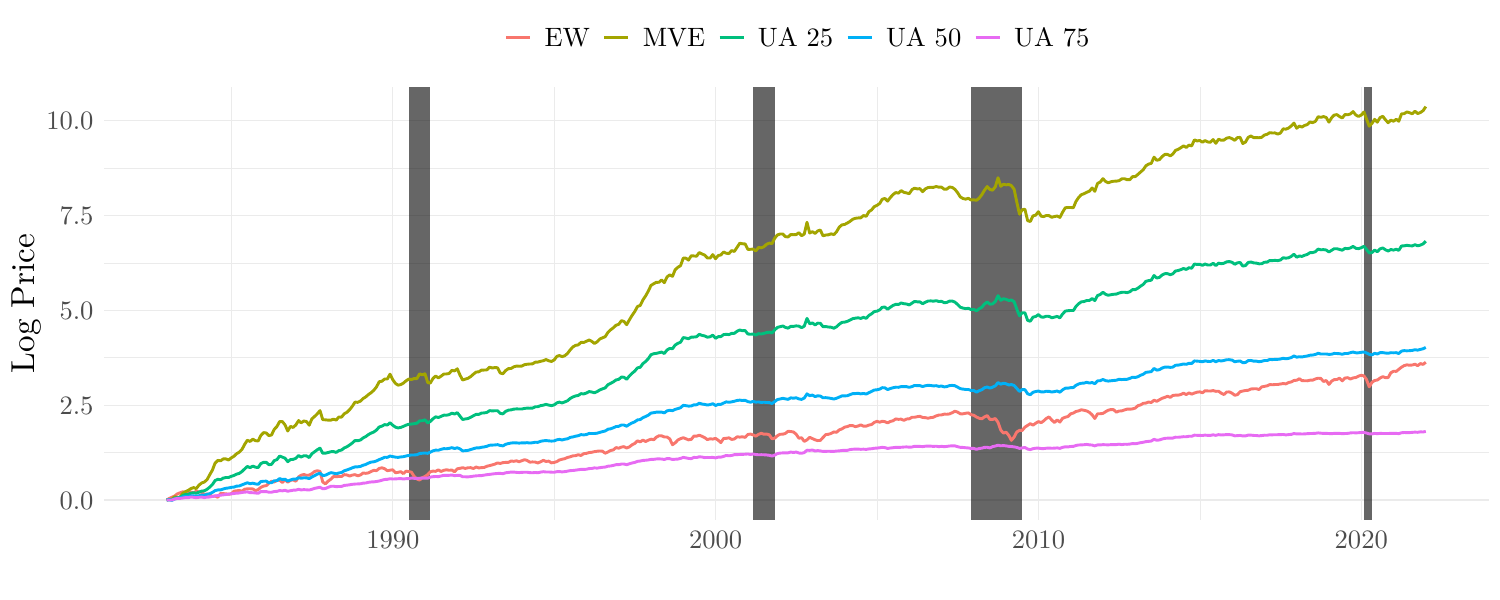}    
\end{center}
\end{figure}

\begin{figure}[H]
\caption{Portfolio Weights of the mean-variance efficient and uncertainty average portolios the 500 Largest Firms} \label{fig:ua_weights_ind}
\vspace*{-2mm}
\begin{singlespace} 
\begin{footnotesize}
This figure shows the weights for the mean-variance efficient portfolio (MVE) and the UA 25 (left), UA 50 (middle) and UA 75 (right) portfolio for two months (November 1987 and March 2012). The weight pairs $(\omega^{\text{MV}},\omega^{\text{UA}})$ are shown as blue dots (November 1987) and red triangles (March 2012). The weights are shown for the case of 500 largest firms with histories of 240 months.
\end{footnotesize}
\end{singlespace}
\begin{center}
\includegraphics[width=0.32\linewidth]{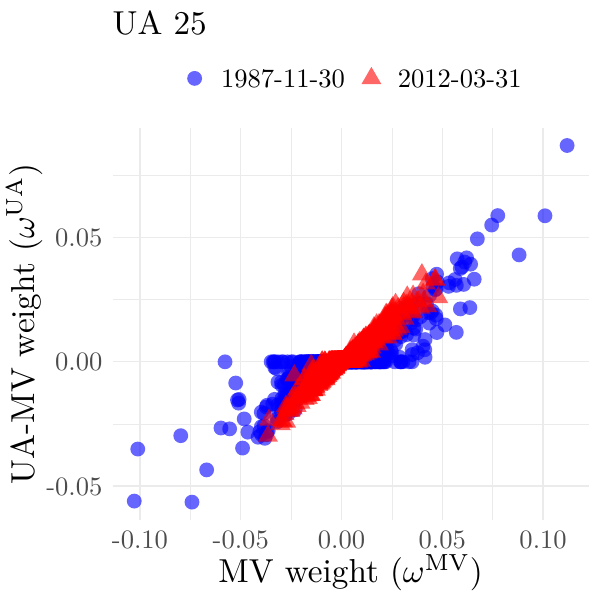}
\includegraphics[width=0.32\linewidth]{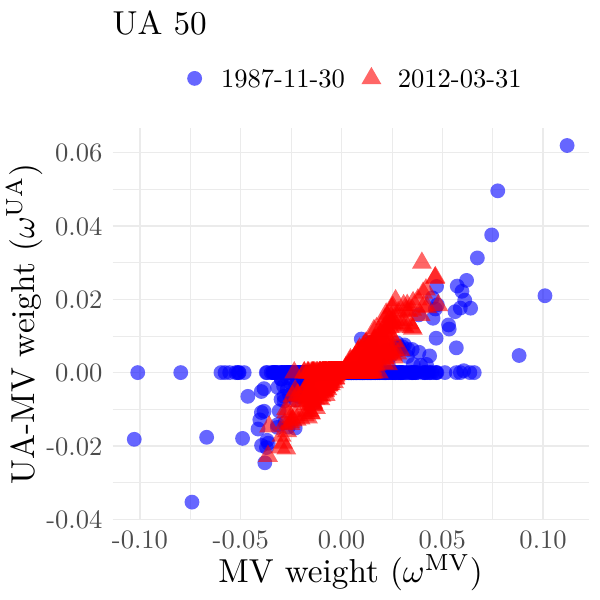}
\includegraphics[width=0.32\linewidth]{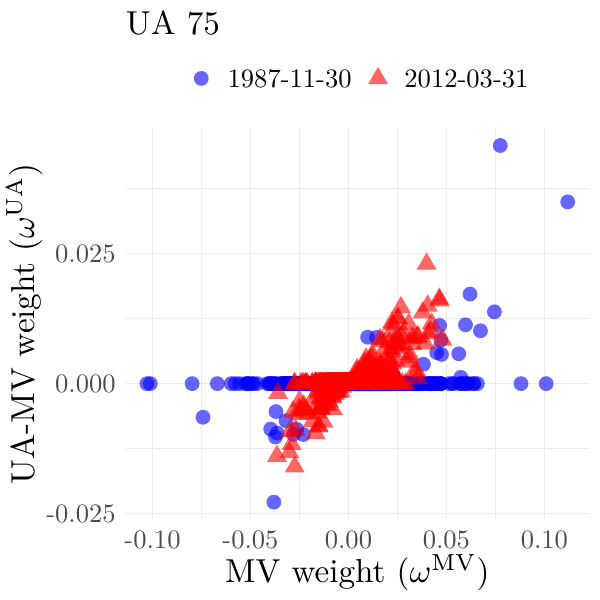}
\end{center}
\end{figure}


\subsection{Long Only Portfolios} \label{sec:ua_factor}
The classic characteristic-based anomalies are typically studied in the form of ``high minus low'' zero-investment portfolio. However, many large investors in the market do not engage in short-selling. Mutual funds are among the largest investors in the US equity markets and only very rarely enter into a short position. It is therefore of interest if a long-only investor could also benefit from expected returns obtained from neural networks and if the developed forecast confidence intervals are useful in their portfolio decisions. We will therefore consider a long-only mutual fund manager. The manager uses neural networks to obtain predictions for expected returns and then aims to buy the stocks with the highest predicted expected excess returns. 

We will study three versions of this portfolio strategy. First, the manager conducts multiple tests: 
$$
H_0^i:  \mathbb E  [y_{i,T+1}|\mathcal F_T] = 0,$$
versus the one-sided alternative 
$$
H_a^i: \mathbb E  [y_{i,T+1}|\mathcal F_T] > 0 
$$
for each $i=1,..., R.$ Therefore, the manager wants to focus on stocks whose expected return is significantly positive. In this setup, she uses the developed FCI to test hypotheses about the expected return predictions from neural networks. She carries out a multiple hypothesis testing adjustment using the false discovery rate as outlined in Section \ref{mult_hypothesis}.  The fund manager then longs the $K$ stocks with the highest predicted expected returns stocks that are also significantly positive and forms an equally weighted portfolio. We will refer to this strategy as ``FCI-FDR''. 

For comparison, a benchmark method is to conduct multiple tests using  a $t$-test from average realized returns and their standard error, i.e.
$$
t_i^N = \frac{\bar y_{i}}{\widehat \SE(\bar y_{i})},
$$
where $\bar y_i$ is the average realized return and $\widehat \SE$ is the conventional standard error. We will refer to this strategy as ``Naive-FDR''. Finally, consider the simplest implementation, where the fund manager simply goes long the $K$ stocks with the highest predicted expected returns using neural networks and forms an equal weighted portfolio. We will refer to this strategy as ``Highest-$K$''.  In our application, we will vary the size of the portfolio ($K$) between 100, 200 and 500.

It has been well documented that the features that are used in the neural network have good predictive power, and our sample necessarily covers a period during which they were particularly successful. It is therefore important to focus on relative comparisons.  The upper panel of Table \ref{tab:perf_best_ideas} shows the results when all firms are considered as potential investments.   We see that the ``Highest-$K$'' portfolio achieves an annualized excess return of about 66\% with a standard deviation of about 49\% annually.  When we compare the strategy
with ``Naive-FDR'',  which uses realized returns to assess statistical significance,  we see a non-trivial drop in average returns to 41.39\% per annum. Notably, this drop is not accompanied by an equal drop in standard deviation, so that the Sharpe ratio also drops in comparison. The FCI-FDR strategy that employs the FCI for neural networks achieves an even higher average return with 67\% per year with a smaller standard deviation, thus achieving the highest Sharpe ratio. 
 
As we move from the concentrated portfolio of 100 stocks to a portfolio of 500 stocks, the average returns are lower. However, when we compare the portfolios among each other, we see the same ranking, i.e. the FCI-FDR strategy outperforms the other two in terms of risk-adjusted returns, but the magnitudes become stronger. The return difference between the FCI-FDR and the Highest-$K$ is now almost 6\% per year,  with a slightly lower standard deviation. The Sharpe ratio is thereby strongly improved.

In Panel B, we repeat the analysis, but only for stocks whose market capitalization is about the 25th percentile at the time of portfolio formation.  This selection further mitigates liquidity concerns. Throughout panel B, we see lower average returns than in Panel A. However, the standard deviation is also lower, so that Sharpe ratios are still fairly large and often above one per year. This again confirms previous findings that the majority of cross-sectional predictability is concentrated in small firms. As we turn to the relative comparison of the three strategies, we find that the differences are now becoming more pronounced. The differences in Sharpe ratios now range between 0.2 and 0.3 per year. Interestingly, in this analysis the Naive-FDR portfolio sometimes even achieves a higher Sharpe ratio than the Highest-$K$ portfolio. Throughout the FCI-FDR portfolio achieves the highest average returns. Since FCI-FDR portfolio is not directly targeting standard deviation, it is not surprising that the standard deviations are sometimes higher than those of the other two portfolios. However, the FCI-FDR portfolio also generally displays lower drawdown, which suggests that the higher standard deviation is not due to a disproportional increase in downside risk.

In the following analysis, we compare the risk-adjusted returns of the strategies. In Table \ref{tab:alpha_best_ideas} we show the monthly alphas in percentage for the FCI-FDR, Highest-$K$ and the Naive-FDR strategies. We again split the analysis between using all firms and only firms above the 25th quantile of market equity. The results are consistent with and even stronger than those in Table \ref{tab:perf_best_ideas}. The difference in alphas between the FCI-FDR and the Highest-$K$ (and the Naive-FDR) is larger than the return spreads.  It is noteworthy that often the alphas get larger when we add the momentum factor. For example, when we move from the \cite{fama1993common} three-factor model (FF3) to the FF4 (FF3 + momentum) model, the alpha increases. This happens because many of these portfolios have negative exposure to the momentum factor, which in turn has a positive risk premium. For the case of all firms (Panel A), we see that the difference in alpha between the FCI-FDR portfolio and the Highest-$K$ is of the same magnitude as the difference in returns. 

For the case of the larger firms in Panel B of Table \ref{tab:alpha_best_ideas}, we see that the level of alphas is lower than in the upper panel. However, many alphas are larger than 1\% per month and still strongly significant -- all alphas in Table \ref{tab:alpha_best_ideas} pass the $t$-stat 3 hurdle proposed by \cite{harvey2016and}. In all cases, the FCI-FDR strategy achieves higher alphas than the highest-$K$ strategy. The difference is more pronounced for larger portfolios. The results therefore underscore the importance of accounting for forecast uncertainties in constructing portfolios.

\begin{table}[H]
\caption{Performance Statistics for Best $K$ Portfolio (conservative SE)} \label{tab:perf_best_ideas}
\begin{singlespace}
\begin{footnotesize}
This table shows the annualized mean, annualized standard deviation, annualized Sharpe ratio, Sortino ratio, the maximum drawdown, the best and worst months for the following portfolio: ``FCI-FDR'' is the portfolio that  longs the $K$ stocks with the highest predicted expected excess returns that are significantly greater than zero. The ``Highest-$K$'' portfolio is to long the $K$ largest predicted returns. The ``Naive-FDR'' is the portfolio of the $K$ largest predicted returns for which the average realized returns are significantly greater than zero. In all cases we use the False Discovery Rate (FDR) as a multiple testing adjustment, as outlined in Section \ref{mult_hypothesis}. We vary the size of the portfolio between 100, 200 and 500. The upper panel of the table shows the results when all stocks are potential candidates for the portfolios. The lower panel considers stocks   whose market capitalization is larger than the 25\% quantile of firms' market capitalization at the time of portfolio formation. The standard error used in FCI-FDR is the maximum of the closed-form ML approximation and the bootstrap standard error using an estimation window of 240 months. The portfolios are equal weighted and all results are out-of-sample for the period from January 1983 - December 2021.
\end{footnotesize}
\end{singlespace}
\resizebox{\ifdim\width>\linewidth\linewidth\else\width\fi}{!}{
\begin{tabular}[t]{lrrrrrrrr}
\toprule
\multicolumn{1}{c}{} & \multicolumn{1}{c}{\# Stocks ($K$)} & \multicolumn{1}{c}{Mean (\%)} & \multicolumn{1}{c}{\makecell[c]{Standard\\ Deviation (\%)}} & \multicolumn{1}{c}{\makecell[c]{Sharpe\\ Ratio}} & \multicolumn{1}{c}{\makecell[c]{Sortino\\ Ratio}} & \multicolumn{1}{c}{\makecell[c]{Maximum\\ Drawdown}} & \multicolumn{1}{c}{\makecell[c]{Best\\ Month}} & \multicolumn{1}{c}{\makecell[c]{Worst\\ Month}} \\
\cmidrule(l{3pt}r{3pt}){2-2} \cmidrule(l{3pt}r{3pt}){3-3} \cmidrule(l{3pt}r{3pt}){4-4} \cmidrule(l{3pt}r{3pt}){5-5} \cmidrule(l{3pt}r{3pt}){6-6} \cmidrule(l{3pt}r{3pt}){7-7} \cmidrule(l{3pt}r{3pt}){8-8} \cmidrule(l{3pt}r{3pt}){9-9}
\addlinespace[0.3em]
\multicolumn{9}{l}{\textbf{Panel A: All Firms}}\\
\hspace{1em}FCI-FDR & 100 & 67.33 & 44.48 & 1.51 & 1.30 & 52.16 & 100.70 & -26.44\\
\hspace{1em}Highest-$K$ & 100 & 66.98 & 48.80 & 1.37 & 1.13 & 55.54 & 100.66 & -28.28\\
\hspace{1em}Naive-FDR & 100 & 41.39 & 36.13 & 1.15 & 0.77 & 60.24 & 89.25 & -30.44\\
\hspace{1em}FCI-FDR & 200 & 55.19 & 37.98 & 1.45 & 1.23 & 49.71 & 96.61 & -26.31\\
\hspace{1em}Highest-$K$ & 200 & 52.67 & 42.25 & 1.25 & 0.95 & 58.04 & 96.08 & -29.48\\
\hspace{1em}Naive-FDR & 200 & 33.28 & 29.36 & 1.13 & 0.68 & 57.26 & 74.03 & -30.04\\
\hspace{1em}FCI-FDR & 500 & 44.05 & 31.46 & 1.40 & 1.10 & 42.91 & 81.51 & -26.16\\
\hspace{1em}Highest-$K$ & 500 & 38.52 & 33.01 & 1.17 & 0.78 & 57.00 & 85.35 & -28.42\\
\hspace{1em}Naive-FDR & 500 & 25.79 & 22.67 & 1.14 & 0.60 & 52.18 & 46.27 & -27.76\\
\addlinespace[0.3em]
\multicolumn{9}{l}{\textbf{Panel B: Firm greater than 25\% market equity}}\\
\hspace{1em}FCI-FDR & 100 & 29.16 & 26.69 & 1.09 & 0.64 & 50.07 & 65.14 & -29.18\\
\hspace{1em}Highest-$K$ & 100 & 26.91 & 33.18 & 0.81 & 0.45 & 62.11 & 75.89 & -35.82\\
\hspace{1em}Naive-FDR & 100 & 23.89 & 25.94 & 0.92 & 0.46 & 57.19 & 51.80 & -31.60\\
\hspace{1em}FCI-FDR & 200 & 27.16 & 24.10 & 1.13 & 0.62 & 48.51 & 41.15 & -28.83\\
\hspace{1em}Highest-$K$ & 200 & 23.71 & 27.73 & 0.85 & 0.45 & 60.37 & 58.67 & -34.36\\
\hspace{1em}Naive-FDR & 200 & 21.65 & 22.54 & 0.96 & 0.46 & 54.93 & 36.25 & -29.61\\
\hspace{1em}FCI-FDR & 500 & 25.20 & 23.00 & 1.10 & 0.58 & 47.26 & 41.15 & -28.83\\
\hspace{1em}Highest-$K$ & 500 & 20.79 & 22.34 & 0.93 & 0.45 & 56.22 & 35.76 & -29.08\\
\hspace{1em}Naive-FDR & 500 & 19.28 & 19.79 & 0.97 & 0.45 & 51.01 & 21.46 & -27.50\\
\bottomrule
\end{tabular}}
\end{table}

\begin{table}[H]
\caption{Alpha for Best $K$ Portfolio (conservative SE)} \label{tab:alpha_best_ideas}
\begin{singlespace}
\begin{footnotesize}
This table shows monthly alphas in percentage of regression of the excess returns of the three strategies onto risk factors. CAPM represents the CAPM, FF3 refers to the \cite{fama1993common} three-factor model, FF4 denotes the FF3 model, augmented with the momentum factor (\cite{carhart1997persistence}), FF5 is the \cite{fama2016dissecting}, FF6 adds the momentum factor to the five-factor model, FF6+ adds the short-term reversal factor. All data are obtained from Kenneth French's data library. ``FCI-FDR'' is the portfolio that longs the $K$ stocks with the highest predicted expected returns that are significantly greater than zero. The ``Highest-$K$'' longs  the $K$ largest predicted returns. The ``Naive-FDR'' is the portfolio of the $K$ largest predicted returns for which the average realized returns are significantly greater than zero. In all cases we use the False Discovery Rate (FDR) as a multiple testing adjustment. The procedure is outlined in Section \ref{mult_hypothesis}. We vary the size of the portfolio between 100, 200 and 500. The upper panel of the table shows the results when all stocks are potential candidates for the portfolios. The lower panel only considers stocks whose market capitalization is larger than the 25\% quantile of firms' market capitalization at the time of portfolio formation. The standard error  used in FCI-FDR  is the maximum of the closed-form ML approximation and the bootstrap standard error using an estimation window of 240 months. The portfolios are equally weighted and all results are out-of-sample for the period from January 1983 - December 2021.
\end{footnotesize}
\vspace{-0.5cm}
\end{singlespace}
\begin{center}
\begin{footnotesize}
\begin{tabular}[t]{lrllllll}
\toprule
\multicolumn{1}{c}{} & \multicolumn{1}{c}{\# Stocks ($K$)} & \multicolumn{1}{c}{CAPM} & \multicolumn{1}{c}{FF3} & \multicolumn{1}{c}{FF4} & \multicolumn{1}{c}{FF5} & \multicolumn{1}{c}{FF6} & \multicolumn{1}{c}{FF6+} \\
\cmidrule(l{3pt}r{3pt}){2-2} \cmidrule(l{3pt}r{3pt}){3-3} \cmidrule(l{3pt}r{3pt}){4-4} \cmidrule(l{3pt}r{3pt}){5-5} \cmidrule(l{3pt}r{3pt}){6-6} \cmidrule(l{3pt}r{3pt}){7-7} \cmidrule(l{3pt}r{3pt}){8-8}
\addlinespace[0.3em]
\multicolumn{8}{l}{\textbf{All Firms}}\\
\hspace{1em}FCI-FDR & 100 & 4.53*** & 4.61*** & 5.24*** & 5.11*** & 5.58*** & 5.38***\\
\hspace{1em}Highest-$K$ & 100 & 4.37*** & 4.48*** & 5.20*** & 5.04*** & 5.57*** & 5.37***\\
\hspace{1em}Naive-FDR & 100 & 2.40*** & 2.44*** & 2.93*** & 2.76*** & 3.13*** & 2.94***\\
\hspace{1em}FCI-FDR & 200 & 3.62*** & 3.65*** & 4.23*** & 4.00*** & 4.43*** & 4.27***\\
\hspace{1em}Highest-$K$ & 200 & 3.25*** & 3.33*** & 3.98*** & 3.76*** & 4.24*** & 4.05***\\
\hspace{1em}Naive-FDR & 200 & 1.81*** & 1.83*** & 2.21*** & 2.07*** & 2.35*** & 2.20***\\
\hspace{1em}FCI-FDR & 500 & 2.76*** & 2.75*** & 3.24*** & 2.95*** & 3.32*** & 3.19***\\
\hspace{1em}Highest-$K$ & 500 & 2.21*** & 2.24*** & 2.73*** & 2.56*** & 2.93*** & 2.77***\\
\hspace{1em}Naive-FDR & 500 & 1.28*** & 1.26*** & 1.45*** & 1.34*** & 1.48*** & 1.38***\\
\addlinespace[0.3em]
\multicolumn{8}{l}{\textbf{Firm greater than 25\% market equity}}\\
\hspace{1em}FCI-FDR & 100 & 1.55*** & 1.48*** & 1.81*** & 1.50*** & 1.75*** & 1.62***\\
\hspace{1em}Highest-$K$ & 100 & 1.09*** & 1.09*** & 1.62*** & 1.37*** & 1.76*** & 1.59***\\
\hspace{1em}Naive-FDR & 100 & 0.98*** & 0.97*** & 1.22*** & 1.08*** & 1.26*** & 1.14***\\
\hspace{1em}FCI-FDR & 200 & 1.41*** & 1.33*** & 1.58*** & 1.28*** & 1.47*** & 1.36***\\
\hspace{1em}Highest-$K$ & 200 & 0.93*** & 0.91*** & 1.27*** & 1.05*** & 1.33*** & 1.19***\\
\hspace{1em}Naive-FDR & 200 & 0.87*** & 0.83*** & 0.98*** & 0.86*** & 0.97*** & 0.88***\\
\hspace{1em}FCI-FDR & 500 & 1.26*** & 1.17*** & 1.39*** & 1.10*** & 1.27*** & 1.17***\\
\hspace{1em}Highest-$K$ & 500 & 0.81*** & 0.76*** & 0.95*** & 0.78*** & 0.93*** & 0.83***\\
\hspace{1em}Naive-FDR & 500 & 0.75*** & 0.70*** & 0.76*** & 0.65*** & 0.70*** & 0.63***\\
\bottomrule
\end{tabular}
\end{footnotesize}
\end{center}
\end{table}

\section{Conclusion} \label{sec:conclusion}

The asset pricing literature has long recognized that machine learning uncertainty should be explicitly considered in making portfolio decisions. The framework of uncertainty aversion provides a disciplined way to incorporate such uncertainty about input parameters. This however requires a (asymptotic) distribution theory of the parameters. Although machine learning methods have shown great progress in forecasting expected returns, considering their estimation uncertainty in portfolio choice has remained an open problem due to a lack of distribution theory. 
 
We introduce new methods to quantify prediction uncertainty in machine learning forecasts of asset returns.  We show that neural network forecasts of expected returns share the same asymptotic distribution as classic nonparametric methods such as Fourier series, enabling a closed-form expression for their standard errors.  We also propose a $k$-step bootstrap that simulates the asymptotic distribution without repeatedly retraining networks, thus dramatically reducing computational costs.  We then incorporate these forecast confidence intervals into an uncertainty-averse investment framework. This leads to ``non-participation” in assets when uncertainty exceeds a threshold, providing an economic rationale for shrinkage for portfolio selection.  Empirically, our methods improve out-of-sample performance while mitigating multiple testing problems.

We have illustrated two areas of application; the distribution theory is general and can be used in other applications. An interesting analysis might be to consider ``real-time'' machine learning in cross-sectional asset pricing as in \cite{li2022real} to examine pre- vs.~ post publication effects. We leave this application for future research.

\newpage 

\bibliographystyle{chicago}
\bibliography{liao}
\newpage
\appendix
\newpage


\section{Additional Results} \label{sec:addtional_results}
\subsection{Figures and Tables} \label{sec:add_tabfig}

\begingroup\fontsize{8}{10}\selectfont
\begin{longtable}[t]{llll}
\caption{Overview of the Characteristics} \label{tab:chars_overview} \\
\multicolumn{4}{l}{\parbox{16cm}{This table gives an overview of the characteristic used in the empirical analysis. They are obtained from \cite{jensen2022there}. We refer to their paper and the companion website for the precise construction and reference to the original paper that proposed these predictors.}}\\
\multicolumn{4}{l}{}\\
\multicolumn{1}{c}{Acronym} & \multicolumn{1}{c}{Description} & \multicolumn{1}{c}{Category} & \multicolumn{1}{c}{In-Sample} \\
\toprule
\endfirsthead
\caption[]{Overview of the Characteristics \textit{(continued)}}\\
\multicolumn{1}{c}{Acronym} & \multicolumn{1}{c}{Description} & \multicolumn{1}{c}{Category} & \multicolumn{1}{c}{In-Sample} \\
\toprule
\endhead

age & Firm age & Intangibles & 1965 - 2001\\
aliq\_at & Liquidity of book assets & Intangibles & 1984 - 2006\\
aliq\_mat & Liquidity of market assets & Intangibles & 1984 - 2006\\
cash\_at & Cash-to-assets & Intangibles & 1972 - 2009\\
dgp\_dsale & Change gross margin minus change sales & Intangibles & 1974 - 1988\\
\addlinespace
dsale\_dinv & Change sales minus change Inventory & Intangibles & 1974 - 1988\\
dsale\_drec & Change sales minus change receivables & Intangibles & 1974 - 1988\\
kz\_index & Kaplan-Zingales index & Intangibles & 1968-1995\\
opex\_at & Operating leverage & Intangibles & 1963 - 2008\\
sale\_emp\_gr1 & Labor force efficiency & Intangibles & 1974 - 1988\\
\addlinespace
seas\_11\_15an & Years 11-15 lagged returns, annual & Intangibles & 1965 - 2002\\
seas\_11\_15na & Years 11-15 lagged returns, nonannual & Intangibles & 1965 - 2002\\
seas\_16\_20an & Years 16-20 lagged returns, annual & Intangibles & 1965 - 2002\\
seas\_16\_20na & Years 16-20 lagged returns, nonannual & Intangibles & 1965 - 2002\\
seas\_1\_1an & Year 1-lagged return, annual & Intangibles & 1965 - 2002\\
\addlinespace
seas\_1\_1na & Year 1-lagged return, nonannual & Intangibles & 1965 - 2002\\
seas\_2\_5an & Years 2-5 lagged returns, annual & Intangibles & 1965 - 2002\\
seas\_2\_5na & Years 2-5 lagged returns, nonannual & Intangibles & 1965 - 2002\\
seas\_6\_10an & Years 6-10 lagged returns, annual & Intangibles & 1965 - 2002\\
seas\_6\_10na & Years 6-10 lagged returns, nonannual & Intangibles & 1965 - 2002\\
\addlinespace
tangibility & Asset tangibility & Intangibles & 1973-2001\\
at\_gr1 & Asset Growth & Investment & 1968 - 2003\\
be\_gr1a & Change in common equity & Investment & 1962 - 2001\\
capex\_abn & Abnormal corporate investment & Investment & 1973 - 1996\\
capx\_gr1 & CAPEX growth (1 year) & Investment & 1971-1992\\
\addlinespace
capx\_gr2 & CAPEX growth (2 years) & Investment & 1976 - 1998\\
capx\_gr3 & CAPEX growth (3 years) & Investment & 1976 - 1998\\
chcsho\_12m & Net stock issues & Investment & 1970 - 2003\\
coa\_gr1a & Change in current operating assets & Investment & 1962 - 2001\\
col\_gr1a & Change in current operating liabilities & Investment & 1962 - 2001\\
\addlinespace
cowc\_gr1a & Change in current operating working capital & Investment & 1962 - 2001\\
dbnetis\_at & Net debt issuance & Investment & 1971 - 2000\\
eqnpo\_12m & Equity net payout & Investment & 1968 - 2003\\
fnl\_gr1a & Change in financial liabilities & Investment & 1962 - 2001\\
inv\_gr1 & Inventory growth & Investment & 1965 - 2009\\
\addlinespace
inv\_gr1a & Inventory change & Investment & 1970 - 1997\\
ncoa\_gr1a & Change in noncurrent operating assets & Investment & 1962 - 2001\\
ncol\_gr1a & Change in noncurrent operating liabilities & Investment & 1962 - 2001\\
nfna\_gr1a & Change in net financial assets & Investment & 1962 - 2001\\
nncoa\_gr1a & Change in net noncurrent operating assets & Investment & 1962 - 2001\\
\addlinespace
noa\_at & Net operating assets & Investment & 1964-2002\\
noa\_gr1a & Change in net operating assets & Investment & 1964-2002\\
oaccruals\_at & Operating accruals & Investment & 1962 - 1991\\
oaccruals\_ni & Percent operating accruals & Investment & 1989 - 2008\\
ppeinv\_gr1a & Change PPE and Inventory & Investment & 1970 - 2005\\
\addlinespace
taccruals\_at & Total accruals & Investment & 1962 - 2001\\
taccruals\_ni & Percent total accruals & Investment & 1989 - 2008\\
prc\_highprc\_252d & Current price to high price over last year & Momentum & 1963 - 2001\\
resff3\_12\_1 & Residual momentum t-12 to t-1 & Momentum & 1930 - 2009\\
resff3\_6\_1 & Residual momentum t-6 to t-1 & Momentum & 1930 - 2009\\
\addlinespace
ret\_12\_1 & Price momentum t-12 to t-1 & Momentum & 1965 - 1989\\
ret\_6\_1 & Price momentum t-6 to t-1 & Momentum & 1965 - 1989\\
tax\_gr1a & Tax expense surprise & Momentum & 1977 - 2006\\
corr\_1260d & Market correlation & New & 1925 - 2015\\
mispricing\_mgmt & Mispricing factor: Management & New & 1967 - 2013\\
\addlinespace
mispricing\_perf & Mispricing factor: Performance & New & 1967 - 2013\\
ni\_be & Return on equity & New & 1979 - 1993\\
ocf\_at & Operating cash flow to assets & New & 1990 - 2015\\
ocf\_at\_chg1 & Change in operating cash flow to assets & New & 1990 - 2015\\
qmj\_prof & Quality minus Junk: Profitability & New & 1957 - 2016\\
\addlinespace
qmj\_safety & Quality minus Junk: Safety & New & 1957 - 2016\\
ret\_12\_7 & Price momentum t-12 to t-7 & New & 1925 - 2010\\
ret\_3\_1 & Price momentum t-3 to t-1 & New & 1965 - 1989\\
ret\_9\_1 & Price momentum t-9 to t-1 & New & 1965 - 1989\\
rmax5\_21d & Highest 5 days of return & New & 1993 - 2012\\
\addlinespace
rmax5\_rvol\_21d & Highest 5 days of return scaled by volatility & New & 1925-2015\\
at\_be & Book leverage & Profitability & 1963 - 1990\\
at\_turnover & Capital turnover & Profitability & 1979 - 1993\\
cop\_at & Cash-based operating profits-to-book assets & Profitability & 1967 - 2016\\
cop\_atl1 & Cash-based operating profits-to-lagged book assets & Profitability & 1963 - 2014\\
\addlinespace
ebit\_bev & Return on net operating assets & Profitability & 1984 - 2002\\
ebit\_sale & Profit margin & Profitability & 1984 - 2002\\
f\_score & Pitroski F-score & Profitability & 1976 - 1996\\
gp\_at & Gross profits-to-assets & Profitability & 1963 - 2010\\
gp\_atl1 & Gross profits-to-lagged assets & Profitability & 1967 - 2016\\
\addlinespace
o\_score & Ohlson O-score & Profitability & 1981 - 1995\\
op\_at & Operating profits-to-book assets & Profitability & 1963 - 2013\\
op\_atl1 & Operating profits-to-lagged book assets & Profitability & 1963 - 2014\\
ope\_be & Operating profits-to-book equity & Profitability & 1963 - 2013\\
ope\_bel1 & Operating profits-to-lagged book equity & Profitability & 1967 - 2016\\
\addlinespace
pi\_nix & Taxable income-to-book income & Profitability & 1973-2000\\
sale\_bev & Asset turnover & Profitability & 1984 - 2002\\
ami\_126d & Amihud Measure & Trading Frictions & 1964 - 1997\\
beta\_60m & Market Beta & Trading Frictions & 1935 - 1968\\
beta\_dimson\_21d & Dimson beta & Trading Frictions & 1955 - 1974\\
\addlinespace
betabab\_1260d & Frazzini-Pedersen market beta & Trading Frictions & 1926 - 2012\\
betadown\_252d & Downside beta & Trading Frictions & 1963 - 2001\\
bidaskhl\_21d & The high-low bid-ask spread & Trading Frictions & 1927 - 2006\\
coskew\_21d & Coskewness & Trading Frictions & 1963 - 1993\\
dolvol\_126d & Dollar trading volume & Trading Frictions & 1966 - 1995\\
\addlinespace
dolvol\_var\_126d & Coefficient of variation for dollar trading volume & Trading Frictions & 1966 - 1995\\
iskew\_capm\_21d & Idiosyncratic skewness from the CAPM & Trading Frictions & 1967 - 2016\\
iskew\_ff3\_21d & Idiosyncratic skewness from the Fama-French 3-factor model & Trading Frictions & 1925-2012\\
ivol\_capm\_21d & Idiosyncratic volatility from the CAPM (21 days) & Trading Frictions & 1967 - 2016\\
ivol\_capm\_252d & Idiosyncratic volatility from the CAPM (252 days) & Trading Frictions & 1976 - 1997\\
\addlinespace
ivol\_ff3\_21d & Idiosyncratic volatility from the Fama-French 3-factor model & Trading Frictions & 1963 - 2000\\
market\_equity & Market Equity & Trading Frictions & 1926 - 1975\\
prc & Price per share & Trading Frictions & 1940 - 1978\\
ret\_1\_0 & Short-term reversal & Trading Frictions & 1929 - 1982\\
rmax1\_21d & Maximum daily return & Trading Frictions & 1962 -2005\\
\addlinespace
rskew\_21d & Total skewness & Trading Frictions & 1925-2012\\
rvol\_21d & Return volatility & Trading Frictions & 1963 - 2000\\
turnover\_126d & Share turnover & Trading Frictions & 1963 - 1991\\
turnover\_var\_126d & Coefficient of variation for share turnover & Trading Frictions & 1966 - 1995\\
zero\_trades\_126d & Number of zero trades with turnover as tiebreaker (6 months) & Trading Frictions & 1963 - 2003\\
\addlinespace
zero\_trades\_21d & Number of zero trades with turnover as tiebreaker (1 month) & Trading Frictions & 1963 - 2003\\
zero\_trades\_252d & Number of zero trades with turnover as tiebreaker (12 months) & Trading Frictions & 1963 - 2003\\
at\_me & Assets-to-market & Value & 1963 - 1990\\
be\_me & Book-to-market equity & Value & 1973 - 1984\\
bev\_mev & Book-to-market enterprise value & Value & 1962 - 2001\\
\addlinespace
debt\_me & Debt-to-market & Value & 1948 - 1979\\
div12m\_me & Dividend yield & Value & 1940 -1980\\
ebitda\_mev & Ebitda-to-market enterprise value & Value & 1963 - 2009\\
eq\_dur & Equity duration & Value & 1962 - 1998\\
fcf\_me & Free cash flow-to-price & Value & 1963 - 1990\\
\addlinespace
ival\_me & Intrinsic value-to-market & Value & 1975 - 1993\\
netdebt\_me & Net debt-to-price & Value & 1962 - 2001\\
ni\_me & Earnings-to-price & Value & 1963-1979\\
ocf\_me & Operating cash flow-to-market & Value & 1973 - 1997\\
ret\_60\_12 & Long-term reversal & Value & 1926 - 1982\\
\addlinespace
sale\_gr1 & Sales Growth (1 year) & Value & 1968 - 1989\\
sale\_gr3 & Sales Growth (3 years) & Value & 1968 - 1989\\
sale\_me & Sales-to-market & Value & 1979-1991\\*
\end{longtable}
\endgroup{}
\begin{table}[H]
\caption{Performance Statistics for Mean-Variance and Max-Min Portfolio (analytical SE)} \label{tab:ua_ind_analytical}
\begin{singlespace}
\begin{footnotesize}
This table shows the annualized mean, annualized standard deviation, annualized Sharpe ratio, Sortino ratio, the maximum drawdown, the best and worst months for the mean-variance efficient portfolio (MVE), the global minimum variance portfolio (GMVP), an equal-weighted portfolio and the UA-MV portfolios for the 25\%, 50\% and 75\% confidence level. The standard error is obtained from the closed-form ML approximation using an estimation window of 240 months. All results are out-of-sample for the period from January 1983 - December 2021.
\end{footnotesize}
\end{singlespace}
\resizebox{\ifdim\width>\linewidth\linewidth\else\width\fi}{!}{
\begin{tabular}[t]{lrrrrrrr}
\toprule
\multicolumn{1}{c}{} & \multicolumn{1}{c}{Mean (\%)} & \multicolumn{1}{c}{\makecell[c]{Standard\\ Deviation (\%)}} & \multicolumn{1}{c}{\makecell[c]{Sharpe\\ Ratio}} & \multicolumn{1}{c}{\makecell[c]{Sortino\\ Ratio}} & \multicolumn{1}{c}{\makecell[c]{Maximum\\ Drawdown}} & \multicolumn{1}{c}{\makecell[c]{Best\\ Month}} & \multicolumn{1}{c}{\makecell[c]{Worst\\ Month}} \\
\cmidrule(l{3pt}r{3pt}){2-2} \cmidrule(l{3pt}r{3pt}){3-3} \cmidrule(l{3pt}r{3pt}){4-4} \cmidrule(l{3pt}r{3pt}){5-5} \cmidrule(l{3pt}r{3pt}){6-6} \cmidrule(l{3pt}r{3pt}){7-7} \cmidrule(l{3pt}r{3pt}){8-8}
\addlinespace[0.3em]
\multicolumn{8}{l}{\textbf{500 largest first with 240 months history}}\\
\hspace{1em}MVE & 30.37 & 26.32 & 1.15 & 0.57 & 68.49 & 36.26 & -28.30\\
\hspace{1em}GMVP & 23.88 & 31.52 & 0.76 & 0.35 & 74.82 & 39.67 & -37.20\\
\hspace{1em}EW & 10.54 & 15.29 & 0.69 & 0.30 & 53.58 & 14.18 & -24.29\\
\hspace{1em}UA 25 & 19.12 & 16.38 & 1.17 & 0.58 & 49.87 & 23.38 & -18.11\\
\hspace{1em}UA 50 & 11.21 & 9.21 & 1.22 & 0.65 & 28.02 & 13.02 & -10.41\\
\hspace{1em}UA 75 & 5.07 & 4.78 & 1.06 & 0.56 & 11.76 & 8.03 & -4.58\\
\addlinespace[0.3em]
\multicolumn{8}{l}{\textbf{750 largest first with 240 months history}}\\
\hspace{1em}MVE & 31.52 & 29.44 & 1.07 & 0.51 & 71.14 & 37.16 & -38.59\\
\hspace{1em}GMVP & 23.58 & 32.72 & 0.72 & 0.32 & 78.23 & 40.08 & -50.49\\
\hspace{1em}EW & 10.83 & 16.01 & 0.68 & 0.29 & 53.80 & 18.44 & -25.72\\
\hspace{1em}UA 25 & 20.84 & 18.70 & 1.11 & 0.57 & 51.51 & 31.54 & -23.86\\
\hspace{1em}UA 50 & 12.36 & 11.56 & 1.07 & 0.65 & 29.57 & 34.20 & -11.20\\
\hspace{1em}UA 75 & 6.73 & 8.39 & 0.80 & 0.75 & 14.13 & 39.80 & -4.17\\
\addlinespace[0.3em]
\multicolumn{8}{l}{\textbf{1000 largest first with 240 months history}}\\
\hspace{1em}MVE & 33.89 & 29.23 & 1.16 & 0.56 & 68.42 & 37.16 & -37.57\\
\hspace{1em}GMVP & 24.83 & 32.33 & 0.77 & 0.34 & 84.91 & 36.92 & -45.10\\
\hspace{1em}EW & 11.12 & 16.32 & 0.68 & 0.29 & 55.78 & 19.83 & -25.72\\
\hspace{1em}UA 25 & 23.48 & 18.26 & 1.29 & 0.67 & 45.24 & 31.54 & -23.86\\
\hspace{1em}UA 50 & 14.64 & 11.09 & 1.32 & 0.87 & 29.57 & 34.20 & -11.20\\
\hspace{1em}UA 75 & 7.70 & 8.24 & 0.93 & 0.90 & 14.13 & 39.80 & -4.40\\
\bottomrule
\end{tabular}}
\end{table}
\begin{table}[H]
\caption{Performance Statistics for Mean-Variance and Max-Min Portfolio (bootstrap SE)} \label{tab:ua_ind_bootstrap}
\begin{singlespace}
\begin{footnotesize}
This table shows the annualized mean, annualized standard deviation, annualized Sharpe ratio, Sortino ratio, the maximum drawdown, the best and worst months for the mean-variance efficient portfolio (MVE), the global minimum variance portfolio (GMVP), an equal-weighted portfolio and the UA-MV portfolios for the 25\%, 50\% and 75\% confidence level. The standard error is obtained from the $k$-step bootstrap. All results are out-of-sample for the period from January 1983 - December 2021.
\end{footnotesize}
\end{singlespace}
\resizebox{\ifdim\width>\linewidth\linewidth\else\width\fi}{!}{
\begin{tabular}[t]{lrrrrrrr}
\toprule
\multicolumn{1}{c}{} & \multicolumn{1}{c}{Mean (\%)} & \multicolumn{1}{c}{\makecell[c]{Standard\\ Deviation (\%)}} & \multicolumn{1}{c}{\makecell[c]{Sharpe\\ Ratio}} & \multicolumn{1}{c}{\makecell[c]{Sortino\\ Ratio}} & \multicolumn{1}{c}{\makecell[c]{Maximum\\ Drawdown}} & \multicolumn{1}{c}{\makecell[c]{Best\\ Month}} & \multicolumn{1}{c}{\makecell[c]{Worst\\ Month}} \\
\cmidrule(l{3pt}r{3pt}){2-2} \cmidrule(l{3pt}r{3pt}){3-3} \cmidrule(l{3pt}r{3pt}){4-4} \cmidrule(l{3pt}r{3pt}){5-5} \cmidrule(l{3pt}r{3pt}){6-6} \cmidrule(l{3pt}r{3pt}){7-7} \cmidrule(l{3pt}r{3pt}){8-8}
\addlinespace[0.3em]
\multicolumn{8}{l}{\textbf{500 largest first with 240 months history}}\\
\hspace{1em}MVE & 30.37 & 26.32 & 1.15 & 0.57 & 68.49 & 36.26 & -28.30\\
\hspace{1em}GMVP & 23.88 & 31.52 & 0.76 & 0.35 & 74.82 & 39.67 & -37.20\\
\hspace{1em}EW & 10.54 & 15.29 & 0.69 & 0.30 & 53.58 & 14.18 & -24.29\\
\hspace{1em}UA 25 & 21.73 & 18.84 & 1.15 & 0.56 & 56.70 & 25.60 & -22.33\\
\hspace{1em}UA 50 & 14.36 & 13.11 & 1.10 & 0.53 & 43.21 & 15.36 & -16.72\\
\hspace{1em}UA 75 & 8.36 & 8.21 & 1.02 & 0.51 & 27.99 & 9.44 & -10.40\\
\addlinespace[0.3em]
\multicolumn{8}{l}{\textbf{750 largest first with 240 months history}}\\
\hspace{1em}MVE & 31.52 & 29.44 & 1.07 & 0.51 & 71.14 & 37.16 & -38.59\\
\hspace{1em}GMVP & 23.58 & 32.72 & 0.72 & 0.32 & 78.23 & 40.08 & -50.49\\
\hspace{1em}EW & 10.83 & 16.01 & 0.68 & 0.29 & 53.80 & 18.44 & -25.72\\
\hspace{1em}UA 25 & 23.78 & 20.85 & 1.14 & 0.56 & 61.43 & 27.37 & -31.66\\
\hspace{1em}UA 50 & 16.73 & 15.20 & 1.10 & 0.57 & 49.40 & 28.34 & -22.69\\
\hspace{1em}UA 75 & 10.96 & 11.80 & 0.93 & 0.62 & 31.12 & 40.69 & -14.47\\
\addlinespace[0.3em]
\multicolumn{8}{l}{\textbf{1000 largest first with 240 months history}}\\
\hspace{1em}MVE & 33.89 & 29.23 & 1.16 & 0.56 & 68.42 & 37.16 & -37.57\\
\hspace{1em}GMVP & 24.83 & 32.33 & 0.77 & 0.34 & 84.91 & 36.92 & -45.10\\
\hspace{1em}EW & 11.12 & 16.32 & 0.68 & 0.29 & 55.78 & 19.83 & -25.72\\
\hspace{1em}UA 25 & 26.40 & 20.33 & 1.30 & 0.64 & 58.38 & 27.37 & -31.48\\
\hspace{1em}UA 50 & 19.39 & 14.51 & 1.34 & 0.71 & 46.44 & 28.34 & -21.86\\
\hspace{1em}UA 75 & 13.05 & 11.34 & 1.15 & 0.82 & 29.40 & 40.69 & -14.57\\
\bottomrule
\end{tabular}}
\end{table}


\subsection{Tuning Parameters} \label{sec:tuning}
Table \ref{tab:tune} summarizes the tuning parameters for the three-layer feedforward neural network. For the Fourier series regressions we set the Fourier basis functions as $(\sin(j\pi x/4), \cos(j\pi x/4))$ for $j=1,...,J$. 

\begin{table}[H]
\caption{Tuning Parameters for the Neural Network and Fourier Basis} \label{tab:tune}
\begin{center}
\begin{tabular}[t]{ll}
\hline
Parameter & Values\\
\hline
L2 Penalty & $\lambda_2 \in \{1e^{-5}, 1e^{-3}\}$\\
Learning Rate & $\eta = 0.001$\\
Batch Size & 10,000\\
Epochs & 100\\
Activation Function & ReLu\\
Algorithm & Adam (default parameters)\\
Num. Fourier Basis & $J=3$\\
\hline
\end{tabular}
\end{center}
\end{table}

\subsection{The UA-MV Portfolio without Risk-free Rate} \label{sec:ua_mult}
We now consider the problem where  the investor would like to allocate her portfolio into $R$ risky assets without risk-free rate. Then   $z_{T+1|T,i}$ for $i=1,...,R$ denote the expected returns. To incorporate her uncertainty constraints for forecasting the expected returns, the investor then formulates a multivariate, constrained, mean-variance problem:
$$
\max_{\bomega=(\omega_1,...,\omega_R)} \min_{\mu_1,...,\mu_R} \sum_{i=1}^R\omega_i \mu_i -\frac{\gamma}{2}\bomega' \Sigma_T\bomega
$$
subject to the constraint on $(\bomega, \mu_1,...,\mu_R)$:
$$
 \sum_{j=1}^R\omega_j =1,\quad |\mu_i-\widehat z_{T+1|T,i}|\leq q_{\alpha,i} ,i=1,...,R.
$$
 
To see how the investor's uncertainty aversion determines the portfolio shrinkage, we focus on the case when there are two risky assets ($R=2$), as this allows us to gain the most intuition. Since the expressions of the optimal portfolios $\omega_1^*$ and $\omega_2^*$ are symmetric, we will focus on $\omega_1^*$. The optimal portfolio has a closed-form solution: for $c_0 = \gamma^{-1}\Var(z_{T+1,1} - z_{T+1,2})^{-1}$ (see Section \ref{secb.4} for the proof).
\begin{equation}\label{eq5.3}
\omega_1^*=\begin{cases}
0 & \text{when }-c_0(q_{\alpha,1} + q_{\alpha,2})<\omega_1^{\MV}<c_0(q_{\alpha,k} - q_{\alpha,j})\cr 
\omega_1^{\MV} +c_0(q_{\alpha,1} + q_{\alpha,2}) & \text{when } \omega_1^{\MV}<-c_0(q_{\alpha,1} + q_{\alpha,2})\cr 
 \omega_1^{\MV} -c_0(q_{\alpha,k} - q_{\alpha,j}) & \text{when } \omega_1^{\MV}>c_0(q_{\alpha,1} - q_{\alpha,2}) \text{ and } \omega_2^{\MV}>c_0(q_{\alpha,2} - q_{\alpha,1})
\end{cases}.
\end{equation}
Here $\omega_1^{\MV}$ is the usual mean-variance efficient portfolio: $\omega_2^{\MV} = 1-\omega_1^{\MV}$ and 
\begin{eqnarray*}
\omega_1^{\MV}&:=&\arg\max_{\omega\in\mathbb R}\omega\widehat z_{T+1|T,1} + (1-
\omega)\widehat z_{T+1|T,2} -\frac{\gamma}{2}\bomega'\Sigma_T\bomega,\quad \bomega= (\omega, 1-\omega)'.
\end{eqnarray*}
There are three cases, the investor holds no position in the first asset, the investor holds a short-position in the first asset and the investors invests a positive amount in the first asset. To explain the economic intuition, let us compare the behavior of two investors: one MV-investor, who adopts  $\omega_1^{\MV}$; and one UA-investor, who is uncertainty averse and adopts $\omega_1^*$.

Consider the first case of (\ref{eq5.3}). The ``non-participation" position  $\omega_1^*=0$ appears    when 
\begin{equation}\label{eq5.4}
 |\omega_1^{\MV} +q_{\alpha,2} c_0| \leq q_{\alpha,1}c_0.
\end{equation}
Recall that $q_{\alpha,k}$ measures the UA-investor's degree of uncertainty aversion. The higher $q_{\alpha,k}$, the more cautious she is when allocating to asset $k$, and vice versa. Suppose $q_{\alpha,2} \to 0$ and is much smaller than $q_{\alpha,1}$. When (\ref{eq5.4}) holds, then
$$
|\omega_1^{\MV}|\leq q_{\alpha,1}c_0[1+ o(1)],
$$
meaning that the MV-investor will not invest  in $z_{T+1,1}$ more than  $q_{\alpha,1}c_0$. Meanwhile, as $q_{\alpha,2}$ is very small, the UA-investor has little uncertainty about the predictor $\widehat{z}_{T+1|T,2}$, so she is be better off by allocating all her assets in $z_{T+1,2}$, and thus $\omega_1^*=0$.

Next, consider the second case in \eqref{eq5.3}. In this case $\omega_1^{\MV} < 0$, meaning that the MV-investor will short the first asset. Meanwhile, the UA-investor will also short asset 1 because $\omega_1^* = \omega_1^{\MV} + c_0(q_{\alpha,1} + q_{\alpha,2}) < 0$. However, she will short less of asset 1 and shift her allocation towards asset 2, as $\omega_1^{\MV} < \omega_1^* < 0$. The amount of shrinkage is $c_0(q_{\alpha,1} + q_{\alpha,2})$. Therefore, the UA-investor is more cautious than the MV-investor when shorting.

Finally, consider the third case of (\ref{eq5.3}), then 
$$
\omega_1^*= \omega_1^{\MV} -c_0(q_{\alpha,1} - q_{\alpha,2}) ,\text{ which is }\begin{cases}
<\omega_1^{\MV} & \text{ if } q_{\alpha,1}>q_{\alpha, 2}\cr 
>\omega_1^{\MV} & \text{ if } q_{\alpha,1}<q_{\alpha, 2}.
\end{cases}
$$
When $q_{\alpha,1} > q_{\alpha,2}$, this case implies $\omega_1^{\MV} > c_0(q_{\alpha,1} - q_{\alpha,2}) > 0$, so the MV-investor will have a long position in asset 1. Meanwhile, because the UA-investor is more uncertain about $z_{T+1,1}$ than $z_{T+1,2}$, her allocation in the first asset satisfies $0 < \omega_1^* < \omega_1^{\MV}$, meaning that she shrinks her allocation toward the second asset. The amount of shrinkage, $c_0(q_{\alpha,1} - q_{\alpha,2})$, is proportional to the difference in the degrees of uncertainty. The case when $q_{\alpha,1} < q_{\alpha,2}$ follows from a similar insight due to the symmetry between $\omega_1^*$ and $\omega_2^*$.

\subsection{Risk-Sensitive Optimization} \label{sec:ba}
The UA-constrained portfolio allocation tends to adopt a non-participation, which may be overly conservative. In this subsection, we introduce a less restrictive framework that explicitly accounts for forecast uncertainty. Our approach builds on the robust decision-making framework of \cite{hansen2008robustness}.

\subsubsection{One risky and one risk-free asset}
\cite{anderson2016robust} formulate a mean-variance optimization for choosing portfolio allocations that are robust to misspecification in the predictive model. To illustrate, consider an example with one risky asset with excess return $z_{T+1}$. Their problem is formulated as:
\begin{equation}\label{eq5.5}
\max_{\omega}\min_{h_T} \omega\mathbb E_{h_T}(z_{T+1} )
 -\frac{\gamma}{2}\omega^2 \Var_{h_T}(z_{T+1}) + \frac{1}{\tau}D(h_T||f_T)
\end{equation}
where both $h_T$  and $f_T$ are probability density functions, and $\mathbb E_{h_T}(z_{T+1} )$ and $\Var_{h_T}(z_{T+1} )$   respectively denote the expectation and variance of the future risky return with respect to the distribution $h_T$; $D(h_T||f_T)$ denotes the Kulback-Leibler divergence from $f_T$ to $h_T$. Therefore, different from the UA-constrained problem, (\ref{eq5.5}) introduces an inner optimization with respect to the density function $h_T$ with an additional measure $D(h_T||f_T)$. This problem is also called ``risk-sensitive optimization" by \cite{hansen2008robustness}.

The idea is that the investor takes $f_T$ as the benchmark predictive density of $z_{T+1}$, but she is concerned that $f_T$ might be misspecified. So she considers an alternative predictive density $h_T$ for the risky return and constructs portfolio choices to maximize utility under the worst specification of $h_T$. Meanwhile, the investor also believes that $f_T$ is reasonably specified, so by introducing the penalization term $D(h_T||f_T)$, she focuses on alternatives close to $f_T$. The scalar parameter $\tau$ measures the investor's uncertainty aversion. One suggestion by \cite{anderson2016robust} is to use the benchmark predictive density \begin{equation}\label{eq4.6progress} f_T\sim\mathcal{N}(\widehat{z}_{T+1|T},\sigma_T^2), \end{equation} which directly uses the predicted return $\widehat{z}_{T+1|T}$ as the mean of the predictive density and is interpreted as ``the agent's best approximation for the distribution (of returns)."

Importantly, in the Anderson-Cheng model, the incorporated uncertainty is from the density $f_T$ of the true future return, rather than the uncertainty from the prediction $\widehat{z}_{T+1|T}$. It can be easily shown that the solution is equivalent to the MV-portfolio with an increased risk-aversion parameter $\tau+\gamma$. Hence, benchmarking against (\ref{eq4.6progress}) does not incorporate the prediction uncertainty, which could yield portfolios that are too aggressive to be robust to sudden changes in market and idiosyncratic risk.

The objective is to derive a better predictive density $f_T$ in place of (\ref{eq4.6progress}) to account for the uncertainty of the prediction $\widehat{z}_{T+1|T}$, which in our applications is the pooled ML (e.g., DNN) forecast. 
Suppose the investor has a prior distribution of the expected return:
$$
p(z_{T+1|T})\sim \mathcal N(\pi,v)
$$
where $(\pi, v)$ respectively denote the prior mean and prior variance. Meanwhile, the asymptotic distribution of the ML forecast is approximately (as proved in Theorem \ref{th2}):
$$ p(\widehat{z}_{T+1|T}|z_{T+1|T}) \sim \mathcal{N}(z_{T+1|T}, \FSE^2) $$
where $\FSE = \SE(\widehat{z}_{T+1})$ is the forecast standard error obtained using closed-form-ML or the $k$-step bootstrap. This distribution serves as the likelihood function. Applying the Bayesian updating rule, it then yields a posterior distribution of the future return $z_{T+1}$:
$$ f_T(z_{T+1}) = \int p(z_{T+1}|z_{T+1|T}) p(z_{T+1|T}) p(\widehat{z}_{T+1|T}|z_{T+1|T}) \, dz_{T+1|T}. $$
This predictive density accounts for the ML uncertainty and yields:
\begin{equation}
    f_T\sim \mathcal N(\widetilde z_{T+1}, \widetilde \sigma_T^2)
\end{equation}
where $\widetilde{z}_{T+1}$ and $\widetilde{\sigma}_T^2$ are the posterior mean and variance, respectively, given by
$$ \widetilde{z}_{T+1} = (1 - W_1)\widehat{z}_{T+1|T} + W_1 \pi, \quad \widetilde{\sigma}_T^2 = vW_1 + \sigma_T^2, \quad W_1 = \frac{\FSE^2}{\FSE^2 + v}. $$

In addition, we let  $h_T$  in the inner optimization of (\ref{eq5.5}) take the form: 
$$
h_T\sim \mathcal N(\mu,\widetilde\sigma_T^2), 
$$
where $\mu \in \mathbb{R}$ is unspecified and we search for the optimal $\mu$ in the inner optimization of (\ref{eq5.5}). This allows us to concentrate on the mean forecast as the main source of misspecification.
In this case, the Kullback-Leibler divergence for two normal distributions is simply $D(h_T||f_T) = \frac{1}{2\widetilde{\sigma}T^2}(\widetilde{z}_{T+1|T} - \mu)^2$. Therefore, the optimal portfolio can be obtained as:
$$
\omega^{\RS}:=\arg\max_{\omega}\min_{\mu} \omega \mu -\frac{\gamma}{2}\omega^2 \widetilde\sigma_T^2 +   \frac{1}{2\tau \widetilde\sigma_T^2  }(\widetilde z_{T+1|T} -\mu)^2
$$

To characterize the solution, let 
$$
\omega^{\MV} =  \frac{\widehat z_{T+1|T} }{\sigma_T^2\gamma},\quad \omega^{\MV}_{\pi}
=\frac{\pi }{\sigma_T^2\gamma}
$$
be the MV portfolios based on the predicted mean $\widehat z_{T+1|T}$ and the  prior mean $\pi$.  Then it can be shown that the solution is (see Section \ref{secb.4} for the proof)
\begin{equation}\label{eq5.6}
\omega^{\RS}=\left[\omega^{\MV}(1-W_1) + \omega^{\MV}_{\pi}W_1 \right]\frac{\sigma_T^2}{vW_1+\sigma_T^2} \frac{\gamma}{\tau+\gamma}.
\end{equation}
Therefore, when taking into account the forecast uncertainty,  $\omega^{\RS}$ conducts a \textit{double shrinkage}:   First, it shrinks  $\omega^{\MV}$ towards  $\omega^{\MV}_{\pi}$ due to the weight $W_1$. Second,  it shrinks the overall portfolio towards zero due to the factor $ \frac{\sigma_T^2}{vW_1+\sigma_T^2}$.

\begin{figure}[H]
\caption{Mean-Variance vs. Uncertainty Averse (risk-sensitive) Portfolio Allocation} \label{figmonoto}
\vspace{-2mm}
\caption*{\setstretch{1.0} \scriptsize This figure shows the uncertainty averse (risk-sensitive) solution for the optimal portfolio weight in comparison with the standard mean-variance portfolio weight. In the upper panel, we hold the level of the forecast standard error (FSE) fixed and vary the mean-variance weight ($\omega^{\MV}$). The red line, $\omega^{\text{RS}} \text{ I}$ is constructed using $\omega_{\pi}^{\MV}=\omega^{\MV}$. The blue line, $\omega^{\text{RS}} \text{ II}$, is constructed using $\omega_{\pi}^{\MV}=0.5\times \omega^{\MV}$. In the lower panel, we consider the case where $\omega^{\MV}>a \omega_{\pi}^{\MV}$ and vary the forecast standard error (FSE).}
    \includegraphics[width=1.0\textwidth]{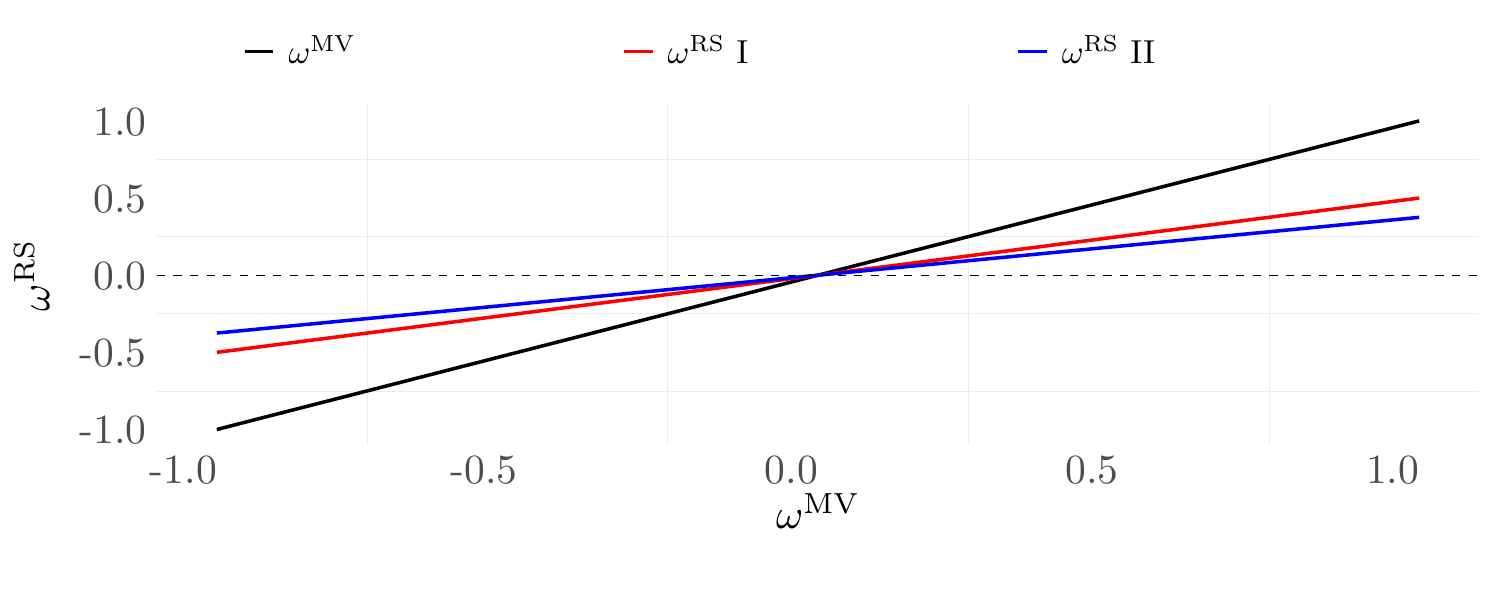}
    \includegraphics[width=1.0\textwidth]{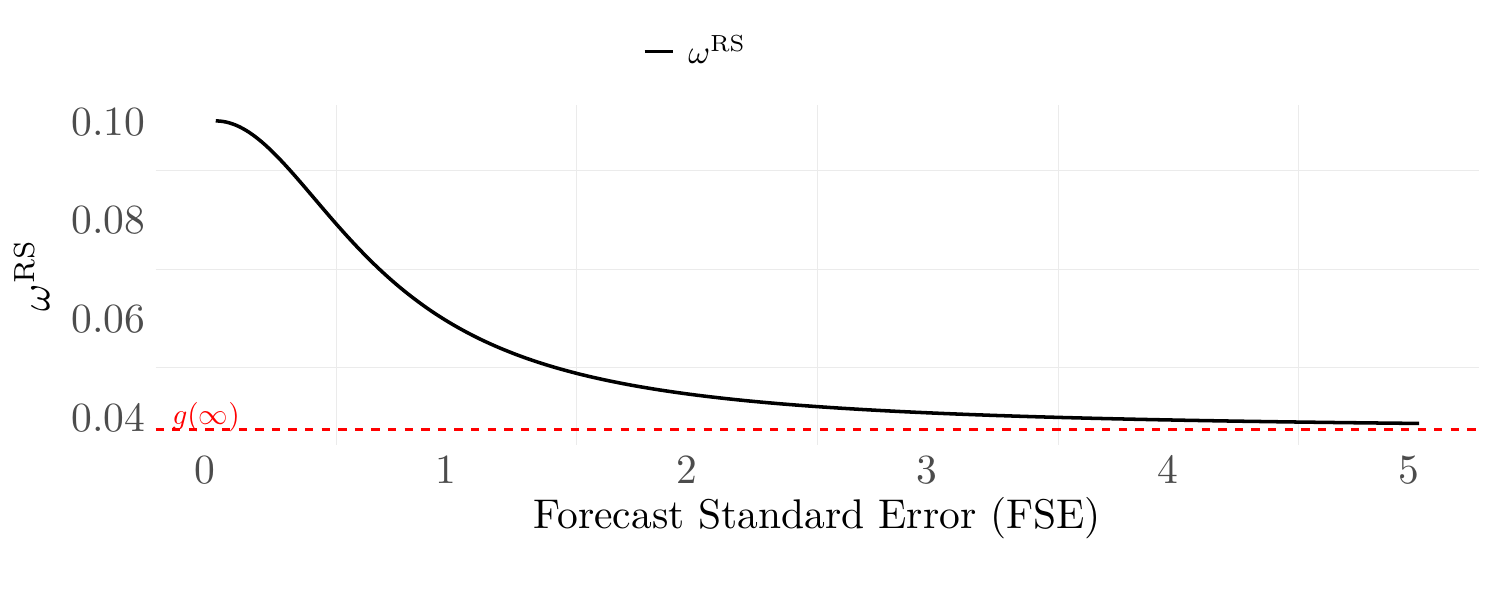}
\end{figure}

To further shed light on the effect of incorporating forecast uncertainty into the shrinkage portfolio, we recall  that $W_1 =\FSE^2/(\FSE^2+ v)$, hence $\omega^{\RS}$ can be written as an explicit function of   $\FSE$ that shrinks the mean-variance portfolio: 
$\omega^{\RS}= g(\FSE)$, where \begin{eqnarray*}
g(s)&:=&\left[\omega^{\MV}(1-W(s)) + \omega^{\MV}_{\pi}W(s) \right]\frac{\sigma_T^2}{vW(s)+\sigma_T^2} \frac{\gamma}{\tau+\gamma}\cr 
W(s)&:=& \frac{s^2}{s^2+v}.
\end{eqnarray*} 
 The monotonicity of $g(\cdot)$  depends on the relative magnitude between $\omega^{\MV}$ and $a\times \omega_{\pi}^{\MV}$, where $ a= \frac{ \sigma_T^2}{v+\sigma_T^2} .$

As an illustration, consider the case $\omega^{\MV} > a \omega_{\pi}^{\MV}$, which means that the traditional MVE position suggests holding a ``large" position in the risky asset.
The upper panel of Figure \ref{figmonoto} compares $\omega^{\RS}$ with $\omega^{\MV}$ as the latter increases with a fixed FSE. Clearly, $\omega^{\RS}$ increases linearly in $\omega^{\MV}$ but at a slower rate due to the linear shrinkage effect. The lower panel plots the $g(\cdot)$ function with respect to the forecast standard error. 
As FSE increases, the uncertainty-averse investor decreases her position in the risky asset, i.e., moving in the opposite direction relative to $\omega^{\MV}$. However, there is no region of complete non-participation. As forecast uncertainty increases, $\omega^{\RS}$ converges to a ``discounted" MV portfolio that relies solely on prior beliefs:
$$
g(\infty)=\omega_{\pi}^{\MV} \frac{\sigma_T^2}{v+\sigma_T^2}\frac{\gamma}{\tau+\gamma},
$$
Essentially, when forecast uncertainty is extremely high, the portfolio is predominantly guided by prior beliefs, minimizing the impact of the predicted expected returns.


\subsubsection{Multiple risky assets} \label{sec:mult_risky}

We now extend the framework outlined in the previous subsection to the scenario of multiple risky assets. We obtain the pooled ML forecast $\widehat{z}_{T+1|T}$ for the $R$-dimensional expected return $z_{T+1|T}$. Let $\SE^2$ be the forecast covariance matrix of $\widehat{z}_{T+1|T}$, and let $\Sigma_T$ be the $R \times R$ covariance of $z_{T+1} - z_{T+1|T}$. We impose a multivariate normal prior $z_{T+1|T} \sim \mathcal{N}(\pi, v)$, where $v$ is an $R \times R$ prior covariance matrix, set to $v = g \Sigma_T$ for some $g > 0$ as in Zellner's $g$-prior.\footnote{In the context of linear regression, \cite{zellner1986assessing} proposed a $g$-prior where the prior variance is proportional to the inverse Fisher information matrix through a constant $g$. As $g \to \infty$, the posterior mean converges to the maximum likelihood estimator.}
Then the posterior of the predictive density for the $N$-dimensional $z_{T+1}$ is $f_T(z_{T+1}) \sim \mathcal{N}(\widetilde{z}_{T+1}, \widetilde{\Sigma})$, where for $W_1 = \SE^2(\SE^2 + v)^{-1}$,
$$
\widetilde z_{T+1}= (I-W_1)\widehat z_{T+1|T} + W_1\pi, \quad \widetilde\Sigma=\Sigma_T +v W_1'.
$$
Now consider the problem: for $h_T\sim\mathcal N(\mu, \widetilde\Sigma)$: (see Section \ref{secb.4} for the  proof)
\begin{eqnarray}\label{eq4.11}
\bomega^{\RS}&:=&\arg\max_{\bomega\in\mathbb R^N}\min_{\mu\in\mathbb R^N} \bomega '\mu  -\frac{\gamma}{2}\bomega' \widetilde\Sigma\bomega +   \frac{1}{\tau   } D(h_T||f_T)\cr 
      \cr 
      &=& \frac{\gamma }{\tau+\gamma}\widetilde\Sigma^{-1}\Sigma_T\left[ (I-W_1) '\bomega^{\MV} +W_1'\bomega^{\MV}_{\pi
}\right]  
\end{eqnarray}
where 
$$
\bomega^{\MV} =  \Sigma_T^{-1}\widehat z_{T+1|T}  \frac{1}{\gamma},\quad \bomega_{\pi}^{\MV} =  \Sigma_T^{-1}\pi    \frac{1}{ \gamma}
$$
are respective the mean-variance portfolios based on $\widehat z_{T+1|T}$ and $\pi$. Hence   $\omega^{\RS}$ still  conducts a \textit{double shrinkage}:   First, it shrinks  $\bomega^{\MV}$ towards  $\bomega^{\MV}_{\pi}$ due to the weight $W_1$. Secondly,  it shrinks the overall portfolio to zero due to the $ \widetilde\Sigma^{-1}\Sigma_T$. 

Note that as $g\to\infty$, the prior becomes more diffuse and then $W_1\to 0$, hence more weights are imposed on $\bomega^{\MV}$. The resulting portfolio $\bomega^{\RS}$ becomes less robust to the forecast uncertainty in $\widehat z_{T+1|T}$. We shall see the impact of this in the empirical study. 


\section{Technical Appendix}

\subsection{Intuition}\label{sec:ti}

In this subsection we briefly explain the technical intuition on how we derive  the expression $$\widehat z_{T+1|T}  - z_{T+1|T} = \frac{1}{T}\sum_{t=1}^T\mathcal C_t+ o_P(T^{-1/2})$$
and why the leading term $\mathcal C_t$ does not depend on the specific choice of the ML space.

The Riesz-representation plays a key role in the asymptotic analysis, and has been commonly used in the inferene for semiparametric models, e.g., \cite{newey1994asymptotic, shen1997methods,chen1998sieve, chen1999improved,chernozhukov2018biased}, among many others.  The use of Riesz-representation allows to directly span the ML forecast using the least squares loss function, which also requires studying the object of interest in a Hilbert space. To do so,  define an inner product:  $$\langle h_1, h_2 \rangle:= \frac{1}{NT}\sum_{it}\mathbb E h_1(x_{i,t-1})h_2(x_{i,t-1})$$ where the expectation is taken  jointly with respect to the serial and cross-sectional distribution of $x_{i,t-1}$, treating $h_1, h_2$ as fixed functions.  
 Define
  $$
 g_{NT,\mathcal G}:=\arg\min_{h\in\mathcal G_{\DNN}\cup\mathcal G_B}\| h -g\|_{L^2}^2,
 $$
 which is the best approximation to the true $g$ on the space $\mathcal G_{\DNN}\cup\mathcal G_B$ under the norm $\|_{L^2}$. Then $\mathcal A_{NT}:= \text{span}( \mathcal G_{\DNN}\cup\mathcal G_B -\{g_{NT,\mathcal G}\} ) $   is a finite dimensional Hilbert spaced endowed with the inner product $\langle\cdot,\cdot\rangle.$   Next, evaluated at the out-of-sample characteristics $x_{i,T}$, define  a sequence of  linear functionals: 
$$\mathcal T_i(h):=  h(x_{i,T}),\quad  i=1,...,N.$$  

Because $\mathcal T_j$ is a linear functional it is always bounded  on the finite dimensional Hilbert space $\mathcal A_{NT}$. The Riesz representation theorem then implies that there is a function $m_j^*\in\mathcal A_{NT}$, called Riesz representer, so that 
$$
\mathcal T_j(h)= \langle h, m_j^*\rangle,\quad \forall h\in\mathcal A_{NT}.
 $$
 The key fact to our argument is that $m^*_j$  does not depend on the specific machine learning space  being used (whether $\mathcal G_{NT,\mathcal G}$ or $\mathcal G_B$) for   $\widehat z_{T+1|T}$. It only depends on the joint distribution of $\{x_{i,t-1}\}$, the realization $\{x_{j,T}\}$ and the union space $\mathcal G_{\DNN}\cup\mathcal G_B$.

 Next, using an  argument for M-estimation (e.g. Theorem 3.2.5 of \cite{VW}),  we show  in  Lemma \ref{lem2} that 
  uniformly for  $j\leq N$,
$$
 \langle  \widehat g-g, m^*_j\rangle= \frac{1}{NT}\sum_{it}e_{i,t}m^*_j(x_{i,t-1}) + o_P(T^{-1/2}).
$$
 Then heuristically, 
 \begin{eqnarray*}
\widehat z_{T+1|T} -  z_{T+1|T}  &=&\sum_j w_j[\widehat g(x_{j,T}) -g(x_{j,T}) ] =\sum_j w_j\mathcal T_j(\widehat g-g) \cr
&\approx&  \sum_j w_j\langle \widehat g-g, m_j^*\rangle
\cr
&\approx& \frac{1}{NT}\sum_{it}\sum_jw_je_{i,t}m^*_j(x_{i,t-1}) .
\end{eqnarray*}
This yields the desired expansion with $\mathcal C_t\approx \frac{1}{N}\sum_{i}\sum_jw_je_{i,t}m^*_j(x_{i,t-1})$. It is clear from this expression that $\mathcal C_t$ does not depend on the specific choice of the ML method.

\subsection{Assumptions} \label{sec:assumptions}
We shall use two machine learning spaces: the ``forecast ML" $\mathcal G_{\DNN}$, which is the space for deep neural networks, and  the closed-form ML $\mathcal G_B$, which is the Fourier series to compute the standard error.   Let $\mathcal N(\delta, \mathcal G, \|.\|_\infty)$ denote the entropy cover of $\mathcal G$, which is the smallest number of $\|.\|_{\infty}$-balls of radius $\delta$ to cover $\mathcal G$.

\begin{assumption}\label{ass2}  
Conditioning on $X$, $(v_t, u_{i,t})$ and  are independent over time and are sub-Gaussian.  Also suppose $x_{i,t-1}$ are independent across $i$. 
\end{assumption}

\begin{assumption}\label{ass1} The following conditions holds for $\mathcal G= \mathcal G_{\DNN}$, the neural network space.

\begin{enumerate}[label=\roman*)] 
    \item There is $p(\mathcal G)$ so that the covering number satisfies: for any $\delta>0$,  
\begin{equation}\label{asseq4.1}
\mathcal N(\delta, \mathcal G, \|.\|_\infty)\leq \left(\frac{CT}{\delta}\right)^{p(\mathcal G)}.
\end{equation}
    \item Let $$\varphi_{\mathcal G}^2 =  \inf_{v\in\mathcal G} \sup_{h\in\{\tau_T g+ \mathcal G\}\cup \{g\}}\|v-h\|_{\infty}^2.$$
Suppose $ \varphi_{\mathcal G}=o(T^{-1})$,  
$p(\mathcal G) \log(NT)  = o(T^{1/2})$
and $\varphi_{\mathcal G} p(\mathcal G)\log (NT) = o(1)$.
    \item Define  the best approximation to the true  $g$ function under the $\|.\|_{L^2}$-norm: 
    $$
    g_{NT,\mathcal G}:=\arg\min_{h\in\mathcal G}\frac{1}{NT}\sum_{it}\mathbb E  (h(x_{i,t-1})-g(x_{i,t-1} ))^2.
    $$
    Suppose $\sqrt{T}\max_{j\leq N}|g_{NT,\mathcal G}(x_{j,T})- g(x_{j,T}))|=o_P(1)$. 
\end{enumerate}
\end{assumption}  

The above assumption is well known to be satisfied by   DNN  for properly specified width and depth of the layers. \cite{schmidt2020nonparametric} show that   a multilayer feedforward network with ReLu activation functions at each layer can well approximate a rich class of functions with compact support.  It also follows from \cite{anthony2009neural} that (\ref{asseq4.1}) holds with $p(\mathcal G_{\DNN})$ as the {pseudo dimension} of the neural network, with $p(\mathcal G_{\DNN})=O(J^2L^2)$, where $J,L$ are respectively the maximum width and depth of the network. 

\begin{assumption}\label{ass1a} The following conditions holds for $\mathcal G= \mathcal G_{B}$, the Fourier series regression.

\begin{enumerate}[label=\roman*)] 
    \item There is $p(\mathcal G)$ so that the covering number satisfies: for any $\delta>0$,  
\begin{equation}\label{asseq4.1}
\mathcal N(\delta, \mathcal G, \|.\|_\infty)\leq \left(\frac{CT}{\delta}\right)^{p(\mathcal G)}.
\end{equation}
    \item Let $$\varphi_{\mathcal G}^2 =  \inf_{v\in\mathcal G} \sup_{h\in\{\tau_T g+ \mathcal G\}\cup \{g\}}\|v-h\|_{\infty}^2.$$
Suppose $ \varphi_{\mathcal G}=o(T^{-1})$,  
$p(\mathcal G) \log(NT)  = o(T^{1/2})$
and $\varphi_{\mathcal G} p(\mathcal G)\log (NT) = o(1)$.
    \item Define  the best approximation to the true  $g$ function under the $\|.\|_{L^2}$-norm: 
    $$
    g_{NT,\mathcal G}:=\arg\min_{h\in\mathcal G}\frac{1}{NT}\sum_{it}\mathbb E  (h(x_{i,t-1})-g(x_{i,t-1} ))^2.
    $$
    Suppose $\sqrt{T}\max_{j\leq N}|g_{NT,\mathcal G}(x_{j,T})- g(x_{j,T}))|=o_P(1)$. 

    \item Let $J$ be the number of Fourier basis functions. Then $J^{5/2}=o(T)$ and $(\frac{p(\mathcal G_{G}) \log (T ) }{T} + \varphi_{\mathcal G_{G} } )J^2\leq o(1)$. 
\end{enumerate}

\end{assumption}
  
To justify the above assumption, note that  the  Fouries basis  can well approximate a H\"older class of smooth functions; and   (\ref{asseq4.1})  holds with $p(\mathcal G)\sim J$ being the number of basis functions.

\subsection{Proofs}

 \subsubsection{Proof of Theorem \ref{th1}}\label{sec:p1}

For any generic function $h$, which may or may not be in $\mathcal G\in\{\mathcal G_{\DNN}, \mathcal G_1\}$,  let 

$Q(h)=\frac{1}{TN}\sum_{it}(y_{i,t}-  h(x_{i,t-1}))^2$ be the least squares loss function.  Let
 \begin{eqnarray}\label{eqrddafphi}
   \psi_{\mathcal G}(h)=  |Q(\pi_{\mathcal G}(h))-Q(h)| ,\quad \pi_{\mathcal G}(h)=\arg\min_{v\in\mathcal G} \|v- h\|_{\infty}.
\end{eqnarray}
Here $\pi_{\mathcal G}(h)$ is the projection of function $h$ onto the learning space $\mathcal G$. In particular, if $h\in\mathcal G$ then $\pi_{\mathcal G}(h)=h$.

 Theorem \ref{th1} is restated and proved by the following proposition.

\begin{proposition}\label{prop3} Suppose $\sum_i|w_i|<\infty$. 
There  are functions $m_j^*$, which is the same whether DNN or Fourier series are used as the predictor, so that for 
$\widehat z_{T+1}$ being  either the DNN predictor or the Fourier series predictor, 
\begin{eqnarray*}
\widehat z_{T+1} -  z_{T+1|T}  
&=& \frac{1}{T}\sum_{t} \frac{1}{N}\sum_i \zeta^*(x_{i,t-1})  \beta_{i,t-1}'v_t+ o_P(T^{-1/2}) \cr
\zeta^*(\cdot)&:=&\sum_j w_j m_j^*(\cdot) . 
\end{eqnarray*}
\end{proposition}

\begin{proof}

Let $\mathcal  G_{\DNN}$  be the DNN space and $\mathcal G_B$ denote the Fourier series space.
 Define 
 $
\mathcal A= \text{span}( \mathcal G_{\DNN}\cup\mathcal G_B -\{g\} ),
 $
which is the closed linear span of  $\mathcal G_{\DNN}\cup\mathcal G_B-g$ and $g$ is the true $g$ function.   Also, define
  $$
 g_{NT,\mathcal G}:=\arg\min_{h\in\mathcal G_{\DNN}\cup\mathcal G_B}\frac{1}{NT}\sum_{it}\mathbb E  (h(x_{i,t-1})-g(x_{i,t-1} ))^2,
 $$
 which is the best approximation to the true $g$ on the space $\mathcal G_{\DNN}\cup\mathcal G_B$ under the norm $\|_{L^2}$. Define $\mathcal A_{NT}= \text{span}( \mathcal G_{\DNN}\cup\mathcal G_B -\{g_{NT,\mathcal G}\} )$. Then both  $\mathcal A$  and $\mathcal A_{NT}$ are  Hilbert spaces endowed with the inner product $\langle,\rangle$. Also $\mathcal A_{NT}$ is a finite dimensional space. 
Now define a   sequence of  linear functionals, evaluated at the out-of-sample $x_{j,T}$:
$$\mathcal T_j(h):=  h(x_{j,T}),\quad  j=1,...,N,
\quad h\in\mathcal A\cup\mathcal A_{NT}.
$$

Next we adopt the Riesz representation theorem on $\mathcal A_{NT}$.   Because $\mathcal T_j$ is a linear functional,  it is always bounded  on the finite dimensional Hilbert space $\mathcal A_{NT}$. Hence  the Riesz representation theorem implies that there is a function $m_j^*\in\mathcal A_{NT}$, called Riesz representer, so that 
 $$
\mathcal T_j(h)= \langle h, m_j^*\rangle,\forall h\in\mathcal A_{NT}.
 $$
 Hence $m_j^*$ only depends on the definition of the inner product and $\mathcal A_{NT} $, that is,  
 the distribution of $\{x_{i,t-1}: t\leq T\}$ and the realization $x_{j,T}$, but does not depend on whether DNN or Fourier series were used to construct $\widehat z_{T+1}$.  ($\mathcal A_{NT}$ is the same whether DNN or Fourier series were used.)
 We have two claims:

 \textbf{claim 1: } $\langle m_j^*,g_{NT,\mathcal G}-g\rangle=0$. This is because $g_{NT,\mathcal G}$ is the projection of $g$ on $\mathcal A_{NT}$ under the  norm induced by the inner product, so it should be ``orthogonal" to $m_j^*$ and any element in $\mathcal A_{NT}$. 
  
  \textbf{claim 2: } $\sqrt{T}\max_{j\leq N}|g_{NT,\mathcal G}(x_{j,T})- g(x_{j,T}))|=o_P(1)$, following from  Assumption \ref{ass2}.
 
Then take $m_j=m_j^*$ in Lemma \ref{lem2}, and let $h=\widehat g- g_{NT,\mathcal G}\in\mathcal A_{NT}$,  
\begin{eqnarray*}
\widehat z_{T+1|T} -  z_{T+1|T}  &=&\sum_j w_j[\widehat g(x_{j,T}) -g(x_{j,T}) ] =\sum_j w_j\mathcal T_j(\widehat g-g) \cr
&=&\sum_j w_j\mathcal T_j(\widehat g-g_{NT,\mathcal G})
+\sum_j w_j [ g_{NT,\mathcal G}(x_{j,T}) - g (x_{j,T})]
\cr
&=&\sum_j w_j\langle \widehat g- g_{NT,\mathcal G}, m_j^*\rangle
+ o_P(T^{-1/2}) =  \sum_j w_j\langle \widehat g-g, m_j^*\rangle+ o_P(T^{-1/2}) 
\cr
&=^{(a)}&\frac{1}{NT}\sum_{it}\sum_jw_je_{i,t}m^*_j(x_{i,t-1}) + o_P(T^{-1/2}) \cr
&=^{(b)}& \frac{1}{T}\sum_{t} \frac{1}{N}\sum_i \zeta^*(x_{i,t-1})  \beta_{i,t-1}'v_t+L +o_P(T^{-1/2}) \cr 
L&:=&
\frac{1}{NT}\sum_{it}\sum_jw_ju_{i,t}m^*_j(x_{i,t-1}) =o_P(T^{-1/2}). 
\end{eqnarray*}
 In the above (a) follows from Lemma \ref{lem2}; (b)  is due to  $\zeta^*(\cdot)=\sum_j w_j m_j^*(\cdot)  $
and the definition of $e_{i,t}$.

\end{proof}

\begin{lemma}\label{lem1}

Write      $\|g_1\|_T^2=\frac{1}{NT}\sum_{it}g_1(x_{i,t-1})^2$  and $\|g_1\|_{L^2}^2= \mathbb E \|g_1\|_{T}^2$
 for any  function  $g_1$. 
Then

 $$
\|\widehat g- g\|_{L^2}^2 \leq O_P(\frac{p(\mathcal G) \log (T ) }{T} + \varphi_{\mathcal G}).
$$
\end{lemma}

\begin{proof}
 For notational simplicity, we write $\frac{1}{NT}\sum_{it} eg_1= \frac{1}{NT}\sum_{it} e_{i,t}  g_1(x_{i,t-1})  $ 
 for any  function  $g_1$.  
 Let  $\zeta_T:= \|g-\pi g\|_T^2
 +2\frac{1}{NT}\sum_{it} e  ( g-\pi g )$. 
Then   $Q(\widehat g)\leq Q(\pi g)$
implies  (for $\mathbb E $ taken with respect the distribution of $(e_{i,t},x_{i,t-1})$)
$$
  \|g- \widehat g\|_{L^2}^2 \leq 
\mathbb E  \frac{2}{NT}\sum_{it} e (\widehat g - g )
 + \mathbb E \zeta_T\leq\mathbb E  \frac{2}{NT}\sum_{it} e (\widehat g - g )
+ O_P(\varphi_{\mathcal G}).
$$
Suppose $m_1,...,m_{K(\delta)}$ is a $\delta$-cover  of $\mathcal G$ under the norm $\|m\|_\infty$ and $K(\delta):=\mathcal N(\delta, \mathcal G, \|.\|_\infty)$ as the covering number, where $\delta>0$ is to be determined. Then there is $m_j$    so that $\|m_j- \widehat g\|_\infty<\delta$. Now let   $\mu_j(x)=[m_j(x)- g(x)]/\|m_j-g\|_{L^2}$. 

  Conditioning on $X$, $(v_t, u_{i,t})$ and thus $e_{i,t}$ are independent over time and $e_{i,t}$ is sub-Gaussian. Then 
  $$
  \mathbb E \max_{j\leq K(\delta)}\frac{1}{NT}\sum_{it} e_{i,t} \mu_j(x_{i,t-1})\leq O (\sqrt{\frac{\log K(\delta)}{T}}) =  O (\sqrt{\frac{p(\mathcal G) \log (CT/\delta) }{T}}).
  $$
If  $\|\widehat g- g\|_{L^2}\leq T^{-9}$, then we are done.  If  $\|\widehat g- g\|_{L^2} > T^{-9}$, then 
 set $\delta = T^{-9}$, 
\begin{eqnarray*}
   \mathbb E \frac{1}{NT}\sum_{it} e (\widehat g - g )&\leq& 
   \mathbb E \max_j\frac{1}{NT}\sum_{it} e \mu_j \|m_j-g\|_{L^2} +  \mathbb E \frac{1}{NT}\sum_{it} |e_{i,t}|  \delta \cr
 &\leq &  O(\sqrt{\frac{p(\mathcal G) \log (CT/\delta) }{T}}) [ \delta + \|\widehat g - g\|_{L^2}  ]+ O(\delta)
 \cr
 &\leq&  O(\sqrt{\frac{p(\mathcal G) \log (T ) }{T}})  \|\widehat g - g\|_{L^2}  + O(T^{-9}).
\end{eqnarray*}
Hence $\|g-\widehat g\|_{L^2}^2\leq  O_P(\sqrt{\frac{p(\mathcal G) \log (T ) }{T}})  \|\widehat g - g\|_{L^2}  + O_P(T^{-9} +\varphi_{\mathcal G})$. This implies  the result.

\end{proof}

 \begin{lemma}\label{lem2}
 
Fix   functions $m_j$ so that $\max_{j\leq N}\frac{1}{NT}\sum_{it} m_j(x_{i,t-1})^2= O_P(1)$. Let $\widehat g$ denote the minimizer of $Q(\cdot)$ on the space $\mathcal G$.  
Then uniformly for all $j=1,...,N$,
$$
 \langle  \widehat g-g, m_j\rangle= \frac{1}{NT}\sum_{it}e_{i,t}m_j(x_{i,t-1}) + o_P(T^{-1/2}).
$$

 \end{lemma}

 \begin{proof}  
  Recall $ \psi_{\mathcal G}(h)=  |Q(\pi(h))-Q(h)|$. 
  For any deterministic sequence $\tau_T$, and function $m(x)$, consider a function $h:\mathbb R^{\dim(x)}\rightarrow \mathbb R$ so that
$$
h_j(x):= \widehat g(x) + \tau_T m_j(x).
$$ 

Then, $Q(\widehat g)\leq Q(\pi h_j)$ implies 
$
Q(\widehat g)  - Q(h_j)\leq \psi_T(h_j).
$
Expand this inequality, and note that $y_{i,t}= g(x_{i,t-1}) + e_{i,t}$, 
$$
-\tau_T^2\frac{1}{NT}\sum_{it}m_j(x_{i,t-1})^2+2\tau_T\underbrace{\left[\frac{1}{NT}\sum_{it}(g(x_{i,t-1})-\widehat g(x_{i,t-1}))m_j(x_{i,t-1})+\frac{1}{N}\sum_{it}e_{i,t}m_j(x_{i,t-1})\right]}_{\mathcal M_j}\leq\psi_T(h_j)
$$
  where the left hand side of this inequality equals $Q(\widehat g)- Q(h_j).$ Fix a positive sequence $\epsilon_T>0$ and respectively take value $\tau_T\in\{\epsilon_T, -\epsilon_T\}$. Then the above inequality yields, 
\begin{equation} \label{rateA}
2|\mathcal M_j|\leq \epsilon_T\max_j\frac{1}{NT}\sum_{it} m_j(x_{i,t-1})^2 +\psi_T(h_j)/\epsilon_T.
\end{equation}
We have
	 $
	\psi(h_j)\leq Q(\pi h_j )-Q( h_j )= \|h_j- \pi h_j\|_T^2+\frac{2}{N}\sum_i (y_{i,t}- h_j )( h_j -\pi h_j) \leq O_P(\varphi_{\mathcal G})=o_P(T^{-1}).
 $
This  means $\sqrt{T}\max_j\psi_T(h_j)=o_P(T^{-1/2})$.   Then for $\epsilon_T:=\sqrt{\max_j\psi_T(h_j)}$,
$$
 \sqrt{T}\max_j\psi_T(h_j)=o_P(\epsilon_T),\quad \epsilon_T=o(T^{-1/2}).
$$
For this choice of $\epsilon_T$, the right hand side of (\ref{rateA}) is $o_P(T^{-1/2})$.    This shows 
\begin{equation}\label{eqb.3}
A_j:=\frac{1}{NT}\sum_{it}(\widehat g(x_{i,t-1})-  g(x_{i,t-1}))m_j(x_{i,t-1}) = \frac{1}{NT}\sum_{it}e_{i,t}m_j(x_{i,t-1}) + o_P(T^{-1/2})
\end{equation}
uniformly for $j\leq N$.
 
Next, we show  $A_{j} = \langle  \widehat g-g, m_j\rangle +o_P(T^{-1/2})$   uniformly in $j\leq N$ by proving a  conclusion similar to the stochastic equicontinuity.   Let
\begin{eqnarray*}
G_j(v)&:=&\frac{1}{NT}\sum_{it}(v(x_{i,t-1})-g(x_{i,t-1}))m_j(x_{i,t-1}) -\mathbb E (v(x_{i,t-1})-g(x_{i,t-1}))m_j(x_{i,t-1})\cr
\mathcal C&:=&\{v\in\mathcal G:   \|v\|_{L^2}^2<C \frac{p(\mathcal G)\log T}{T}+C\varphi_N\}.
\end{eqnarray*}
Lemma \ref{lem1} shows   $\widehat g\in\mathcal C$ with   probability arbitrarily close to one. Hence 
$$
|A_{j} -
\langle \widehat g-  g, m_j\rangle | =| G_j(\widehat g)|
\leq  \sup_{v\in\mathcal C}\max_{j\leq N} |G_j(v)|. 
$$

We now show the right hand side is $o_P(T^{-1/2})$  by establishing the stochastic equicontinuity.   Let $p_1,...,p_{K(\delta)}$ be a  $\delta$-cover  of $\mathcal G$ under the norm $\|p\|_\infty$. Let 
$$
h_{kj}(x) := [p_k(x) - g(x)] m_j(x) - \langle p_k-g,m_j\rangle ,\quad \mu_{kj}(x) = h_{kj}(x)/\|h_{kj}\|_{L^2}
.$$
Then $\mathbb E \mu_{kj}(x_{i,t-1})=0$; also $\mu_{kj}(x_{i,t-1})$ are independent across $i$. 
Also let $\mathcal B:= \{k: \exists v\in\mathcal C, \|v-p_k\|_{\infty}<\delta\}$.  Then set $\delta= T^{-1}$, for some $C>0$, 
\begin{eqnarray*}
\max_{k\in\mathcal B, j\leq N}\|h_{kj}\|_{L^2}
&\leq&  C\sup_{k\in\mathcal B, j\leq N}\|p_k-g\|_{L^2}\leq  C \sqrt{\frac{p(\mathcal G)\log T}{T}}+C\sqrt{\varphi_{\mathcal G}} + C\delta \text{ (Cauchy-Schwarz)}\cr
\max_{k\in\mathcal B, j\leq N}|G_j(p_k)|   &=&\max_{k\in\mathcal B, j\leq N}\left|\frac{1}{NT}\sum_{it}h_{kj}(x_{i,t-1}) \right| =\max_{k\in\mathcal B, j\leq N}\left|\frac{1}{NT}\sum_{it}\mu_{kj}(x_{i,t-1}) \right|\|h_{kj}\|_{L^2}\cr
&\leq & O_P(\sqrt{\frac{p(\mathcal G)\log T}{T}}+ \sqrt{\varphi_{\mathcal G}} + \delta)\sqrt{\frac{\log( K(\delta) N)}{T}} \cr \cr 
 \sup_{v\in\mathcal C}\max_j|G_j(v)|&\leq& \max_{p\in\mathcal B, j\leq N}|G_j(p_k)|\cr 
 &&+ \max_{j}\sup_{v\in\mathcal C, \|v-p_k\|_{\infty}<\delta}| \frac{1}{NT}\sum_{it}(p_k(x_{i,t-1})-v(x_{i,t-1}))m_j(x_{i,t-1}) -\langle p_k-v, m_j\rangle|\cr
&\leq & O_P(\sqrt{\frac{p(\mathcal G)\log T}{T}}+ \sqrt{\varphi_{\mathcal G}} + \delta)\sqrt{\frac{\log( K(\delta) N)}{T}}+\delta= o_P(T^{-1/2}).
\end{eqnarray*}
The last equality holds since $\log (K(\delta) N)
= p(\mathcal G)\log (NT) , 
$ and $p(\mathcal G) \log(NT)  = o(T^{1/2})$ and $\varphi_{\mathcal G} p(\mathcal G)\log (NT) = o(1)$. This implies  uniformly in $j\leq N$,
$$
\frac{1}{NT}\sum_{it}(\widehat g(x_{i,t-1})- g(x_{i,t-1}))m_j(x_{i,t-1}) =  \langle  \widehat g-g, m_j\rangle +o_P(T^{-1/2}). 
$$
This finishes the proof.

\end{proof}

\subsubsection{Proof of Theorem \ref{th2}}
\label{proofth2}
Let $J$ be  the number of Fourier series basis and suppose $J^{5/2}=o(T)$ and $(\frac{p(\mathcal M) \log (T ) }{T} + \varphi_{\mathcal M} )J^2\leq o(1)$.

Let $\widehat z_{T+1|T}$ denote the  DNN predictor, and let $\widehat z_{T+1|T, B}$ denote the Fourier series predictor.  Proposition \ref{prop3} show that 
\begin{eqnarray}\label{eqb.4}
\widehat z_{T+1|T} -  z_{T+1|T}  
&=& \frac{1}{T}\sum_{t} \frac{1}{N}\sum_i \zeta^*(x_{i,t-1})  \beta_{i,t-1}'v_t+ o_P(T^{-1/2}) \cr 
\widehat z_{T+1|T,B} -  z_{T+1|T}  
&=& \frac{1}{T}\sum_{t} \frac{1}{N}\sum_i \zeta^*(x_{i,t-1})  \beta_{i,t-1}'v_t+ o_P(T^{-1/2}) . 
\end{eqnarray}
Hence $\widehat z_{T+1|T} -  z_{T+1|T}  =\widehat z_{T+1|T,B} -  z_{T+1|T} + o_P(T^{-1/2}) $, which means the DNN predictor and Fourier series are asymptotically equivalent. 
In addition, it is easy to derive (using the analytical form of $\widehat z_{T+1|T,B} $): 
 $
\widehat z_{T+1|T, B}  - z_{T+1|T} =
\sum_{t=1}^T   H '  \Phi_{t-1}' \beta_{t-1}v_t
+ o_P(T^{-1/2}).
 $
Thus  the DNN predictor satisfies:
 $$
\widehat z_{T+1|T}  - z_{T+1|T} =
\sum_{t=1}^T   H '  \Phi_{t-1}' \beta_{t-1}v_t
+ o_P(T^{-1/2}).
 $$ Define 
 $
\text{SE}(\widehat z_{T+1}):= \sqrt{\sum_{t=1}^T H'\Phi_{t-1}'\beta_{t-1}\Cov(f_t)\beta_{t-1}  '\Phi_{t-1} H}.
 $
We now verify the Lindeberg condition, conditioning on $X$, which suffices to bound the fourth moment.
 Write $ \text{SE}(\widehat z_{T+1})^2:=T^{-1}\sigma^2$ and $F_t:= H '  \Phi_{t-1}' \beta_{t-1}v_t$. Then  $\frac{1}{T}\sum_t \mathbb E(\|TF_t\|^4|X)<C J^{5/2}$, because 
almost surely, 
$$
\frac{1}{T}\sum_t\|TH'\Phi_{t-1}'\beta_{t-1}\|^4\mathbb E(\|v_t\|^4|X)
\leq C\| W'\Phi_T(\frac{1}{TN}\Psi'\Psi)^{-1}\| \frac{1}{TN}\sum_{it}\|\phi(x_{i,t-1})\|^4<CJ^{5/2} 
$$
given $W'\Phi_T<C\sqrt{J}$, $\|(\frac{1}{TN}\Psi'\Psi)^{-1}\|<C$, where $J$ is the number of Fourier series basis.
If $J^{5/2}=o(T)$,
$
\frac{1}{\text{SE}(\widehat z_{T+1})^2}
\sum_t\mathbb E(\|F_t\|^21\{\|F_t\|>\epsilon \text{SE}(\widehat z_{T+1})\}|X)
\leq \frac{1}{T}\frac{1}{T\sigma^4}\sum_t\mathbb E(\|TF_t\|^4|X) \to0.
$ 
Therefore, 
 $$
 \SE(\widehat z_{T+1})^{-1}(\widehat z_{T+1|T}  - z_{T+1|T})\to^d\mathcal N(0,1).
 $$
 It remains to prove $T(\widehat{\text{SE}}(\widehat z_{T+1})^2-\text{SE}(\widehat z_{T+1})^2)= o_P(1)$.

Write $L':=W'\Phi_T(\frac{1}{NT}\Psi'\Psi)^{-1} $ and $s^2= \frac{1}{T}\sum_t L' \frac{1}{N^2}\Phi_{t-1}'  e_t  e_t' \Phi_{t-1} L$. Then $\|L\|= O_P(\sqrt{J})$, and $\frac{1}{NT}\sum_t\|\widehat e_t- e_t\|^2=\frac{1}{NT}\sum_{it}|\widehat g(x_{i,t-1})-g(x_{i,t-1})|^2
=O_P(\frac{p(\mathcal M) \log (T ) }{T} + \varphi_{\mathcal M}).
$
Hence 
$$
T\widehat{\text{SE}}(\widehat z_{T+1})^2= \frac{1}{T}\sum_t L' \frac{1}{N^2}\Phi_{t-1}'\widehat e_t\widehat e_t' \Phi_{t-1} L=s^2+ O_P(\frac{p(\mathcal M) \log (T ) }{T} + \varphi_{\mathcal M} )J^2= s^2+ o_P(1).
$$
Meanwhile, $\frac{1}{T}\sum_t L' \frac{1}{N^2}\Phi_{t-1}'  \beta_{t-1}\Cov(f_t)\beta_{t-1}' \Phi_{t-1} L = T\text{SE}(\widehat z_{T+1})^2$, hence 
$$
s^2=  \frac{1}{T}\sum_t L' \frac{1}{N^2}\Phi_{t-1}'  \mathbb E(e_t  e_t'|\mathcal F_{t-1}) \Phi_{t-1} L +o_P(1)
= T\text{SE}(\widehat z_{T+1})^2+o_P(1).
$$

\subsubsection{Proof of Theorem \ref{th3}: The Bootstrap}\label{proofth3}

For ease of technical proofs, we prove  for  the asymptotic validity of the fully trained bootstrap neural networks. That is, let 
$(y_{i,t}^*)$ denote the bootstrap data.
We prove for the case when the bootstrap DNN $\widehat g^*(\cdot)$ is defined as:
\begin{equation}\label{eq4.3}
\widehat g^*(\cdot) = \arg\min_{g\in\mathcal G_{\DNN}} \sum_{i=1}^N\sum_{t=1}^T(y^*_{i,t}-g(x_{i,t-1})  )^2,
\end{equation}
where $\mathcal G_{\DNN}$ is the pooled-ML space, such as DNN.  

\begin{proof}
Here we focus on the DNN predictor, $\mathcal G=\mathcal M$. Let $\widehat g^*$ be the DNN estimator of $\widehat g$ using the bootstrap data. 
First, similar to Lemma \ref{lem1}, we can establish
 $$
\|\widehat g^*- \widehat g\|_{L^2}^2 \leq O_{P^*}(\frac{p(\mathcal G) \log (T ) }{T} + \varphi_{\mathcal G}).
$$
Hence $\widehat g^*\in\mathcal C$ with   probability arbitrarily close to one. 
Also similar to the proof of (\ref{eqb.3}), we can establish, for $m_j$ being the Riesz representer as in the proof of Proposition \ref{prop3},
 \begin{equation}\label{eqb.41}
A_j^*:=\frac{1}{NT}\sum_{it}(\widehat g^*(x_{i,t-1})-  \widehat g(x_{i,t-1}))m_j(x_{i,t-1}) = \frac{1}{NT}\sum_{it}\widehat e^*_{i,t}m_j(x_{i,t-1}) + o_{P^*}(T^{-1/2}).
\end{equation}
Next, we show   $A^*_{j} = \langle  \widehat g-g, m_j\rangle +o_{P^*}(T^{-1/2})$   uniformly in $j\leq N$. 
The proof is slightly different from that in Lemma 
\ref{lem2}. 
Let
\begin{eqnarray*}
G_j(v_1,v_2)&:=&\frac{1}{NT}\sum_{it}(v_1(x_{i,t-1})-v_2(x_{i,t-1}))m_j(x_{i,t-1}) -\mathbb E (v_1(x_{i,t-1})-v_2(x_{i,t-1}))m_j(x_{i,t-1}).
\end{eqnarray*}
Hence 
 \begin{equation}\label{eqb.5}
|A_{j}^* -
\langle \widehat g^*-  \widehat g, m_j\rangle | =| G_j(\widehat g^*,\widehat g)|
\leq  \sup_{v_1, v_2\in\mathcal C}\max_{j\leq N} |G_j(v_1, v_2)|. 
\end{equation}
    Let $p_1,...,p_{K(\delta)}$ be a  $\delta$-cover  of $\mathcal G$ under the norm $\|p\|_\infty$. Let 
$$
h_{kd,j}(x) := [p_k(x) - p_d(x)] m_j(x) - \langle p_k-p_d,m_j\rangle ,\quad \mu_{kd,j}(x) = h_{kd,j}(x)/\|h_{kd,j}\|_{L^2}
.$$
Then $\mathbb E \mu_{kd,j}(x_{i,t-1})=0$; also $\mu_{kd,j}(x_{i,t-1})$ are independent across $i$. 
Also let $\mathcal B:= \{k: \exists v\in\mathcal C, \|v-p_k\|_{\infty}<\delta\}$.  Then set $\delta= T^{-1}$, for some $C>0$, 
\begin{eqnarray*}
\max_{k,d\in\mathcal B, j\leq N}\|h_{kd,j}\|_{L^2}
&\leq&  C\sup_{k,d\in\mathcal B}\|p_k-p_d\|_{L^2}\leq  C \sqrt{\frac{p(\mathcal G)\log T}{T}}+C\sqrt{\varphi_{\mathcal G}} + C\delta  \cr
\max_{k,d\in\mathcal B, j\leq N}|G_j(p_k, p_d)|   &=&\max_{k,d\in\mathcal B, j\leq N}\left|\frac{1}{NT}\sum_{it}h_{kd,j}(x_{i,t-1}) \right| =\max_{k,d\in\mathcal B, j\leq N}\left|\frac{1}{NT}\sum_{it}\mu_{kd,j}(x_{i,t-1}) \right|\|h_{kd,j}\|_{L^2}\cr
&\leq & O_P(\sqrt{\frac{p(\mathcal G)\log T}{T}}+ \sqrt{\varphi_{\mathcal G}} + \delta)\sqrt{\frac{\log( K(\delta) N)}{T}} \cr \cr 
 \sup_{v_1,v_2\in\mathcal C}\max_j|G_j(v_1,v_2)|&\leq& \max_{k,d\in\mathcal B, j\leq N}|G_j(p_k,p_d)|+ O_P(\delta) \cr
&\leq & O_P(\sqrt{\frac{p(\mathcal G)\log T}{T}}+ \sqrt{\varphi_{\mathcal G}} + \delta)\sqrt{\frac{\log( K(\delta) N)}{T}}+\delta= o_P(T^{-1/2}).
\end{eqnarray*}
Combining  with (\ref{eqb.41})(\ref{eqb.5}),  for $\widehat e_{i,t}^*= \widehat e_{i,t}\eta_t^*$,
$$
  \langle  \widehat g^*-\widehat g, m_j\rangle  =\frac{1}{NT}\sum_{it}\widehat e^*_{i,t}m_j(x_{i,t-1}) + o_{P^*}(T^{-1/2}) . 
$$
Next, for the same functional $\mathcal T_j$ and the Riesz representer $m_j=m_j^*$ in the proof of Proposition \ref{prop3} and $\zeta^*=\sum w_j m_j^*$, 
\begin{eqnarray*}
\widehat z^*_{T+1|T} -  \widehat z_{T+1|T}  &=&\sum_j w_j[\widehat g^*(x_{j,T}) -\widehat g(x_{j,T}) ] =\sum_j w_j\mathcal T_j(\widehat g^*-\widehat g)
=\sum_j w_j\langle \widehat g^*- \widehat g, m_j^*\rangle\cr
&=&\frac{1}{NT}\sum_{it}\sum_jw_j\widehat e_{i,t}\eta^*_tm^*_j(x_{i,t-1}) + o_{P^*}(T^{-1/2}) \cr
&=&\frac{1}{NT}\sum_{it} \widehat e_{i,t}\eta^*_t\zeta^*(x_{i,t-1}) + o_{P^*}(T^{-1/2}) .
\end{eqnarray*}
Next, $\mathbb E\eta_t^*=0$ and $\Var^*(\eta_t^*)=1$. Let $\widetilde\zeta_{t-1}$ be the $N$-vector of $\zeta^*(x_{i,t-1})$. Then
\begin{eqnarray*}
\widehat{\SE}_{inf}^{*2}&:=&\frac{1}{T^2}\sum_t (\frac{1}{N} \widehat e_{t}'\widetilde\zeta_{t-1})^2  =  \frac{1}{T^2}\sum_t (\frac{1}{N}  e_{t}'\widetilde\zeta_{t-1})^2  +o_P(T^{-1})\cr
&=&
\frac{1}{T^2}\sum_t (\frac{1}{N} \widetilde \zeta'_{t-1}\beta_{t-1}v_t)^2+o_P(T^{-1})={\SE}_{inf}^{2} +o_P(T^{-1}) \cr
{\SE}_{inf}^{2} &:=&
 \frac{1}{T^2}\sum_t \frac{1}{N^2} \widetilde \zeta'_{t-1}\beta_{t-1}\Cov(v_t)\beta_{t-1}'\widetilde\zeta_{t-1}.
\end{eqnarray*}
  This shows
$c_{inf}:=\SE_{inf}^{-1}\widehat{\SE}_{inf}^{*}\to^P 1 $ because $\sqrt{T}\SE_{inf}>c>0$.  
Meanwhile, 
\begin{eqnarray*}
 \SE_{inf}^{-1}(\widehat z^*_{T+1|T} -  \widehat z_{T+1|T} )&=&c_{inf}
\widehat{\SE}_{inf}^{*-1}(\widehat z^*_{T+1|T} -  \widehat z_{T+1|T} )\cr
&=&\left[\frac{1}{T^2}\sum_t (\frac{1}{N} \widehat e_{t}'\widetilde\zeta_{t-1})^2\right]^{-1/2}\frac{1}{NT}\sum_{it} \widehat e_{i,t}\eta^*_t\zeta^*(x_{i,t-1}) + o_{P^*}(1)\cr
&\to^{d^*}&\mathcal N(0,1).
\end{eqnarray*}
In addition, (\ref{eqb.4}) shows 
$
 \SE_{inf}^{-1}(\widehat z_{T+1|T} -   z_{T+1|T} )\to^d \mathcal N(0,1).
$
 Let $q^*_{\alpha}$ be the bootstrap critical value of $ |\widehat z^*_{T+1|T} -  \widehat z_{T+1|T} |$ so that $P^*(| \widehat z^*_{T+1|T} -  \widehat z_{T+1|T} |\leq q^*_{\alpha}) =1-\alpha$, where 
  $P^*$ denotes the probability measure with respect to the bootstrap distribution (i.e., the   distribution of $\{\eta_t^*,t\leq T\}$ conditional on the data). 
 Then
\begin{eqnarray*}
 P(|\widehat z_{T+1|T} -   z_{T+1|T}|\leq q^*_{\alpha}) &=&P(| \SE_{inf}^{-1}(\widehat z_{T+1|T} -   z_{T+1|T} |\leq \SE^{-1}_{inf}q^*_{\alpha}) \cr
 &=&P^*(| \SE_{inf}^{-1}(\widehat z^*_{T+1|T} -  \widehat z_{T+1|T} )|\leq \SE^{-1}_{inf}q^*_{\alpha}) +o_P(1)\cr
 &=&1-\alpha +o_P(1).
\end{eqnarray*}
 
\end{proof}

\subsubsection{Proof of Theorem \ref{th5} }\label{proofth4}

\begin{proof}

 The constraint MV problem is $\max_{\omega}\min_{\mu} F(\omega,\mu)$, subject to $|\mu-\widehat z_{Z+1|T}|\leq q$, where 
$
F(\omega,\mu)= \bomega'\mu  -\frac{\gamma}{2} \bomega'\Sigma_T\bomega.  
$ 
 For the inner  problem, $\min_{\mu_i}\omega_i \mu_i$ subject to  $|\mu_i-\widehat z_{T+1|T,i}|\leq q_{\alpha, i}$, the minimum is 
$\omega_i (\widehat z_{T+1|T,i}-\sgn(\omega_i)q_{\alpha, i})= \omega_i\widehat z_{T+1|T,i} -|\omega_i|q_{\alpha,i}$.  This leads to the designated result.

In the one dimensional case, the problem is equivalent to 
$$
\min_{\omega}   \frac{\gamma}{2} \omega^2\sigma^2 - \omega\widehat z_{T+1|T}   + q|\omega|
$$
 which is also equivalent to $\min_{\omega}   \frac{\gamma\sigma^2}{2} [\omega -\omega^{\MV}]^2  + q|\omega|
$. The solution is well known to be the soft thresholding operator given in the theorem.

\end{proof}

\subsection{Proof  of (\ref{eq5.3}), (\ref{eq5.6}) and  (\ref{eq4.11})}\label{secb.4}

\textit{Proof of Expression (\ref{eq5.3})}. 

\begin{proof}
For the case $R=2$, $\min_{\bomega=(\omega, 1-\omega)} F(\omega)$, where

$
 F(\omega)=  \frac{\gamma}{2}\bomega' \Sigma_T\bomega - \omega \widehat z_{T+1|T, 1} 
- (1-\omega) \widehat z_{T+1|T, 2} 
+ q_{\alpha,1}|\omega|+ q_{\alpha,2}|1-\omega|.
$ This function is piecewise quadratic on the regions $(-\infty, 0], [0,1], $ and $[1,\infty)$.
 We define and verify that 
\begin{eqnarray*}
\arg\min_{\omega>1} F(\omega) &=&\max\{A,1\}, \quad \text{where} \quad A= \omega_1^{\MV} -c_0(q_{\alpha,1} + q_{\alpha,2}) \cr
\arg\min_{\omega<0} F(\omega) &=&\min\{C,0\}, \quad \text{where}\quad C=\omega_1^{\MV} +c_0(q_{\alpha,1} + q_{\alpha,2}) \cr 
\arg\min_{0<\omega<1} F(\omega) &=&1\{0<B<1\}B+1\{B>1\}
\quad \text{where}\quad B= \omega_1^{\MV} -c_0(q_{\alpha,1} - q_{\alpha,2}) \cr
\omega_1^{\MV}&=&c_0\left[\widehat   z_{T+1|T,1} -\widehat  z_{T+1|T,2} -\gamma (\Cov(z_{T+1,1},z_{T+1,2}) -\Var(z_{T+1,2}))\right] 
\end{eqnarray*}
where $c_0= \gamma^{-1}\Var(z_{T+1,1}- z_{T+1,2})^{-1}$.
Note that $C>B>A$. Also write $\omega_2^{\MV} =1-\omega_1^{\MV}$. Let the optimal solution be 
$\omega_1^*$, and $\omega_2^*=1-\omega_1^*$. 
Hence we have the following cases:

Case 1. $A>1$, which is $\omega_2^{\MV}<- c_0(q_{\alpha,1} + q_{\alpha,2}) $, then 
$\omega_1^{*} = A$, and 

$\omega_2^*= \omega_2^{\MV}+c_0(q_{\alpha,1}+q_{\alpha,2})$.

Case 2. $B>1>A$, which is $-c_0(q_{\alpha,1} + q_{\alpha,2})<\omega_2^{\MV}<c_0(q_{\alpha,2} - q_{\alpha,1})$, then $\omega_1^{*}=1$.

Case 3. $0<B<1$, which is $\omega_1^{\MV}>c_0(q_{\alpha,1} - q_{\alpha,2})$ and $\omega_2^{\MV}>c_0(q_{\alpha,2} - q_{\alpha,1})$, then $\omega_1^*=B$, and $\omega_2^*= 1-B= \omega_2^{\MV} -c_0(q_{\alpha,2} - q_{\alpha,1})$.

Case 4. $C>0>B$, which is $-c_0(q_{\alpha,1} + q_{\alpha,2})<\omega_1^{\MV}<c_0(q_{\alpha,1} - q_{\alpha,2})$, then $\omega_1^*=0$.

Case 5. $C<0$, which is $\omega_1^{\MV}<-c_0(q_{\alpha,1} + q_{\alpha,2}) $, then $\omega_1^*=C$.

Hence for $k, j\in\{1,2\}$ and $k\neq j$,

When $\omega_1^{\MV}<-c_0(q_{\alpha,1} + q_{\alpha,2}) $, then $\omega_k^*=\omega_1^{\MV} +c_0(q_{\alpha,1} + q_{\alpha,2})$.

When $-c_0(q_{\alpha,1} + q_{\alpha,2})<\omega_1^{\MV}<c_0(q_{\alpha,k} - q_{\alpha,j})$, then $\omega_k^*=0$.

When  $\omega_1^{\MV}>c_0(q_{\alpha,1} - q_{\alpha,2})$ and $\omega_2^{\MV}>c_0(q_{\alpha,2} - q_{\alpha,1})$, then  $\omega_k^*=  \omega_1^{\MV} -c_0(q_{\alpha,k} - q_{\alpha,j})$.

\end{proof}

\textit{Proof of Expressions (\ref{eq5.6}), (\ref{eq4.11})}
\begin{proof}
 
 (\ref{eq5.6})  is a special case of (\ref{eq4.11}). Hence we directly prove the latter. 
 The conditional  distribution of $z_{T+1}|\widehat z_{T+1|T}$ is normal. Also,   $(z_{T+1|T}, \widehat z_{T+1|T})$ is jointly normal, we have 
  $$\mathbb E(  z_{T+1}|\widehat z_{T+1|T})
  =\mathbb E[\mathbb E(z_{T+1}|\widehat z_{T+1|T}, z_{T+1|T})|\widehat z_{T+1|T}]
  =\mathbb E(z_{ T+1|T}|\widehat z_{T+1|T})=
(I-W_1)\widehat z_{T+1|T} + W_1\pi
  $$
  where
$
W_1= \SE^2(\SE^2+v)^{-1}. 
$
Also,  
$$
\begin{pmatrix}
    z_{T+1|T}\\
    \widehat z_{T+1|T}
\end{pmatrix}\sim \mathcal N\left(
\begin{pmatrix}
\pi \\
\pi 
\end{pmatrix}, \begin{pmatrix}
v & v \\
v& \SE^2 + v 
\end{pmatrix}
\right).
$$
Hence $\Var(  z_{T+1|T}|\widehat z_{T+1|T})=
v-v(\SE^2+v)^{-1}v= vW_1'= W_1v,
  $ and thus 
\begin{eqnarray*}
\widetilde\Sigma&=&\Var( z_{T+1}|\widehat z_{T+1|T})
=\mathbb E [\Var( z_{T+1}|  z_{T+1|T})| \widehat z_{T+1|T}]
+ \Var(  z_{T+1|T}|\widehat z_{T+1|T}) 
\cr 
&=&\Sigma_T 
+ \Var(  z_{T+1|T}|\widehat z_{T+1|T}) 
=\Sigma_T +(I-W_1)\SE^2= \Sigma_T+ vW_1'.
\end{eqnarray*}
 Note that for two multivariate normal densities with the same covariance $\widetilde\Sigma$,  the Kullback–Leibler divergence is 
$$
D(h_T||f_T)= \frac{1}{2}(\mu-\widetilde z_{T+1})'\widetilde\Sigma^{-1}(\mu-\widetilde z_{T+1}).
$$
Hence 
$
\bomega^{RS}:=\arg\max_{\bomega\in\mathbb R^N}\min_{\mu\in\mathbb R^N}L(\bomega, \mu)
$
 where
 $$
 L(\bomega, \mu) = \bomega '\mu   -\frac{\gamma}{2}\bomega' \widetilde\Sigma\bomega +   \frac{1}{\tau   } D(h_T||f_T).
 $$
 We have $\min_{\mu\in\mathbb R^N}L(\bomega, \mu)=\bomega '\widetilde z_{T+1}    -\frac{\gamma+\tau}{2}\bomega' \widetilde\Sigma\bomega $. Maximizing it yields 
 $$
 \bomega^{RS}= \widetilde\Sigma^{-1}\widetilde z_{T+1}   \frac{1}{\tau+\gamma}. 
 $$

Recall that 
$ \widetilde z_{T+1}= (1-W_1)\widehat z_{T+1|T} + W_1\pi$, $\bomega^{\MV} =  \Sigma_T^{-1}\widehat z_{T+1|T}   \frac{1}{\gamma}$ and 

$\bomega^{RS}_{\pi}  =  \Sigma_T^{-1}\pi   \frac{1}{ \gamma}$. Also,   for any constant $c$, if $v= c\Sigma_T $, then    $W_1\Sigma_T=\Sigma_T W_1'$, both sides equal  $\Sigma_T-c\Sigma_T(\SE^2+c\Sigma_T)^{-1}\Sigma_T$.
Hence
$$
 \bomega^{RS}= \frac{\gamma }{\tau+\gamma}\widetilde\Sigma^{-1}\Sigma_T\left[ (I-W_1)' \bomega^{\MV} +W_1'\bomega^{\MV}_{\pi}\right]   .  
 $$

 In the one-dimensional case, this implifies to 
 $$
 \omega^{RS} = \frac{\gamma }{\tau+\gamma} g(W_1),\quad 
 g(W_1):=
 \frac{\sigma^2_T}{\sigma_T^2+ vW_1}\left[ (1-W_1) \omega^{\MV} +W_1\omega^{\MV}_{\pi}\right]  .
 $$
 In addition, $W_1= \SE^2/(\SE^2+v)$, which is monotonically increasing in $\SE$. Hence the monotonicity of $\omega^{RS}$ with respect to $\SE$ is determined by the monotonicity of $g(\cdot)$. It can be rewritten as 
 $$
 g(W_1)= c +\frac{\sigma_T^2}{v} \frac{v+\sigma_T^2}{\sigma_T^2+vW_1} [\omega^{\MV} - \frac{\sigma_T^2}{v+\sigma_T^2}\omega_{\pi}^{\MV} ].
 $$
 Hence $g(\cdot)$ is increasing if and only if $\omega^{\MV} < \frac{\sigma_T^2}{v+\sigma_T^2}\omega_{\pi}^{\MV}$.
\end{proof}

\newpage
\end{document}